\documentclass[opre, nonblindrev]{informs3}

\OneAndAHalfSpacedXI 

\usepackage{endnotes}
\let\footnote=\endnote

\usepackage{amsmath}
\usepackage{booktabs}
\allowdisplaybreaks
\usepackage{dsfont, pifont, mathrsfs}
\usepackage{multirow}
\usepackage{setspace}
\usepackage[hidelinks, colorlinks=true, citecolor=blue]{hyperref}
\usepackage[ruled,vlined,linesnumbered]{algorithm2e}

\renewenvironment{proof}[1][Proof]{%
  \par\noindent\textit{#1. }\ignorespaces
}{%
  \hfill$\square$\par
}

\usepackage{natbib}
 \bibpunct[, ]{(}{)}{,}{a}{}{,}%
 \def\bibfont{\small}%
\TheoremsNumberedThrough     
\ECRepeatTheorems

\EquationsNumberedThrough    

\usepackage{macro}

\begin{document}

\RUNAUTHOR{Wang et al. (2026)}
\RUNTITLE{Optimal Sample Size Calculation in Cost-Effectiveness Longitudinal Cluster Randomized Trials}

\TITLE{Optimal Sample Size Calculation in Cost-Effectiveness Longitudinal Cluster Randomized Trials}

\ARTICLEAUTHORS{%
\AUTHOR{Hao Wang\textsuperscript{$1, 2$} \quad\quad Jingxia Liu\textsuperscript{$3, 4$} \quad\quad Drew B. Cameron\textsuperscript{$5$} \quad\quad Jiaqi Tong\textsuperscript{$1, 2$} \quad\quad \\Donna Spiegelman\textsuperscript{$1, 2$} \quad\quad Daniella Meeker\textsuperscript{6} \quad\quad Fan Li\textsuperscript{$\ast, 1, 2$}}
\AFF{\textsuperscript{1}Department of Biostatistics, Yale School of Public Health, New Haven, CT, USA\texorpdfstring{\\}{}%
\textsuperscript{2}Center for Methods in Implementation and Prevention Science, Yale School of Public Health, New Haven, CT, USA\texorpdfstring{\\}{}%
\textsuperscript{3}Division of Public Health Sciences, Department of Surgery, Washington University School of Medicine, St Louis, MO, USA\texorpdfstring{\\}{}%
\textsuperscript{4}Institute for Informatics, Data Science \& Biostatistics, Washington University School of Medicine, St Louis, MO, USA\texorpdfstring{\\}{}%
\textsuperscript{5}Department of Health Policy and Management, Yale School of Public Health, New Haven, CT, USA\texorpdfstring{\\}{}%
\textsuperscript{6}Department of Biomedical Informatics and Data Science, Yale School of Medicine, New Haven, CT, USA\texorpdfstring{\\}{}%
\EMAIL{$^\ast$\texttt{fan.f.li@yale.edu}}}
} 

\ABSTRACT{%
Longitudinal cluster randomized trials (L-CRTs) are increasingly used to evaluate the cost-effectiveness of healthcare interventions across multiple assessment periods, yet design methods for powering these trials remain underdeveloped. Existing methods for cost-effectiveness analyses in cluster settings are limited to simple parallel-arm cluster randomized trials with a single follow-up assessment period. These methods cannot accommodate the complex correlation structures in L-CRTs conducted over multiple periods, which require differentiation between within-period and between-period correlations for both clinical and cost outcomes, as well as between-outcome correlations. Moreover, while substantial methodological advances have been made for the design of L-CRTs with univariate outcomes, none specifically address cost-effectiveness objectives where clinical and cost outcomes must be jointly modeled. We provide a design-stage framework for powering cost-effectiveness L-CRTs across three design variants: parallel-arm, crossover, and stepped wedge designs. We derive closed-form variance expressions for the generalized least squares estimator of the average incremental net monetary benefit under a bivariate linear mixed model. We propose a standardized ceiling ratio that adjusts willingness-to-pay for relative outcome variability to inform optimal design. We then develop local optimal designs that maximize statistical power under known correlation parameters and MaxiMin designs that ensure robust performance across parameter uncertainty for all three design variants. Through a real stepped wedge trial data example, we demonstrate the sample size calculation for testing intervention cost-effectiveness under local optimal and MaxiMin designs.
}%


\maketitle

\section{Introduction}\label{sec:introduction}

Cluster randomized trials (CRTs) are increasingly used for evaluating the cost-effectiveness of complex interventions in public health and healthcare delivery research \citep{Gomes2012review}. Unlike individually randomized trials, CRTs randomize intact clusters (e.g., providers, clinics, or hospitals) to intervention conditions \citep{Turner2017design}, and have become one of the most common study designs for embedded pragmatic clinical trials \citep{Cook2016, Weinfurt2017}. When studying cost-effectiveness objectives, CRTs require specialized methods to appropriately account for both clinical outcomes and cost measurements that are correlated within clusters \citep{Turner2017analysis, Gomes2012developing}. While cost-effectiveness analysis is typically framed as a decision-analytic approach for informing resource allocation under uncertainty \citep{Sanders2016}, clinical trial design and sample size calculation require a formal hypothesis testing framework. For this reason, we focus our discussion of cost-effectiveness in terms of the incremental net monetary benefit (INMB), following the approach outlined in \citet{Manju2014}.

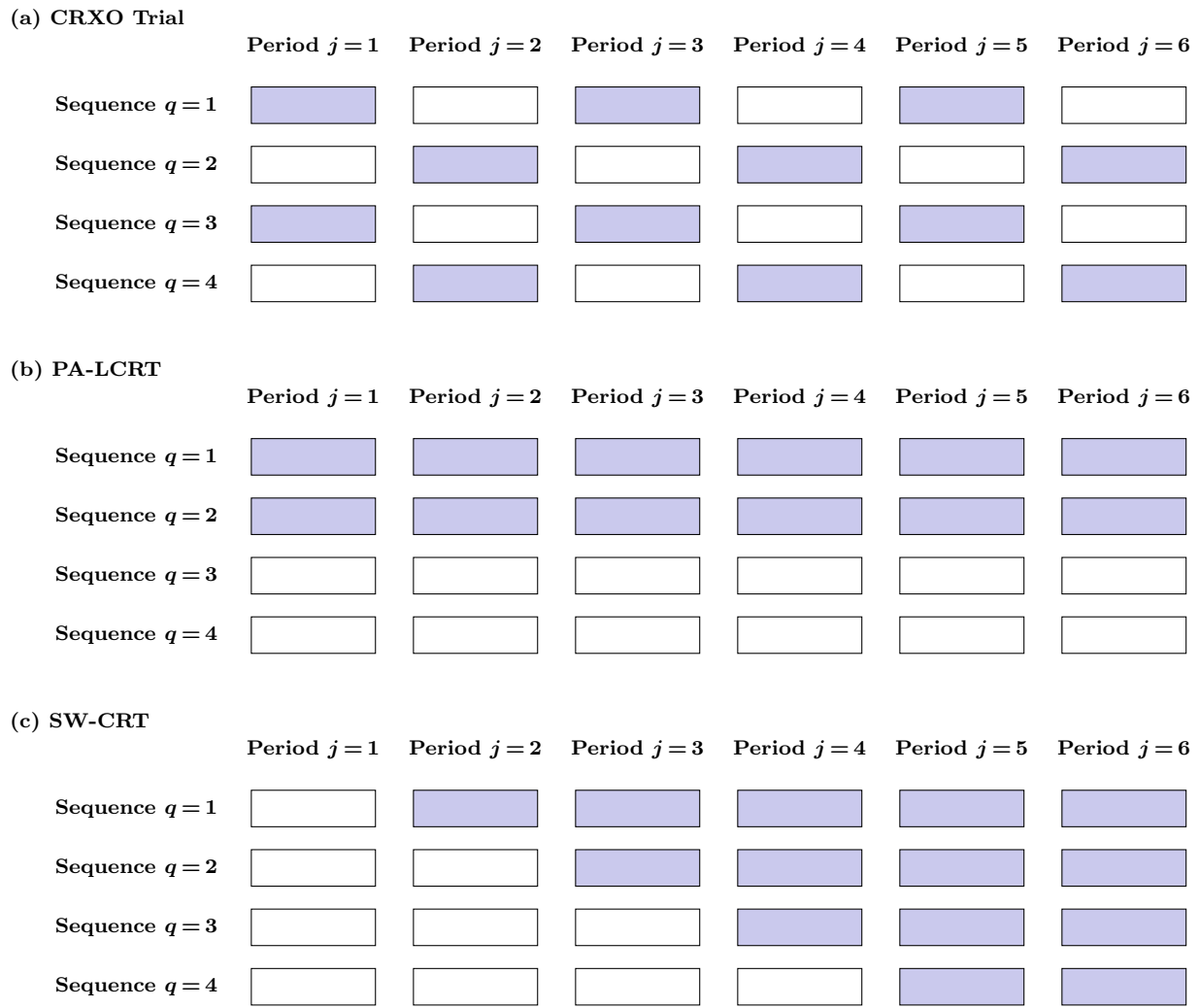
\begin{figure}[t]
    \centering
    \scriptsize
    \definecolor{controlcolor}{RGB}{255, 255, 255}
    \definecolor{treatmentcolor}{RGB}{202, 201, 237}

    \begin{tikzpicture}
        \node[font=\bfseries, anchor=west] at (-8.5, 0) {(a) CRXO Trial};
        
        \matrix (crxo) [matrix of nodes, nodes in empty cells,
            nodes={draw, minimum width=1.7cm, minimum height=0.5cm, anchor=center},
            column sep=2mm, row sep=3mm,
            column 1/.style={nodes={draw=none}},
            row 1/.style={nodes={draw=none}},
            ] at (0, -2) {
                & \textbf{Period }$\bm{j=1}$ & \textbf{Period }$\bm{j=2}$ & \textbf{Period }$\bm{j=3}$ & \textbf{Period }$\bm{j=4}$ & \textbf{Period }$\bm{j=5}$ & \textbf{Period }$\bm{j=6}$ \\
                \textbf{Sequence }$\bm{q=1}$ & |[fill=treatmentcolor]| & |[fill=controlcolor]| & |[fill=treatmentcolor]| & |[fill=controlcolor]| & |[fill=treatmentcolor]| & |[fill=controlcolor]| \\
                \textbf{Sequence }$\bm{q=2}$ & |[fill=controlcolor]| & |[fill=treatmentcolor]| & |[fill=controlcolor]| & |[fill=treatmentcolor]| & |[fill=controlcolor]| & |[fill=treatmentcolor]| \\
                \textbf{Sequence }$\bm{q=3}$ & |[fill=treatmentcolor]| & |[fill=controlcolor]| & |[fill=treatmentcolor]| & |[fill=controlcolor]| & |[fill=treatmentcolor]| & |[fill=controlcolor]| \\
                \textbf{Sequence }$\bm{q=4}$ & |[fill=controlcolor]| & |[fill=treatmentcolor]| & |[fill=controlcolor]| & |[fill=treatmentcolor]| & |[fill=controlcolor]| & |[fill=treatmentcolor]| \\
            };
    \end{tikzpicture}
    
    \vspace{0.6cm}

    \begin{tikzpicture}
        \node[font=\bfseries, anchor=west] at (-8.5, 0) {(b) PA-LCRT};
        
        \matrix (pa) [matrix of nodes, nodes in empty cells,
            nodes={draw, minimum width=1.7cm, minimum height=0.5cm, anchor=center},
            column sep=2mm, row sep=3mm,
            column 1/.style={nodes={draw=none}},
            row 1/.style={nodes={draw=none}},
            ] at (0, -2) {
                & \textbf{Period }$\bm{j=1}$ & \textbf{Period }$\bm{j=2}$ & \textbf{Period }$\bm{j=3}$ & \textbf{Period }$\bm{j=4}$ & \textbf{Period }$\bm{j=5}$ & \textbf{Period }$\bm{j=6}$ \\
                \textbf{Sequence }$\bm{q=1}$ & |[fill=treatmentcolor]| & |[fill=treatmentcolor]| & |[fill=treatmentcolor]| & |[fill=treatmentcolor]| & |[fill=treatmentcolor]| & |[fill=treatmentcolor]| \\
                \textbf{Sequence }$\bm{q=2}$ & |[fill=treatmentcolor]| & |[fill=treatmentcolor]| & |[fill=treatmentcolor]| & |[fill=treatmentcolor]| & |[fill=treatmentcolor]| & |[fill=treatmentcolor]| \\
                \textbf{Sequence }$\bm{q=3}$ & |[fill=controlcolor]| & |[fill=controlcolor]| & |[fill=controlcolor]| & |[fill=controlcolor]| & |[fill=controlcolor]| & |[fill=controlcolor]| \\
                \textbf{Sequence }$\bm{q=4}$ & |[fill=controlcolor]| & |[fill=controlcolor]| & |[fill=controlcolor]| & |[fill=controlcolor]| & |[fill=controlcolor]| & |[fill=controlcolor]| \\
            };
    \end{tikzpicture}
    
    \vspace{0.6cm}

    \begin{tikzpicture}
        \node[font=\bfseries, anchor=west] at (-8.5, 0) {(c) SW-CRT};
        
        \matrix (sw) [matrix of nodes, nodes in empty cells,
            nodes={draw, minimum width=1.7cm, minimum height=0.5cm, anchor=center},
            column sep=2mm, row sep=3mm,
            column 1/.style={nodes={draw=none}},
            row 1/.style={nodes={draw=none}},
            ] at (0, -2) {
                & \textbf{Period }$\bm{j=1}$ & \textbf{Period }$\bm{j=2}$ & \textbf{Period }$\bm{j=3}$ & \textbf{Period }$\bm{j=4}$ & \textbf{Period }$\bm{j=5}$ & \textbf{Period }$\bm{j=6}$ \\
                \textbf{Sequence }$\bm{q=1}$ & |[fill=controlcolor]| & |[fill=treatmentcolor]| & |[fill=treatmentcolor]| & |[fill=treatmentcolor]| & |[fill=treatmentcolor]| & |[fill=treatmentcolor]| \\
                \textbf{Sequence }$\bm{q=2}$ & |[fill=controlcolor]| & |[fill=controlcolor]| & |[fill=treatmentcolor]| & |[fill=treatmentcolor]| & |[fill=treatmentcolor]| & |[fill=treatmentcolor]| \\
                \textbf{Sequence }$\bm{q=3}$ & |[fill=controlcolor]| & |[fill=controlcolor]| & |[fill=controlcolor]| & |[fill=treatmentcolor]| & |[fill=treatmentcolor]| & |[fill=treatmentcolor]| \\
                \textbf{Sequence }$\bm{q=4}$ & |[fill=controlcolor]| & |[fill=controlcolor]| & |[fill=controlcolor]| & |[fill=controlcolor]| & |[fill=treatmentcolor]| & |[fill=treatmentcolor]| \\
            };
    \end{tikzpicture}
    \caption{Schematic illustration of L-CRT designs with $J = 6$ periods. Purple cells indicate intervention periods and white cells indicate control periods. Panel (a) shows a CRXO trial, where sequences $q \in \{1, 3\}$ receive intervention in periods $j \in \{1, 3, 5\}$ and control in periods $j \in \{2, 4, 6\}$, while sequences $q \in \{2, 4\}$ receive control in periods $j \in \{1, 3, 5\}$ and intervention in periods $j \in \{2, 4, 6\}$. Panel (b) displays a PA-LCRT, where sequences $q \in \{1, 2\}$ are randomized to intervention and sequences $q \in \{3, 4\}$ to control for all periods $j \in \{1, \ldots, 6\}$. Panel (c) presents an SW-CRT, where sequence $q$ transitions from control to intervention in period $j = q + 1$, with all clusters under control in period $j = 1$ and all under intervention in period $j \in \{5, 6\}$.}
    \label{fig:designs}
\end{figure}

With healthcare systems facing increasing pressure to deliver value-based care, there is growing interest in conducting longitudinal CRTs (L-CRTs) to evaluate the sustained economic impact of interventions over time \citep{Lung2021}. L-CRTs encompass several design variants, including parallel-arm L-CRTs (PA-LCRTs) with multiple assessment periods, cluster randomized crossover (CRXO) trials, and stepped wedge cluster randomized trials (SW-CRTs); see Figure~\ref{fig:designs} for a schematic illustration of these three designs. These designs hold strong promise for understanding the clinical value and value for money of interventions in real-world, pragmatic settings \citep{Weinfurt2017, Hemming2020}, because they enable assessment of how interventions affect both clinical outcomes and healthcare costs over multiple follow-up periods. For instance, \citet{Kinchin2018} used an SW-CRT to evaluate the efficacy and cost-effectiveness of a community-based care model for older patients with complex needs; \citet{Di2024} conducted an SW-CRT across multiple hospital units to assess a care-bundle intervention for preventing falls among older inpatients, with cost-effectiveness as a primary outcome. Cost-effectiveness evaluations in L-CRTs have also been planned or conducted across a range of clinical areas, including oncology and palliative care \citep{Luckett2018}, cardiovascular disease prevention \citep{Nieuwenhuijse2023}, cardiac surgery \citep{Rigal2019}, and nursing communication \citep{Liu2023}.

Despite the increasing application of cost-effectiveness analyses in L-CRTs, statistical methods for designing these trials lag significantly behind. On one hand, existing methods for cost-effectiveness analyses in cluster settings have only been developed primarily for the simplest PA-CRTs with a single follow-up assessment period \citep{Gomes2012review, Gomes2012developing, Bachmann2007}. \citet{Manju2014} first derived optimal sample size formulas for cost-effectiveness PA-CRTs that maximize power under a given budget, requiring specification of intracluster correlation coefficients (ICCs) for both effects and costs. \citet{Li2023experimental, Li2024sample} subsequently extended these methods to incorporate covariates and developed power analysis formulas for three-level PA-CRTs. However, these methods are strictly limited to non-longitudinal CRTs, do not address the secular trend \citep{Li2021overview, Ouyang2022, Ouyang2023}, nor do they account for the special randomization schemes inherent to different L-CRT designs in Figure~\ref{fig:designs}. Furthermore, the aforementioned sample size methods for CRTs have not differentiated between a within-period ICC (i.e., the correlation between measurements from two individuals in the same cluster-period) from a between-period ICC (i.e., the correlation between measurements from two individuals in the same cluster but different periods) for both clinical and cost outcomes, an important feature that characterize a typical design procedure for L-CRTs.

On the other hand, although a flourishing line of literature has explored study design methods for L-CRTs, none of these developments specifically address cost-effectiveness objectives. For CRXOs, \citet{Li2019} developed sample size calculations for two-treatment, two-period designs that account for both within-period and between-period correlations via generalized estimating equations. For PA-LCRTs, \citet{Wang2021} proposed flexible closed-form sample size formulas based on generalized estimating equations that accommodate arbitrary correlation structures, missing data patterns, and randomly varying cluster sizes. In the context of SW-CRTs, substantial methodological advances have been made to accommodate different correlation structures; see \citet{Li2021overview} for a comprehensive review. For example, \citet{Davis-Plourde2021} extended these methods to SW-CRTs with subclusters under a block exchangeable correlation structure. Additional work has examined practical design considerations, including the efficiency impact of unequal cluster-period sizes \citep{Tian2022}, power calculations for incomplete designs \citep{Zhang2023}, and optimal design algorithms that minimize cost for a desired level of power or maximize power given fixed budgets \citep{Liu2024, Liu2025}. However, none of these methods have been generalized to address cost-effectiveness objectives, where clinical and cost outcomes must be jointly modeled. Recently, \citet{Davis-Plourde2023} developed power analyses for multivariate outcomes in SW-CRTs, demonstrating efficiency gains of the multivariate linear mixed model over the univariate approach. This motivates the joint modeling of clinical and cost outcomes through their bivariate distribution in L-CRTs to address cost-effectiveness objectives. From a practical standpoint, joint modeling is essential because cost outcomes are subject to distinct sources of within-cluster and temporal variability that may induce correlation patterns differing substantially from those of clinical outcomes. For example, facilities vary in their underlying cost structures, input prices and reimbursement policies shift over time, and learning effects alter resource use as clinicians gain experience with an intervention \citep{Gomes2012review}.

In this article, we develop a unified framework for the design of cost-effectiveness L-CRTs that addresses the aforementioned methodological gaps. Specifically, we first introduce a bivariate linear mixed model that simultaneously captures clinical and cost outcomes, accommodating both within-period and between-period correlations for each outcome as well as three types of between-outcome correlations. Under this model, we then derive new, closed-form variance expressions for the generalized least squares estimator of the average INMB across three L-CRT design variants: CRXO trials, PA-LCRTs, and SW-CRTs. To account for the heterogeneous scales of clinical and cost outcomes, we propose a \textit{standardized ceiling ratio} that adjusts willingness-to-pay by the relative variability of the two outcomes. Additionally, we develop local optimal designs (LODs) that maximize statistical power under known correlation parameters, where closed-form solutions are available for CRXO and PA-LCRT designs in Sections~\ref{sec:crxo} and~\ref{sec:pa}. Because correlation parameters are rarely known exactly at the design stage, we further develop algorithms for constructing MaxiMin designs (MMDs) that maximize the minimum achievable power over a prespecified parameter space, ensuring optimal protection against the worst case parameter configuration. Finally, we illustrate the application of our methods through a real SW-CRT data example.

The remainder of this article is structured as follows. In Section~\ref{sec:set_up}, we introduce the bivariate linear mixed model and derive new variance expressions for the average INMB estimator. In Section~\ref{sec:designs}, we formulate the design optimization problem under budget constraints and present algorithms for computing LODs and MMDs. We then apply the framework to CRXO trials, PA-LCRTs, and SW-CRTs in Sections~\ref{sec:crxo}--\ref{sec:swcrt}, respectively, providing closed-form LOD solutions where available along with numerical results characterizing optimal designs under varying correlation structures. In Section~\ref{sec:da}, we illustrate the proposed methods using the context of a real SW-CRT. In Section~\ref{sec:discussion}, we conclude with some practical recommendations and areas for future research.
\section{Model Formulations} \label{sec:set_up}

We focus on a cross-sectional and complete L-CRT with $I$ clusters (indexed by $i \in \{1, \ldots, I\}$), $J$ periods (indexed by $j \in \{1, \ldots, J\}$), and $K$ individuals per cluster-period (indexed by $k \in \{1, \ldots, K\}$), where different individuals are recruited in each period for any given cluster \citep{Copas2015}. This cross-sectional design is the most common in practice; a systematic review of 160 SW-CRTs published from January 2016 to March 2022 found that 76.3\% adopted a cross-sectional recruitment scheme \citep{Nevins2024}. We let $E_{ijk}$ and $C_{ijk}$ denote the clinical outcome (i.e., effect) and cost outcome, respectively, for individual $k$ from cluster $i$ enrolled in period $j$, and define $Z_{ij}$ as the treatment indicator for cluster $i$ in period $j$, with $Z_{ij} = 1$ indicating assignment to treatment and $Z_{ij} = 0$ indicating assignment to control. To jointly model the effect and cost data, we consider the following bivariate linear mixed model \citep{Davis-Plourde2023}:
\begin{align}\label{eqn:model}
     E_{ijk} = \alpha_{0j} + \alpha_1Z_{ij} + b_i^E + s_{ij}^E + \epsilon_{ijk}^E, \quad C_{ijk} = \gamma_{0j} + \gamma_1Z_{ij} + b_i^C + s_{ij}^C + \epsilon_{ijk}^C,
\end{align}
where $\bm{b}_i = (b_i^E, b_i^C)^\prime \sim \calN(\bm{0}_{2 \times 1}, \bfSigma_b)$ is a bivariate cluster-level random effect, $\bm{s}_{ij} = (s_{ij}^E, s_{ij}^C)^\prime \sim \calN(\bm{0}_{2 \times 1}, \bfSigma_s)$ is a bivariate cluster-period-level random effect, and $\bm{\epsilon}_{ijk} = (\epsilon_{ijk}^E, \epsilon_{ijk}^C)^\prime \sim \calN(\bm{0}_{2 \times 1}, \bfSigma_e)$ is the individual-level error term, with
\begin{align*}
    \bfSigma_b = \left(\begin{array}{cc}
        \sigma_{bE}^2 & \sigma_{bEC}\\
        \sigma_{bEC} & \sigma_{bC}^2
    \end{array}\right), 
    \quad \bfSigma_s = \left(\begin{array}{cc}
        \sigma_{sE}^2 & \sigma_{sEC}\\
        \sigma_{sEC} & \sigma_{sC}^2
    \end{array}\right), 
    \quad \bfSigma_e = \left(\begin{array}{cc}
        \sigma_{\epsilon E}^2 & \sigma_{\epsilon EC}\\
        \sigma_{\epsilon EC} & \sigma_{\epsilon C}^2
    \end{array}\right).
\end{align*}
In Model~\eqref{eqn:model}, $\alpha_{0j}$ and $\gamma_{0j}$ are fixed period effects that capture secular trends, $\alpha_1$ and $\gamma_1$ are the treatment effect parameters for the clinical outcome and cost, respectively, $\bm{b}_i$ are cluster-level random effects, $\bm{s}_{ij}$ are cluster-by-period random effects, and $\bm{\epsilon}_{ijk}$ are individual-level residual error terms. We assume that random effects at different levels are mutually independent (i.e., $\bm{b}_i \perp \bm{s}_{i'j'} \perp \bm{\epsilon}_{i''j''k''} \, \forall \,i, i', i'', j', j'', k''$) while allowing for correlation within each bivariate random effect vector \citep{Li2021overview}. Under this formulation, the total variances for the clinical outcome and cost are $\sigma_E^2 = \sigma_{bE}^2 + \sigma_{sE}^2 + \sigma_{\epsilon E}^2$ and $\sigma_C^2 = \sigma_{bC}^2 + \sigma_{sC}^2 + \sigma_{\epsilon C}^2$, respectively. We note that cost data are often right-skewed and heavy-tailed in practice \citep{Gomes2012review}; however, the bivariate linear mixed model provides a tractable working model for design-stage variance approximation, and robustness to distributional assumptions can be explored through sensitivity analyses or appropriate transformations at the analysis stage.

\begin{table}[htbp]
    \centering
    \caption{Definition of ICCs, where the total variances for the clinical outcome and cost are $\sigma_E^2 = \sigma_{bE}^2 + \sigma_{sE}^2 + \sigma_{\epsilon E}^2$ and $\sigma_C^2 = \sigma_{bC}^2 + \sigma_{sC}^2 + \sigma_{\epsilon C}^2$, respectively.}
    \resizebox{\linewidth}{!}{
    \begin{tabular}{llll}
        \toprule
        Type & ICC & Definition & Expression \\
        \midrule
        Outcome-specific & $\rho_0^E$ & Within-period effect ICC & $\Corr(E_{ijk},E_{ijk^\prime}) = (\sigma_{bE}^2 + \sigma_{sE}^2)/\sigma_E^2$ \\
        & $\rho_0^C$ & Within-period cost ICC & $\Corr(C_{ijk},C_{ijk^\prime}) = (\sigma_{bC}^2 + \sigma_{sC}^2)/\sigma_C^2$ \\
        & $\rho_1^E$ & Between-period effect ICC & $\Corr(E_{ijk},E_{ij^\prime k^\prime}) = \sigma_{bE}^2/\sigma_E^2$ \\
        & $\rho_1^C$ & Between-period cost ICC & $\Corr(C_{ijk},C_{ij^\prime k^\prime}) = \sigma_{bC}^2/\sigma_C^2$ \\
        \midrule
        Between-outcome & $\rho_0^{EC}$ & Within-period effect-cost ICC & $\Corr(E_{ijk},C_{ijk^\prime}) = (\sigma_{bEC} + \sigma_{sEC})/(\sigma_E\sigma_C)$ \\
        & $\rho_1^{EC}$ & Between-period effect-cost ICC & $\Corr(E_{ijk^\prime},C_{ij^\prime k^\prime}) = \sigma_{bEC}/(\sigma_E\sigma_C)$ \\
        & $\rho_2^{EC}$ & Within-individual effect-cost ICC & $\Corr(E_{ijk},C_{ijk}) = (\sigma_{bEC} + \sigma_{sEC} + \sigma_{\epsilon EC})/(\sigma_E\sigma_C)$ \\
        \bottomrule
    \end{tabular}}
    \label{tab:ICC}
\end{table}

Model~\eqref{eqn:model} induces four outcome-specific ICCs and three between-outcome ICCs, where the following definitions hold for $j \neq j'$ and $k \neq k'$. The outcome-specific ICCs include: (a) within-period effect ICC, $\rho_0^E \coloneqq \Corr(E_{ijk}, E_{ijk'}) = (\sigma_{bE}^2 + \sigma_{sE}^2)/\sigma_E^2$; (b) within-period cost ICC, $\rho_0^C \coloneqq \Corr(C_{ijk}, C_{ijk'}) = (\sigma_{bC}^2 + \sigma_{sC}^2)/\sigma_C^2$; (c) between-period effect ICC, $\rho_1^E \coloneqq \Corr(E_{ijk}, E_{ij'k'}) = \sigma_{bE}^2/\sigma_E^2$; and (d) between-period cost ICC, $\rho_1^C \coloneqq \Corr(C_{ijk}, C_{ij'k'}) = \sigma_{bC}^2/\sigma_C^2$. The between-outcome ICCs include: (e) within-period effect-cost ICC, $\rho_0^{EC} \coloneqq \Corr(E_{ijk}, C_{ijk'}) = (\sigma_{bEC} + \sigma_{sEC})/(\sigma_E\sigma_C)$; (f) between-period effect-cost ICC, $\rho_1^{EC} \coloneqq \Corr(E_{ijk}, C_{ij'k'}) = \sigma_{bEC}/(\sigma_E\sigma_C)$; and (g) within-individual effect-cost ICC, $\rho_2^{EC} \coloneqq \Corr(E_{ijk}, C_{ijk}) = (\sigma_{bEC} + \sigma_{sEC} + \sigma_{\epsilon EC})/(\sigma_E\sigma_C)$. Figure~\ref{fig:ICC} illustrates the data structure within a cost-effectiveness L-CRT along with these ICCs, and formal definitions are provided in Table~\ref{tab:ICC}.

\begin{figure}[t]
    \centering
    \includegraphics[width=0.55\linewidth]{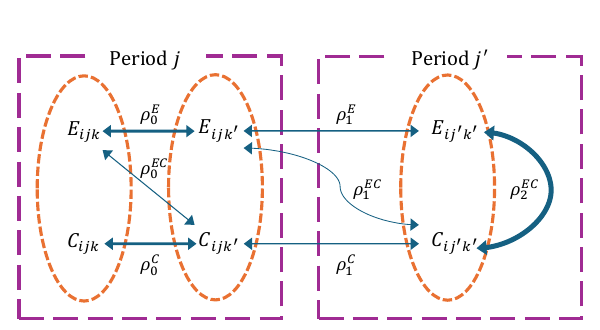}
    \caption{Graphical representation of the data structure within the $i$-th cluster of a cost-effectiveness L-CRT, where $j \neq j'$ and $k \neq k'$. Each dashed oval represents an individual observed in a period (dashed square) nested in the cluster, with both a clinical outcome ($E$) and cost ($C$) measured per individual. Arrows depict the ICCs corresponding to within-period correlations ($\rho_0^E$, $\rho_0^C$), between-period correlations ($\rho_1^E$, $\rho_1^C$), and three types of between-outcome correlations ($\rho_0^{EC}$, $\rho_1^{EC}$, $\rho_2^{EC}$). The width of each arrow indicates the relative strength of correlation under the following constraints: (i) $\rho_1^E \leq \rho_0^E$; (ii) $\rho_1^C \leq \rho_0^C$; (iii) $\rho_0^{EC} \leq \min(\rho_0^E, \rho_0^C)$; (iv) $\rho_1^{EC} \leq \min(\rho_1^E, \rho_1^C)$; and (v) $\rho_1^{EC} \leq \rho_0^{EC} \leq \rho_2^{EC}$.}
    \label{fig:ICC}
\end{figure}

Based on the bivariate linear mixed model, the ICC parameters are subject to the following constraints: (i) $\rho_1^E \leq \rho_0^E$; (ii) $\rho_1^C \leq \rho_0^C$; (iii) $\rho_0^{EC} \leq \min(\rho_0^E, \rho_0^C)$; (iv) $\rho_1^{EC} \leq \min(\rho_1^E, \rho_1^C)$; and (v) $\rho_1^{EC} \leq \rho_0^{EC} \leq \rho_2^{EC}$. Constraints (i) and (ii) ensure that between-period ICCs do not exceed their within-period counterparts, generalizing the standard nested exchangeable assumption in SW-CRTs with univariate outcomes to the bivariate cost-effectiveness setting \citep{Hooper2016}. Constraints (iii) and (iv) require that between-outcome correlations do not exceed the corresponding outcome-specific correlations, while constraint (v) indicates that the between-outcome correlation is strongest within the same individual, intermediate for different individuals in the same cluster-period, and weakest for individuals across different periods. When prior information for specifying ICCs is limited, investigators can set each between-period ICC as a fixed proportion of its corresponding within-period ICC; this proportion is referred to as the cluster autocorrelation coefficient (CAC) \citep{Hooper2016}. Specifically, we define $r^E = \rho_1^E / \rho_0^E$, $r^C = \rho_1^C / \rho_0^C$, and $r^{EC} = \rho_1^{EC} / \rho_0^{EC}$ as the CACs for the effect, cost, and effect-cost ICCs, respectively.

While constraints (i)--(v) are necessary for a valid ICC specification, they are not sufficient to guarantee a positive definite correlation matrix. When fitting Model~\eqref{eqn:model} directly, positive definiteness is easily ensured by requiring $\bfSigma_b$, $\bfSigma_s$, and $\bfSigma_e$ to be positive definite. However, investigators may prefer to specify the seven ICC parameters directly when planning trials \citep{Ouyang2023}, in which case additional conditions are needed. To address this, we derive closed-form expressions for the eigenvalues of the individual-level correlation matrix, providing full conditions for positive definiteness under the ICC parameterization. These conditions allow for efficient verification of valid parameter combinations during the design optimization algorithms introduced in Section~\ref{sec:designs}. Additionally, they are relevant when investigators choose to analyze trial data using generalized estimating equations (GEE) with a working correlation structure specified through ICCs \citep{Liang1986}, as positive definiteness of the working correlation matrix is required for valid GEE estimation \citep{Li2018}.

To characterize the correlation structure induced by Model~\eqref{eqn:model}, for each cluster $i \in \{1, \ldots, I\}$, we stack the bivariate outcomes as $\mathbf{Y}_i = (E_{i11}, C_{i11}, \ldots, E_{i1K}, C_{i1K}, \ldots, E_{iJK}, C_{iJK})^\prime$, a vector of dimension $2JK \times 1$, and let $\mathbf{R}_i$ denote the corresponding $2JK \times 2JK$ correlation matrix. Under the cross-sectional design, correlations between observations in the same cluster are determined by three components: (i) the within-individual correlation $\bfGamma_2$ for the same individual; (ii) the within-period correlation $\bfGamma_0$ for different individuals in the same period; and (iii) the between-period correlation $\bfGamma_1$ for individuals in different periods, where
\begin{align*}
    \bfGamma_0 = \left(\begin{array}{cc}
        \rho_0^E & \rho_0^{EC} \\
        \rho_0^{EC} & \rho_0^C
    \end{array}\right), \quad
    \bfGamma_1 = \left(\begin{array}{cc}
        \rho_1^E & \rho_1^{EC} \\
        \rho_1^{EC} & \rho_1^C
    \end{array}\right), \quad
    \bfGamma_2 = \left(\begin{array}{cc}
        1 & \rho_2^{EC} \\
        \rho_2^{EC} & 1
    \end{array}\right).
\end{align*}
Using these ICC matrices, the correlation matrix $\mathbf{R}_i$ can be expressed as
\begin{align*} 
    \mathbf{R}_i = \mathbf{I}_J \otimes \left\{\mathbf{I}_K\otimes \bfGamma_2 + (\mathbf{1}_K\mathbf{1}_K^\prime - \mathbf{I}_K)\otimes \bfGamma_0\right\} + (\mathbf{1}_J\mathbf{1}_J^\prime - \mathbf{I}_J)\otimes (\mathbf{1}_K\mathbf{1}_K^\prime \otimes \bfGamma_1),
\end{align*} 
where $\mathbf{I}_J$ and $\mathbf{I}_K$ are identity matrices of dimension $J$ and $K$, respectively, $\mathbf{1}_J$ and $\mathbf{1}_K$ are vectors of ones of length $J$ and $K$, respectively, and $\otimes$ denotes the Kronecker product. The eigenvalues of $\mathbf{R}_i$ are provided in Theorem~\ref{thm:eigenvalues}, with the proof given in Appendix~\ref{supp_sec:eigen}.

\begin{theorem} \label{thm:eigenvalues}
    Under Model~\eqref{eqn:model}, the eigenvalues of $\mathbf{R}_i$ are $\{\lambda_1^+, \lambda_1^-, \lambda_2^+, \lambda_2^-, \lambda_3^+, \lambda_3^-\}$ with explicit forms:
    \begin{align*}
        \lambda_1^\pm &= \frac{1}{2}\{2 + (K-1)(\rho_0^E + \rho_0^C) + (J-1)K(\rho_1^E + \rho_1^C)\} \pm \frac{1}{2}\sqrt{\xi_1}, \displaybreak[0]\\
        \lambda_2^\pm &= \frac{1}{2}(\kappa^E + \kappa^C) \pm \frac{1}{2}\sqrt{(\kappa^E - \kappa^C)^2 + 4(\kappa^{EC})^2}, \displaybreak[0]\\
        \lambda_3^\pm &= \frac{1}{2}(2 - \rho_0^E - \rho_0^C) \pm \frac{1}{2}\sqrt{(\rho_0^E - \rho_0^C)^2 + 4(\rho_2^{EC} - \rho_0^{EC})^2},
    \end{align*}
    where $\xi_1 = \{(K-1)(\rho_0^E - \rho_0^C) + (J-1)K(\rho_1^E - \rho_1^C)\}^2 + 4\{\rho_2^{EC} + (K-1)\rho_0^{EC} + (J-1)K\rho_1^{EC}\}^2$, $\kappa^E = 1 + (K - 1)\rho_0^E - K\rho_1^E$, $\kappa^C = 1 + (K - 1)\rho_0^C - K\rho_1^C$, and $\kappa^{EC} = \rho_2^{EC} + (K - 1)\rho_0^{EC} - K\rho_1^{EC}$. The eigenvalues $\lambda_1^\pm$ each have multiplicity $1$, $\lambda_2^\pm$ each have multiplicity $J-1$, and $\lambda_3^\pm$ each have multiplicity $J(K-1)$.
\end{theorem}

The eigenvalues in Theorem~\ref{thm:eigenvalues} generalize the univariate nested exchangeable structure to the bivariate setting. To see this connection, consider the special case where the cost ICCs and between-outcome ICCs are all zero (i.e., $\rho_0^C = \rho_1^C = 0$ and $\rho_0^{EC} = \rho_1^{EC} = \rho_2^{EC} = 0$). Under these conditions, the eigenvalues simplify to $\lambda_1^+ = 1 + (K-1)\rho_0^E + (J-1)K\rho_1^E$, $\lambda_1^- = 1$, $\lambda_2^+ = 1 + (K-1)\rho_0^E - K\rho_1^E$, $\lambda_2^- = 1$, $\lambda_3^+ = 1$, and $\lambda_3^- = 1 - \rho_0^E$. Here, $\lambda_1^+$, $\lambda_2^+$, and $\lambda_3^-$ with multiplicities $1$, $J-1$, and $J(K-1)$, respectively, correspond exactly to the eigenvalues of the univariate nested exchangeable correlation structure derived in \citet{Li2018}, while the remaining eigenvalue of $1$ with total multiplicity $JK$ corresponds to the uncorrelated cost component. Moving forward, positive definiteness of $\mathbf{R}_i$ requires all six eigenvalues to be strictly positive, a condition we impose during the design optimization in Section~\ref{sec:designs} to ensure that all candidate ICC configurations yield a valid correlation structure.

For cost-effectiveness analysis objectives, we follow the framework of \citet{Manju2014}. and focus on the net monetary benefit (NMB) measure, which expresses both clinical outcomes and costs on a common monetary scale using the ceiling ratio $\lambda$ \citep{Stinnett1998}, defined as the maximum amount society is willing to pay per unit of clinical benefit. To proceed, Model~\eqref{eqn:model} induces a linear mixed model for NMB given as
\begin{align}\label{eqn:nmb}
    \text{NMB}_{ijk} = \lambda E_{ijk} - C_{ijk} = \beta_{0j} + \beta_1 Z_{ij} + (\lambda b_i^E - b_i^C) + (\lambda s_{ij}^E - s_{ij}^C) + (\lambda\epsilon_{ijk}^E-\epsilon_{ijk}^C),
\end{align}
where $\beta_{0j} = \lambda\alpha_{0j} - \gamma_{0j}$ represents the average NMB under the control condition in period $j$, and $\beta_1 = \lambda\alpha_1 - \gamma_1$ is the average INMB for the intervention.  Under this formulation, the intervention is considered cost-effective if $\beta_1 > 0$ for a given ceiling ratio $\lambda$. Of note, when $\lambda = 0$, Equation~\eqref{eqn:nmb} reduces to the negative cost model (i.e., the second equation in Model~\eqref{eqn:model}), and as $\lambda \to \infty$, it approximates the $\lambda$-scaled clinical outcome model (i.e., the first equation in Model~\eqref{eqn:model}).

From Model~\eqref{eqn:model}, the NMB random components are normally distributed with mean zero and variances $\bbV(\lambda b_i^E - b_i^C) = \lambda^2\sigma_{bE}^2 + \sigma_{bC}^2 - 2\lambda\sigma_{bEC}$ for the cluster-level component, $\bbV(\lambda s_{ij}^E - s_{ij}^C) = \lambda^2\sigma_{sE}^2 + \sigma_{sC}^2 - 2\lambda\sigma_{sEC}$ for the cluster-period-level component, and $\bbV(\lambda\epsilon_{ijk}^E - \epsilon_{ijk}^C) = \lambda^2\sigma_{\epsilon E}^2 + \sigma_{\epsilon C}^2 - 2\lambda\sigma_{\epsilon EC}$ for the individual-level component. The random components are mutually independent across levels, with $(\lambda b_i^E - b_i^C)$ independent across clusters, $(\lambda s_{ij}^E - s_{ij}^C)$ independent across cluster-periods, and $(\lambda\epsilon_{ijk}^E - \epsilon_{ijk}^C)$ independent across individuals. The variance expression for the generalized least squares estimator of $\beta_1$ is provided in Theorem~\ref{thm:variance}, with the proof given in Appendix~\ref{supp_sec:var}.

\begin{theorem} \label{thm:variance}
    Define $U = \sumi\sumj Z_{ij}$, $V = \sumi(\sumj Z_{ij})^2$, $W=\sumj(\sumi Z_{ij})^2$ as study design constants that depend on the treatment status, the variance of average INMB estimator $\hat\beta_1$ is
    \begin{align*}
        \bbV(\hat\beta_1) = \frac{IJ\sigma_E\sigma_C}{\eta}\left\{\frac{\varphi}{K}\left(\frac{\lambda^2\kappa^E\sigma_E}{\sigma_C}-2\lambda\kappa^{EC} + \frac{\kappa^C\sigma_C}{\sigma_E}\right) - \frac{J}{\Delta^\ast}(U^2-IV)\left(\frac{\lambda^2\rho_1^E\sigma_E}{\sigma_C} - 2\lambda\rho_1^{EC} + \frac{\rho_1^C\sigma_C}{\sigma_E}\right)\right\},
    \end{align*} 
    where 
    \begin{align*}
    &\eta = \varphi J(IU - W) + J^2{\Delta^\ast}^{-1}(U^2 - IV)\sigma_E^2\sigma_C^2(\rho_1^E\rho_1^C - {\rho_1^{EC}}^2)\{\varphi + {\Delta^\ast}^{-1}(U^2 - IV)\}\\
    &\varphi = J(IU - W)\Delta^{-1} + (U^2 - IV)(\Delta^{-1} - {\Delta^\ast}^{-1})\\
    &\Delta^\ast = \Delta + J\sigma_E^2\sigma_C^2\{K^{-1}(\kappa^E\rho_1^C + \kappa^C\rho_1^E - 2\kappa^{EC}\rho_1^{EC}) + J(\rho_1^E\rho_1^C - {\rho_1^{EC}}^2)\}\\
    &\Delta = K^{-2}\sigma_E^2\sigma_C^2(\kappa^E\kappa^C - {\kappa^{EC}}^2),
    \end{align*}
    and $\kappa^E = 1 + (K - 1)\rho_0^E - K\rho_1^E$, $\kappa^C = 1 + (K - 1)\rho_0^C - K\rho_1^C$ are eigenvalues of a nested exchangeable correlation structure for effect and cost, respectively, and $\kappa^{EC} = \rho_2^{EC} + (K - 1)\rho_0^{EC} - K\rho_1^{EC}$ carries the impact of three between-outcome ICCs.
\end{theorem}

Theorem~\ref{thm:variance} indicates that the variance of the average INMB estimator $\hat\beta_1$ is independent of the true treatment effect size $\alpha_1$. Instead, the variance depends only on the trial resources (i.e., number of periods, clusters, and individuals), the study design constants (i.e., treatment assignment), and the ICC-related parameters defined in Table~\ref{tab:ICC}, all of which are either pre-specified during the trial planning stage or approximated based on prior literature. Of note, under a simple exchangeable correlation structure \citep{Hussey2007} where the cluster-period-level random effects $s_{ij}^E$ and $s_{ij}^C$ are absent from Model~\eqref{eqn:model}, the variance expression in Theorem~\ref{thm:variance} can be simplified by setting $\rho_1^E = \rho_0^E$, $\rho_1^C = \rho_0^C$, and $\rho_1^{EC} = \rho_0^{EC}$, leading to eigenvalues $\kappa^E = 1 - \rho_0^E$ and $\kappa^C = 1 - \rho_0^C$ for the effect and cost, respectively.
\section{Local and MaxiMin Optimal Designs} \label{sec:designs}

\begin{algorithm}[t]
    \DontPrintSemicolon
    \SetAlgoLined
    \SetKwInOut{Input}{Input}
    \SetKwInOut{Output}{Output}
    \SetKwFunction{ComputeVar}{ComputeVariance}
    
    \caption{LOD for L-CRTs} \label{alg:LOD}
    
    \Input{$B$, $c_1$, $c_2$, $I_{\max}$, $J$, $K_{\max}$, $\bm{\rho}$,  $\beta_1$}
    \Output{LOD $(I^*, K^*)$ with maximum power $\calP^*$}
    
    \BlankLine
    $\calP^* \leftarrow 0$; \, $(I^*, K^*) \leftarrow (\texttt{null}, \texttt{null})$\;
    
    \BlankLine
    \For{$I \leftarrow 2$ \KwTo $I_{\max}$}{
        \For{$K \leftarrow 2$ \KwTo $K_{\max}$}{
            \If{$I(c_1 + c_2 JK) \leq B$}{
                $\bbV(\hat{\beta}_1) \leftarrow$ \ComputeVar{$I, K, J, \bm{\rho}$}\;
                \If{$\bbV(\hat{\beta}_1) > 0$}{
                    $\calP \leftarrow \Phi\!\left(\dfrac{|\beta_1|}{\sqrt{\bbV(\hat{\beta}_1)}} - z_{\alpha/2}\right)$
                    \If{$\calP > \calP^*$}{
                        $\calP^* \leftarrow \calP$; \, $(I^*, K^*) \leftarrow (I, K)$\;
                    }
                }
            }
        }
    }
    
    \BlankLine
    \Return $(I^*, K^*, \calP^*)$\;
\end{algorithm}

We formulate the design optimization problem for cost-effectiveness L-CRTs under a linear budget constraint given as
\begin{align*} 
    B = I(c_1 + c_2JK), 
\end{align*} 
where $c_1$ is the cost per cluster and $c_2$ is the cost per individual. We let $\calD \coloneqq \{(I, K): I \in \{2, \ldots, I_{\max}\}, K \in \{2, \ldots, K_{\max}\}, I(c_1 + c_2JK) \leq B\}$ denote the feasible design space with $J$ pre-specified periods. Our objective is to identify the optimal sample size combination $(I, K)$ that minimizes the variance of $\hat\beta_1$ such that
\begin{align} 
    (I, K) = \arg\min_{(I,K) \in \calD} \bbV(\hat\beta_1), \label{eqn:opt} 
\end{align} 
which equivalently maximizes power of the optimal design for detecting the average INMB.

When the ICC parameter vector $\bm{\rho} \coloneqq (\rho^E_0, \rho^E_1, \rho^C_0, \rho^C_1, \rho^{EC}_0, \rho^{EC}_1, \rho^{EC}_2)^\prime$ is known, the LOD solves~\eqref{eqn:opt} through exhaustive enumeration over $\calD$ \citep{Liu2024}. Specifically, Algorithm \ref{alg:LOD} evaluates $\bbV(\hat\beta_1)$ for each feasible $(I, K)$ pair and returns the integer-valued sample size combination that achieves minimum variance. For CRXO and PA-LCRTs, the structure of the design constants simplifies the variance expression in Theorem~\ref{thm:variance}, yielding closed-form solutions for the decimal-valued LOD $(I^{\text{dec}}, K^{\text{dec}})$ that serve as theoretical benchmarks (Sections~\ref{sec:crxo} and~\ref{sec:pa}). For SW-CRTs, however, the complexity of the design matrix precludes closed-form solutions and instead requires numerical optimization.

\begin{algorithm}[ht!]
    \DontPrintSemicolon
    \SetAlgoLined
    \SetKwInOut{Input}{Input}
    \SetKwInOut{Output}{Output}
    \SetKwFunction{RelEff}{RelativeEfficiency}
    \SetKwFunction{ComputeTheta}{ComputeTheta}
    \SetKwFunction{ComputeVar}{ComputeVariance}
    
    \caption{MMD for L-CRTs} \label{alg:MMD}
    
    \Input{$B$, $c_1$, $c_2$, $I_{\max}$, $J$, $K_{\max}$, $\bm{\Theta} = \{\bm{\rho}: \bm{\rho}_{\min} \leq \bm{\rho} \leq \bm{\rho}_{\max}\}$, $\beta_1$}
    \Output{MMD $(I^*, K^*)$ with maximum worst-case relative efficiency $\text{RE}^*_{\min}$}
    
    \BlankLine
    $\text{RE}_{\text{MaxiMin}} \leftarrow 0$; \, $(I^*, K^*) \leftarrow (\texttt{null}, \texttt{null})$\;
    
    \BlankLine
    \For{$I \leftarrow 2$ \KwTo $I_{\max}$}{
        \For{$K \leftarrow 2$ \KwTo $K_{\max}$}{
            \If{$I(c_1 + c_2 JK) \leq B$}{
                $f(\bm{\rho}) \leftarrow$ \RelEff{$\bm{\rho}, I, K, J$}\;
                $\bm{\rho}_{\text{worst}} \leftarrow \displaystyle\arg\min_{\bm{\rho}} f(\bm{\rho})$\;
                \Indp
                \textbf{s.t.} \quad $\bm{\rho}_{\min} \leq \bm{\rho} \leq \bm{\rho}_{\max}$\;
                \phantom{\textbf{s.t.}} \quad $\calC_{\text{ineq}}(\bm{\rho}, K, J)$\;
                \Indm
                \BlankLine
                $\text{RE}_{\min} \leftarrow f(\bm{\rho}_{\text{worst}})$\;
                \If{$0 < \normalfont\text{RE}_{\min} \leq 1$ \normalfont\textbf{and} $\text{RE}_{\min} > \text{RE}_{\text{MaxiMin}}$}{
                    $\text{RE}_{\text{MaxiMin}} \leftarrow \text{RE}_{\min}$\;
                    $(I^*, K^*) \leftarrow (I, K)$\;
                }
            }
        }
    }
    
    \BlankLine
    \Return $(I^*, K^*, \text{RE}_{\text{MaxiMin}})$\;
    
    \BlankLine
    \SetKwProg{Fn}{Function}{:}{}
    \Fn{\RelEff{$\bm{\rho}, I, K, J$}}{
        $\vartheta \leftarrow$ \ComputeTheta{$\bm{\rho}$}\;
        \lIf{$\vartheta \leq 0$}{\Return $\infty$}
        \BlankLine
        $K_{\text{LOD}} \leftarrow \sqrt{c_1 \vartheta / (c_2 J)}$\;
        $I_{\text{LOD}} \leftarrow B / (c_1 + \sqrt{\vartheta c_1 c_2 J})$\;
        $\bbV_{\Dec}(\hat{\beta}_1) \leftarrow$ \ComputeVar{$I_{\text{LOD}}, K_{\text{LOD}}, J, \bm{\rho}$}\;
        $\bbV_{\Int}(\hat{\beta}_1) \leftarrow$ \ComputeVar{$I, K, J, \bm{\rho}$}\;
        \lIf{$\bbV_{\Dec}(\hat{\beta}_1) \leq 0$ \normalfont\textbf{or} $\bbV_{\Int}(\hat{\beta}_1) \leq 0$}{\Return $\infty$}
        \BlankLine
        \Return $\bbV_{\Dec}(\hat{\beta}_1) \,/\, \bbV_{\Int}(\hat{\beta}_1)$\;
    }

    \BlankLine
    $\calC_{\text{ineq}}(\bm{\rho}, K, J)$:\;
    \Indp
    $\rho^E_1 \leq \rho^E_0$, \quad $\rho^C_1 \leq \rho^C_0$, \quad $\rho^{EC}_1 \leq \rho^{EC}_0 \leq \rho^{EC}_2$, \quad $\rho^{EC}_0 \leq \min(\rho^E_0, \rho^C_0)$, \quad $\rho^{EC}_1 \leq \min(\rho^E_1, \rho^C_1)$\;
    $\min\{\lambda_1^+, \lambda_1^-, \lambda_2^+, \lambda_2^-, \lambda_3^+, \lambda_3^-\} > 0$\;
    \Indm
    
\end{algorithm}

In practice, ICC parameters are rarely known exactly at the design stage. To address this uncertainty, we consider a parameter space $\bm{\Theta} \coloneqq \{\bm{\rho}: \bm{\rho}_{\min} \preceq \bm{\rho} \preceq \bm{\rho}_{\max}\}$, where $\preceq$ denotes element-wise inequality, and adopt a MaxiMin strategy following \citet{Manju2014} and \citet{Liu2024} that maximizes the worst-case relative efficiency such that
\begin{align} 
    (I_{\text{MaxiMin}}, K_{\text{MaxiMin}}) = \arg\max_{(I,K) \in \calD} \min_{\bm{\rho} \in \bm{\Theta}} \text{RE}, \label{eqn:maximin} 
\end{align} 
where $\text{RE} \coloneqq \bbV_{\text{dec}}(\hat\beta_1)/\bbV_{(I, K)}(\hat\beta_1)$ is the relative efficiency between the decimal-valued LOD and the integer sample size combination $(I, K)$. Specifically, Algorithm \ref{alg:MMD} solves~\eqref{eqn:maximin} through nested optimization: for each candidate design, the inner minimization identifies the ICC configuration $\bm{\rho}_{\text{worst}}$ yielding the lowest relative efficiency (i.e., the worst-case scenario), while the outer maximization selects the design whose worst-case efficiency is highest among all feasible combinations. The inner optimization is subject to both the ICC ordering constraints described in Section~\ref{sec:set_up} and an additional constraint requiring all eigenvalues in Theorem~\ref{thm:eigenvalues} to be positive, ensuring that the correlation matrix $\mathbf{R}_i$ remains positive definite across the parameter space. By maximizing the worst-case relative efficiency over all feasible combinations, the resulting MMD provides robust performance guarantees across the entire parameter space $\bm{\Theta}$, when there is uncertainty around the ICC parameters in the study design stage. In what follows, we provide a detailed investigation of this approach in the context of each variant of the L-CRT design.
\section{Optimal Cluster Randomized Crossover Designs With Multiple Periods} \label{sec:crxo}

\subsection{Local Optimal Designs} \label{sec:crxo_LOD}

We consider CRXO trials with $J$ periods of equal length, where $J$ is even. For $J = 2$, clusters are randomized to either an intervention-control (IC) or control-intervention (CI) sequence, with $\pi$ denoting the proportion assigned to IC. For $J > 2$, the two-period pattern repeats (e.g., ICIC or CICI when $J = 4$). Under this design, $\bbV(\hat\beta_1)$ in Theorem~\ref{thm:variance} simplifies to
\begin{align}
    \bbV(\hat\beta_1) = \frac{\kappa^{C}\sigma_C^2 - 2\lambda\kappa^{EC}\sigma_C\sigma_E + \lambda^2\kappa^E\sigma_E^2}{IJK\pi(1 - \pi)}, \label{eqn:CRXO_var}
\end{align}
with the detailed derivation given in Appendix~\ref{supp_sec:var_crxo}.

To derive the LOD under the budget constraint $B$, we first note that clinical outcomes and costs are measured on different scales with heterogeneous variances. Because direct application of the ceiling ratio $\lambda$ in the optimization can yield designs that are disproportionately driven by the outcome with larger variance (e.g., when cost variability greatly exceeds that of the clinical outcome, the resulting design may only optimize for the cost outcome while largely ignoring the clinical outcome, or vice versa), we define $r \coloneqq \sigma_E/\sigma_C$ as the ratio of effect to cost standard deviations. The product $\lambda r$ then represents the standardized ceiling ratio, which adjusts willingness-to-pay by the relative variability of the two outcomes so that neither outcome dominates the design optimization only due to its scale.  The decimal-valued LOD is given as
\begin{align*}
    I_{\text{LOD}} = \frac{B}{c_1 + \sqrt{\vartheta c_1c_2J}}, \quad \text{and} \quad K_{\text{LOD}} = \sqrt{\frac{c_1\vartheta}{c_2J}},
\end{align*}
where
\begin{align}
    \vartheta = \frac{(1-\rho_1^E) + 2(\rho_1^{EC}-\rho_2^{EC})(\lambda r)^{-1} + (1-\rho_1^C)(\lambda r)^{-2}}{(\rho_0^E-\rho_1^E) + 2(\rho_1^{EC}-\rho_0^{EC})(\lambda r)^{-1} + (\rho_0^C-\rho_1^C)(\lambda r)^{-2}} - 1. \label{eqn:vartheta_crxo}
\end{align}
The proof is provided in Appendices~\ref{supp_sec:vartheta_crxo} and \ref{supp_sec:LOD_crxo}. These expressions reveal that the unit costs $(c_1, c_2)$ and the correlation structure (through $\vartheta$) jointly determine the optimal cluster-period size $K_{\text{LOD}}$, which is independent of the total budget $B$; in contrast, the optimal number of clusters $I_{\text{LOD}}$ depends on both the total budget and unit costs. When unit costs are fixed and the total budget increases, $K_{\text{LOD}}$ remains constant while $I_{\text{LOD}}$ increases. The expression for $K_{\text{LOD}}$ further demonstrates that the optimal cluster-period size increases when $c_1$ increases and $c_2$ decreases, implying that when the cost per cluster is large relative to the cost per individual, the LOD favors larger cluster-period sizes. This is consistent with practice, as recruiting a new cluster is typically more expensive than recruiting an additional individual within existing clusters (i.e., $c_1 > c_2$).

To provide additional insights into how the correlation structure affects LODs and power, Table~\ref{tab:crxo_LOD_0.5} presents LODs for CRXO trials under varying ICC and design parameters. Following \citet{Liu2024}, we consider CRXO trials with $J \in \{2, 4, 6\}$ periods and a balanced design ($\pi = 0.5$) between treatment sequences. We specify a total budget of $B = \$300{,}000$ with cluster-level cost $c_1 = \$3{,}000$ and individual-level cost $c_2 = \$250$. The average INMB is set to $\beta_1 = 4{,}000$ and the ceiling ratio to $\lambda = \$20{,}000$. The ratio of standard deviations is defined as $r = \sigma_E/\sigma_C = 1/3000$, yielding a standardized ceiling ratio of $\lambda r = 20/3$; this configuration reflects substantially higher variability in costs compared to clinical outcomes, a pattern commonly observed in economic evaluations \citep{Gomes2012developing}. For the ICC parameters, we set the within-period effect ICC to $\rho_0^E \in \{0.05, 0.10, 0.20\}$, with $\rho_0^C/\rho_0^E = 1$ and $\rho_0^{EC}/\rho_0^E = 0.4$, and determine the between-period ICCs through CAC values of 0.5 and 0.8 where $\text{CAC} = \rho_1^E/\rho_0^E = \rho_1^C/\rho_0^C = \rho_1^{EC}/\rho_0^{EC}$. The within-individual effect-cost ICC is fixed at $\rho_2^{EC} = 0.5$ across all scenarios. These values are among the commonly reported values in L-CRTs \citep{Korevaar2021}.

\begin{table}[t]
\caption{LODs for cross-sectional, multiple-period CRXO trials and PA-LCRTs (integer estimates): optimal number of clusters $I_{\text{LOD}}$, optimal cluster-period size $K_{\text{LOD}}$, and power of the LOD for detecting the average INMB, assuming known ICC parameters $\bm{\rho} = (\rho^E_0, \rho^E_1, \rho^C_0, \rho^C_1, \rho^{EC}_0, \rho^{EC}_1, \rho^{EC}_2)^\prime$, total budget $B$, cluster-level cost $c_1$, and individual-level cost $c_2$.}\label{tab:crxo_LOD_0.5}
{\centering
\resizebox{\linewidth}{!}{\begin{tabular}{clllccccccccc}
    \toprule
    \multirow{2}{*}{Design} & Effect Association & Cost Association & Cost-Effect Association & \multicolumn{3}{c}{$J = 2$} & \multicolumn{3}{c}{$J = 4$} & \multicolumn{3}{c}{$J = 6$} \\
    \cmidrule(lr){5-7} \cmidrule(lr){8-10} \cmidrule(lr){11-13}
    & $(\rho_0^E, \rho_1^E)$ & $(\rho_0^C, \rho_1^C)$ & $(\rho_0^{EC}, \rho_1^{EC}, \rho_2^{EC})$ & $I_{\text{LOD}}$ & $K_{\text{LOD}}$ & Power & $I_{\text{LOD}}$ & $K_{\text{LOD}}$ & Power & $I_{\text{LOD}}$ & $K_{\text{LOD}}$ & Power \\
    \midrule
    \multirow{6}{*}{CRXO} & (0.05, 0.025) & (0.05, 0.025) & (0.02, 0.010, 0.5) & 30 & 14 & 0.774 & 20 & 12 & 0.841 & 20 & 8 & 0.870 \\ 
    & (0.05, 0.040) & (0.05, 0.040) & (0.02, 0.016, 0.5) & 20 & 24 & 0.858 & 12 & 22 & 0.894 & 10 & 18 & 0.910 \\ 
    & (0.10, 0.050) & (0.10, 0.050) & (0.04, 0.020, 0.5) & 40 & 9 & 0.692 & 30 & 7 & 0.790 & 22 & 7 & 0.827 \\ 
    & (0.10, 0.080) & (0.10, 0.080) & (0.04, 0.032, 0.5) & 26 & 17 & 0.813 & 20 & 12 & 0.873 & 20 & 8 & 0.896 \\ 
    & (0.20, 0.100) & (0.20, 0.100) & (0.08, 0.040, 0.5) & 46 & 7 & 0.597 & 42 & 4 & 0.714 & 40 & 3 & 0.781 \\ 
    & (0.20, 0.160) & (0.20, 0.160) & (0.08, 0.064, 0.5) & 40 & 9 & 0.758 & 30 & 7 & 0.847 & 20 & 8 & 0.879 \\ 
    \midrule
    \multirow{6}{*}{PA} & (0.05, 0.025) & (0.05, 0.025) & (0.02, 0.010, 0.5) & 40 & 9 & 0.610 & 42 & 4 & 0.630 & 40 & 3 & 0.653 \\ 
    & (0.05, 0.040) & (0.05, 0.040) & (0.02, 0.016, 0.5) & 40 & 9 & 0.575 & 42 & 4 & 0.579 & 40 & 3 & 0.590 \\ 
    & (0.10, 0.050) & (0.10, 0.050) & (0.04, 0.020, 0.5) & 50 & 6 & 0.492 & 50 & 3 & 0.527 & 50 & 2 & 0.540 \\ 
    & (0.10, 0.080) & (0.10, 0.080) & (0.04, 0.032, 0.5) & 50 & 6 & 0.455 & 50 & 3 & 0.467 & 50 & 2 & 0.471 \\ 
    & (0.20, 0.100) & (0.20, 0.100) & (0.08, 0.040, 0.5) & 60 & 4 & 0.376 & 60 & 2 & 0.411 & 50 & 2 & 0.417 \\ 
    & (0.20, 0.160) & (0.20, 0.160) & (0.08, 0.064, 0.5) & 60 & 4 & 0.342 & 60 & 2 & 0.353 & 50 & 2 & 0.341 \\ 
    \bottomrule
\end{tabular}}
\par}

\vspace{0.5em}
\footnotesize
CRXO: cluster randomized crossover; PA: parallel-arm; LOD: local optimal design; $\beta_1 = 4000$: incremental net monetary benefit (INMB); $\lambda = 20{,}000$: ceiling ratio; $r = \sigma_E/\sigma_C = 1/3000$: ratio of effect to cost standard deviations; $\rho_0^E$: within-period effect ICC; $\rho_1^E$: between-period effect ICC; $\rho_0^C$: within-period cost ICC; $\rho_1^C$: between-period cost ICC; $\rho_0^{EC}$: within-period effect-cost ICC; $\rho_1^{EC}$: between-period effect-cost ICC; $\rho_2^{EC}$: within-individual effect-cost ICC; $J$: number of periods; $I_{\max} = 100$: maximum number of clusters; $K_{\max} = 200$: maximum cluster-period size; $B = 300{,}000$: total budget; $c_1 = 3000$: cost per cluster; $c_2 = 250$: cost per individual. For CRXO trials, clusters are equally allocated to two treatment sequences: IC/CI for $J = 2$, ICIC/CICI for $J = 4$, and ICICIC/CICICI for $J = 6$. For PA-LCRTs, clusters are equally allocated to treatment or control for the entire study duration.
\end{table}

Table~\ref{tab:crxo_LOD_0.5} reveals several important findings regarding the relationship among the correlation structure, design parameters, and power of the LOD. First, increasing the number of periods from $J = 2$ to $J = 6$ improves power of the LOD. For example, when $\rho_0^E = 0.05$ with CAC $= 0.5$, power increases from 0.774 to 0.870 as $J$ increases from 2 to 6. Second, this LOD power gain is accompanied by reduced resource requirements for each period: the optimal number of clusters decreases from $I_{\text{LOD}} = 30$ to $I_{\text{LOD}} = 20$, and the optimal cluster-period size decreases from $K_{\text{LOD}} = 14$ to $K_{\text{LOD}} = 8$. Third, higher CAC values yield greater power of the LOD for the same within-period ICCs; for example, with $\rho_0^E = 0.10$ and $J = 2$, designs with CAC $= 0.8$ achieve power of 0.813 compared to 0.692 with CAC $= 0.5$, suggesting that stronger between-period correlations can be leveraged to improve efficiency in the CRXO design. We also conduct sensitivity analyses varying the ratio of cost to effect ICCs ($\rho_0^C/\rho_0^E \in \{0.8, 1.2\}$) and the within-individual effect-cost ICC ($\rho_2^{EC} \in \{0.5, 0.8\}$), with results presented in Table~\ref{supp_tab:crxo_LOD} of Appendix~\ref{supp_sec::figures_tables}. These parameters have minimal impact on the optimal sample size combination but do affect power: increasing $\rho_0^C/\rho_0^E$ yields modest power differences, while increasing $\rho_2^{EC}$ improves power across all scenarios.

\begin{figure}
    \centering
    \includegraphics[width=\linewidth]{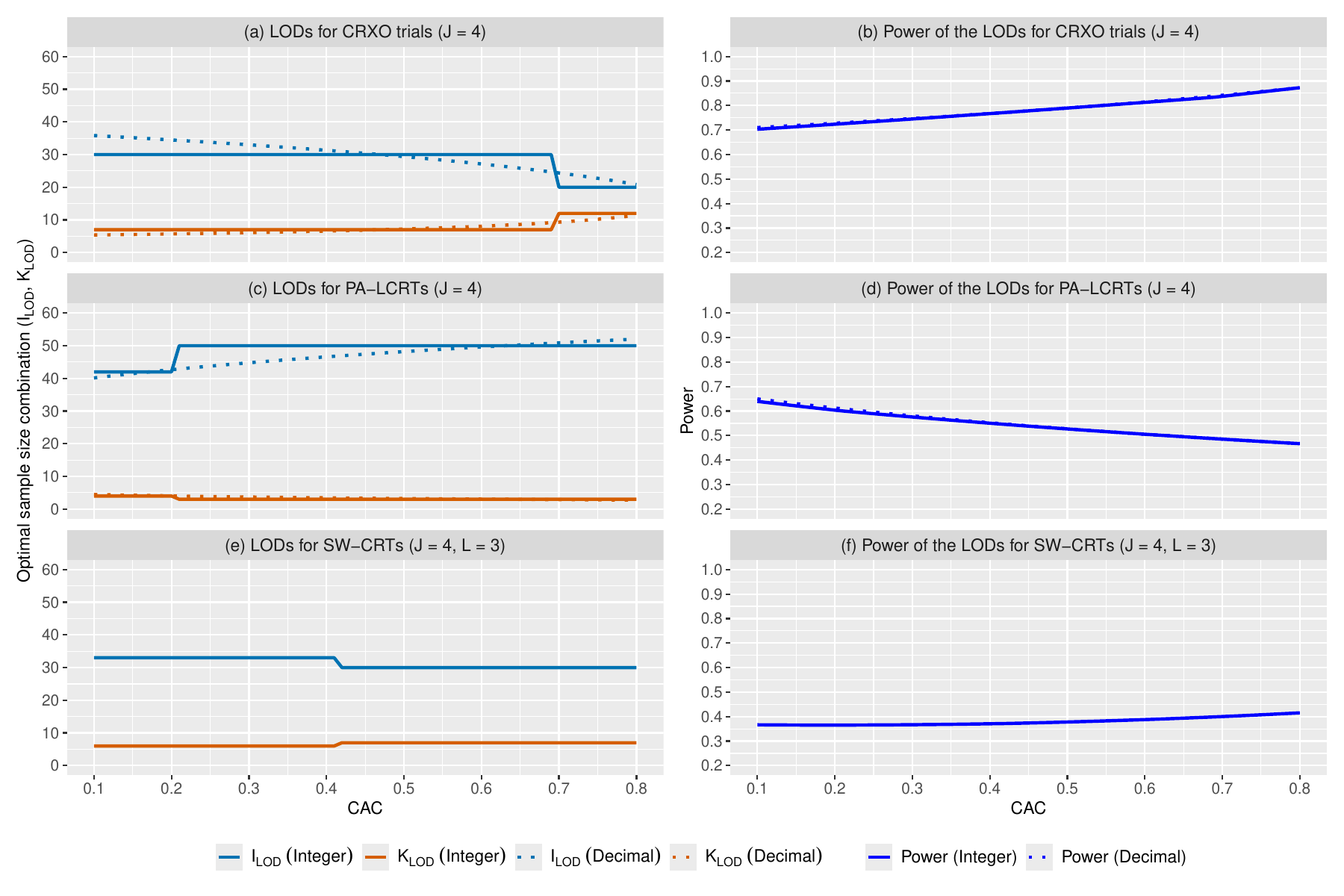}
    \caption{LODs for CRXO trials, PA-LCRTs, and SW-CRTs with varying CAC. Panels (a), (c), and (e) present $I_{\text{LOD}}$ and $K_{\text{LOD}}$ as functions of CAC; panels (b), (d), and (f) display the corresponding power of the LOD. Both integer (solid lines) and decimal (dotted lines) estimates are shown. We consider $J = 4$ periods, $B = \$300{,}000$, $c_1 = \$3{,}000$, and $c_2 = \$250$. For CRXO trials, clusters are equally allocated to two treatment sequences (ICIC/CICI); for PA-LCRTs, clusters are equally allocated to treatment or control for the entire study duration; for SW-CRTs, the number of sequences is $L = 3$. For ICC parameters, we fix $\rho_0^E = 0.1$, $\rho_0^C/\rho_0^E = 1$, $\rho_0^{EC} = 0.04$, $\rho_2^{EC} = 0.5$, with CAC $= \rho_1^E/\rho_0^E = \rho_1^C/\rho_0^C = \rho_1^{EC}/\rho_0^{EC} \in [0.1, 0.8]$. Additional parameters include $\beta_1 = 4{,}000$, $\lambda = 20{,}000$, and $r = \sigma_E/\sigma_C = 1/3000$.}
    \label{fig:crxo_lod_J4}
\end{figure}

To examine the sensitivity of LODs to the CAC, we next focus on CRXO trials with $J = 4$ periods, fixing $\rho_0^C = 0.1$, $\rho_0^C/\rho_0^E = 1$, $\rho_0^{EC}/\rho_0^E = 0.4$, and $\rho_2^{EC} = 0.5$, while varying CAC $\in [0.1, 0.8]$. In Figure~\ref{fig:crxo_lod_J4}(a), we show that as CAC increases, the integer estimates of $I_{\text{LOD}}$ are non-increasing (with decimal estimates strictly decreasing), while the integer estimates of $K_{\text{LOD}}$ are non-decreasing (with decimal estimates strictly increasing). Of note, the change in the integer-valued optimal sample size combinations observed in Figure~\ref{fig:crxo_lod_J4}(a) is a consequence of the discrete optimization space by restricting $I_{\text{LOD}}$ and $K_{\text{LOD}}$ to integer values, consistent with LOD results reported in \citet{Liu2024}. Figure~\ref{fig:crxo_lod_J4}(b) demonstrates that power of the LOD increases with CAC under both integer and decimal estimates, where the decimal estimates yield slightly higher power across the range of CAC values. This is expected because the decimal estimates represent the theoretical optimum without integer constraints; however, the minimal discrepancy between integer and decimal power of the LOD indicates that integer constraints have negligible impact on statistical efficiency in this setting.

Figure~\ref{supp_fig:crxo_lod_J4_varying_lambda_r} in Appendix~\ref{supp_sec::figures_tables} examines the impact of the standardized ceiling ratio $\lambda r$ on LODs for CRXO trials with $J = 4$ periods. We fix $\rho_0^C = 0.1$, $\rho_0^C/\rho_0^E = 1$, $\rho_0^{EC}/\rho_0^E = 0.4$, $\rho_2^{EC} = 0.5$, $\lambda = 20{,}000$, and $\sigma_E = 1$, while varying $\sigma_C \in \{10{,}000, 20{,}000, 200{,}000\}$ to yield $\lambda r \in \{0.1, 1, 2\}$. Across all three scenarios, as CAC increases, the integer estimates of $I_{\text{LOD}}$ are non-increasing while the integer estimates of $K_{\text{LOD}}$ are non-decreasing, with the corresponding decimal estimates being strictly monotonic. Since $K_{\text{LOD}} \propto \sqrt{\vartheta}$ and $I_{\text{LOD}}$ is inversely related to $\sqrt{\vartheta}$, both maximizing the number of clusters and minimizing cluster-period size are achieved by minimizing $\vartheta$ with respect to $\lambda r$. In Appendix~\ref{supp_sec:LOD_crxo}, we show that the value of $\lambda r$ minimizing $\vartheta$ can be obtained in closed form by solving a quadratic equation with coefficients that depend on the ICC parameters, suggesting when $\lambda r$ is close to 1, the LOD achieves the largest number of clusters and smallest cluster-period size under this setting.

\begin{table}[t]
\caption{MMDs for cross-sectional, multiple-period CRXO trials and PA-LCRTs (integer estimates): optimal number of clusters $I_{\text{MMD}}$, cluster-period size $K_{\text{MMD}}$, and RE of the MMD to detect a treatment effect size of $|\beta_1|/\bbV(\beta_1)^{1/2}$ for a given parameter space $\bm{\Theta} = \{\bm{\rho}: \bm{\rho}_{\min} \leq \bm{\rho} \leq \bm{\rho}_{\max}\}$, assuming a total budget $B$, cost per cluster $c_1$, and cost per individual $c_2$.}\label{tab:MMD}
{\centering
\resizebox{\linewidth}{!}{
\begin{tabular}{cccccccccccc}
    \toprule
    \multirow{2}{*}{Design} & Association & Association & \multicolumn{3}{c}{$J = 2$} & \multicolumn{3}{c}{$J = 4$} & \multicolumn{3}{c}{$J = 6$} \\
    \cmidrule(lr){4-6} \cmidrule(lr){7-9} \cmidrule(lr){10-12}
    & $(\rho_{0\min}^E, \rho_{0\max}^E)$ & $(\rho_{1\min}^E, \rho_{1\max}^E)$ & $I_{\text{MMD}}$ & $K_{\text{MMD}}$ & RE & $I_{\text{MMD}}$ & $K_{\text{MMD}}$ & RE & $I_{\text{MMD}}$ & $K_{\text{MMD}}$ & RE \\
    \midrule
    \multirow{6}{*}{CRXO} & (0.05, 0.10) & (0.025, 0.040) & 36 & 10 & 0.937 & 30 & 7 & 0.979 & 22 & 7 & 0.937 \\ 
    & & (0.025, 0.045) & 34 & 11 & 0.925 & 26 & 8 & 0.913 & 22 & 7 & 0.937 \\ 
    & (0.05, 0.20) & (0.025, 0.040) & 46 & 7 & 0.947 & 36 & 5 & 0.909 & 32 & 4 & 0.914 \\ 
    & & (0.025, 0.045) & 46 & 7 & 0.947 & 36 & 5 & 0.909 & 32 & 4 & 0.914 \\ 
    & (0.10, 0.20) & (0.050, 0.080) & 46 & 7 & 0.962 & 36 & 5 & 0.924 & 32 & 4 & 0.929 \\ 
    & & (0.050, 0.090) & 46 & 7 & 0.962 & 36 & 5 & 0.924 & 32 & 4 & 0.929 \\ 
    \midrule
    \multirow{6}{*}{PA} & (0.05, 0.10) & (0.025, 0.040) & 46 & 7 & 0.981 & 42 & 4 & 0.963 & 40 & 3 & 0.973 \\ 
    & & (0.025, 0.045) & 46 & 7 & 0.979 & 42 & 4 & 0.959 & 40 & 3 & 0.966 \\ 
    & (0.05, 0.20) & (0.025, 0.040) & 54 & 5 & 0.973 & 50 & 3 & 0.985 & 50 & 2 & 0.947 \\ 
    & & (0.025, 0.045) & 54 & 5 & 0.972 & 50 & 3 & 0.982 & 50 & 2 & 0.947 \\ 
    & (0.10, 0.20) & (0.050, 0.080) & 60 & 4 & 0.969 & 50 & 3 & 0.961 & 50 & 2 & 0.969 \\ 
    & & (0.050, 0.090) & 60 & 4 & 0.969 & 50 & 3 & 0.955 & 50 & 2 & 0.962 \\ 
    \bottomrule
\end{tabular}}
\par}
\vspace{0.5em}
\footnotesize CRXO: cluster randomized crossover; PA: parallel-arm; MMD: MaxiMin design; $\beta_1 = 4000$: incremental net monetary benefit (INMB); $\lambda = 20{,}000$: ceiling ratio; $r = \sigma_E/\sigma_C = 1/3000$: ratio of effect to cost standard deviations; $\rho_0^E$: within-period effect ICC; $\rho_1^E$: between-period effect ICC; $\rho_0^C$: within-period cost ICC; $\rho_1^C$: between-period cost ICC; $\rho_0^{EC}$: within-period effect-cost ICC; $\rho_1^{EC}$: between-period effect-cost ICC; $\rho_2^{EC}$: within-individual effect-cost ICC; $\rho_0^C \in [0.8\rho_{0\min}^E, 0.8\rho_{0\max}^E]$: within-period cost ICC; $\rho_1^C \in [0.8\rho_{1\min}^E, 0.8\rho_{1\max}^E]$: between-period cost ICC; $\rho_0^{EC} \in [0.01, 0.02]$: within-period effect-cost ICC; $\rho_1^{EC} \in [0.005, 0.01]$: between-period effect-cost ICC; $\rho_2^{EC} \in [0.5, 0.8]$: within-individual effect-cost ICC; $J$: number of periods; $I_{\max} = 100$: maximum number of clusters; $K_{\max} = 200$: maximum cluster-period size; $B = 300{,}000$: total budget; $c_1 = 3000$: cost per cluster; $c_2 = 250$: cost per individual. For CRXO trials, clusters are equally allocated to two treatment sequences: IC/CI for $J = 2$, ICIC/CICI for $J = 4$, and ICICIC/CICICI for $J = 6$. For PA-LCRTs, clusters are equally allocated to treatment or control for the entire study duration.
\end{table}

\subsection{MaxiMin Optimal Designs} \label{sec:crxo_MaxiMin}

In practice, the ICCs are rarely known with certainty in the planning stage; however, investigators can often elicit a plausible range from pilot data or prior literature, motivating the use of maximin optimal designs that ensure robust efficiency in the worst case scenario. For CRXO trials, the relative efficiency between the decimal-valued LOD and any feasible integer design $(I, K)$ is given as
\begin{align} 
    \text{RE} = \frac{g(\vartheta)}{\vartheta + K} \times \frac{K}{c_1 + JKc_2}, \label{eqn:RE_CRXO} 
\end{align} 
where $g(\vartheta) = (\sqrt{c_1} + \sqrt{\vartheta c_2 J})^2$ and $\vartheta$ is defined in Section~\ref{sec:crxo_LOD}. The derivation is provided in Appendix~\ref{supp_sec:MMD_crxo}.

Table~\ref{tab:MMD} presents MMDs for CRXO trials under different parameter space specifications. We consider trials with $J \in \{2, 4, 6\}$ periods and a balanced design ($\pi = 0.5$). For the effect ICCs, we specify $\rho_0^E \in [\rho_{0\min}^E, \rho_{0\max}^E]$ with $(\rho_{0\min}^E, \rho_{0\max}^E) \in \{(0.05, 0.10), (0.05, 0.20), (0.10, 0.20)\}$, and $\rho_1^E \in [\rho_{1\min}^E, \rho_{1\max}^E]$ with the ratio $\rho_{1\min}^E/\rho_{0\min}^E$ fixed at $0.5$ and $\rho_{1\max}^E/\rho_{0\min}^E \in \{0.8, 0.9\}$. The cost ICCs are set proportionally such that $\rho_0^C \in [0.8\rho_{0\min}^E, 0.8\rho_{0\max}^E]$ and $\rho_1^C \in [0.8\rho_{1\min}^E, 0.8\rho_{1\max}^E]$, while the between-outcome ICCs are specified as $\rho_0^{EC} \in [0.01, 0.02]$, $\rho_1^{EC} \in [0.005, 0.01]$, and $\rho_2^{EC} \in [0.5, 0.8]$. All remaining parameters follow the specifications described in Table~\ref{tab:crxo_LOD_0.5}.

Table~\ref{tab:MMD} reveals several important patterns in MMDs across different parameter space specifications. First, as the number of periods increases from $J = 2$ to $J = 6$, both $I_{\text{MMD}}$ and $K_{\text{MMD}}$ generally decrease, consistent with the patterns observed for LODs in Section~\ref{sec:crxo_LOD}. Second, wider parameter spaces for $\rho_0^E$ favor designs with more clusters and smaller cluster-period sizes; for example, when the parameter space widens from $[0.05, 0.10]$ to $[0.05, 0.20]$ with $J = 6$, $I_{\text{MMD}}$ increases from 22 to 32 while $K_{\text{MMD}}$ decreases from 7 to 4. Third, the RE values exceed 90\% across all scenarios, indicating that MMDs maintain near-optimal efficiency even under parameter uncertainty, with comparable RE values across different parameter space specifications further suggesting that these designs provide robust performance without substantial efficiency loss.

\begin{figure}[t]
    \centering
    \includegraphics[width=\linewidth]{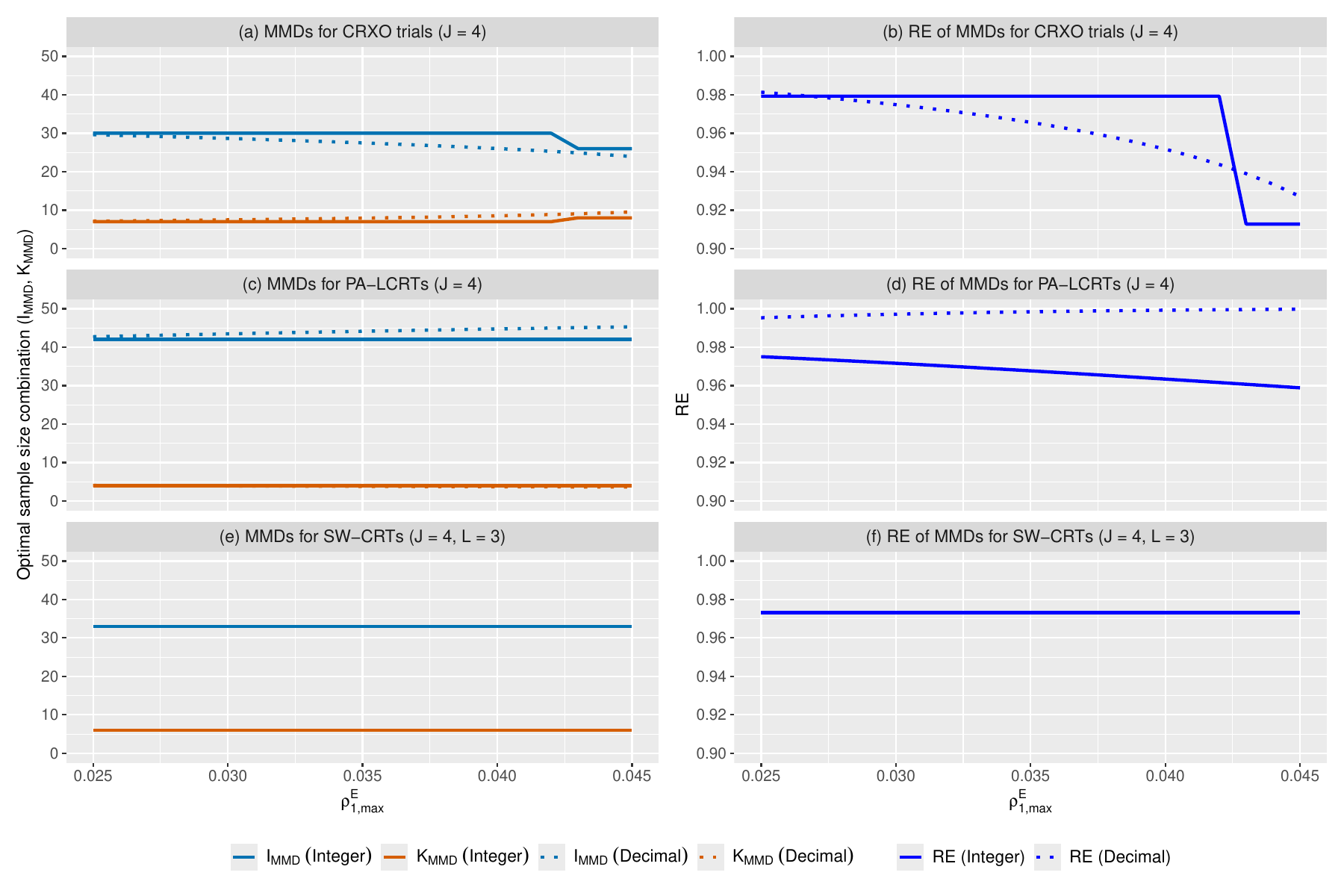}
    \caption{MMDs for CRXO trials, PA-LCRTs, and SW-CRTs with varying maximum between-period effect ICC. Panels (a), (c), and (e) present $I_{\text{MMD}}$ and $K_{\text{MMD}}$ as functions of $\rho_{1,\max}^E$, and panel (b), (d), and (f) display the corresponding RE. Both integer (solid lines) and decimal (dotted lines) estimates are shown. The MMDs are computed for CRXO trials with $J = 4$ periods and balanced design $\pi = 0.5$ under a fixed budget constraint of $B = \$300{,}000$, with $c_1 = \$3{,}000$ and $c_2 = \$250$. The ICC parameter ranges are specified as $\rho_0^E \in [0.05, 0.10]$, $\rho_0^C \in [0.04, 0.08]$, $\rho_0^{EC} \in [0.01, 0.02]$, $\rho_1^{EC} \in [0.005, 0.01]$, $\rho_2^{EC} \in [0.5, 0.8]$. We set $\rho_1^E \in [0.025, \rho_{1,\max}^E]$ and $\rho_1^C \in [0.02, 0.8\rho_{1,\max}^E]$ where $\rho_{1,\max}^E \in (0.025, 0.045]$. Additional parameters include $\beta_1 = 4{,}000$, $\lambda = 20{,}000$, and $r = \sigma_E/\sigma_C = 1/3000$.} \label{fig:crxo_mmd_J4}
\end{figure}

In Figure~\ref{fig:crxo_mmd_J4}, we present MMDs for CRXO trials with $J = 4$ periods under varying upper bounds of the between-period effect ICC. We specify $\rho_0^E \in [0.05, 0.10]$ and $\rho_0^C \in [0.04, 0.08]$ for the within-period ICCs, $\rho_1^E \in [0.025, \rho_{1,\max}^E]$ with $\rho_{1,\max}^E \in (0.025, 0.045]$ and $\rho_1^C \in [0.02, 0.8\rho_{1,\max}^E]$ for the between-period ICCs, and $\rho_0^{EC} \in [0.01, 0.02]$, $\rho_1^{EC} \in [0.005, 0.01]$, and $\rho_2^{EC} \in [0.5, 0.8]$ for the between-outcome ICCs. All other parameters are identical to those described in Table~\ref{tab:crxo_LOD_0.5}. Interestingly, the decimal MMD estimates follow the same closed-form expressions as in \citet{Liu2024}, but with $\vartheta$ replaced by the cost-effectiveness formulation in Equation~\eqref{eqn:vartheta_crxo}. 

Figure~\ref{fig:crxo_mmd_J4}(a) reveals two distinct regimes for the integer MMD. When $\rho_{1,\max}^E \leq 0.042$, the integer estimates remain constant at $(I_{\text{MMD}}, K_{\text{MMD}}) = (30, 7)$; when $\rho_{1,\max}^E > 0.042$, the design shifts to $(I_{\text{MMD}}, K_{\text{MMD}}) = (26, 8)$, favoring larger cluster-period sizes over additional clusters as the parameter space widens. In contrast, the decimal estimates vary continuously throughout, with $I_{\text{MMD}}$ decreasing from 29.6 to 23.9 and $K_{\text{MMD}}$ increasing from 7.1 to 9.5 as $\rho_{1,\max}^E$ increases from 0.025 to 0.045. Figure~\ref{fig:crxo_mmd_J4}(b) shows that RE decreases as parameter uncertainty grows: the integer RE remains stable at 0.979 across the first regime before dropping to 0.913 in the second, while the decimal RE declines steadily from 0.981 to 0.927. These results indicate that the integer MMD maintains robust efficiency under moderate parameter uncertainty, though the integer sample size combination shifts when the parameter space becomes sufficiently wide.
\section{Optimal Parallel-Arm Longitudinal Cluster Randomized Designs} \label{sec:pa}

\subsection{Local Optimal Designs} \label{sec:pa_LOD}

We next consider PA-LCRTs with $J$ periods of equal length, where clusters are randomized to either intervention or control for all periods, with $\pi$ denoting the proportion assigned to intervention. Under this design, the variance expression in Theorem~\ref{thm:variance} simplifies to
\begin{align}
    \bbV(\hat\beta_1) = \frac{\kappa^C \sigma_C^2 - 2 \lambda \kappa^{EC} \sigma_C \sigma_E + \lambda^2 \kappa^E \sigma_E^2}{I J K \pi \left(1 - \pi\right)} + \frac{\rho_1^C \sigma_C^2 - 2 \lambda \rho_1^{EC} \sigma_C \sigma_E + \lambda^2 \rho_1^E \sigma_E^2}{I \pi \left(1 - \pi\right)}, \label{eqn:pa_var}
\end{align}
with the derivation provided in Appendix~\ref{supp_sec:var_pa}. Compared to the CRXO variance in Equation~\eqref{eqn:CRXO_var}, Equation~\eqref{eqn:pa_var} contains an additional term involving the between-period ICCs. This difference arises because CRXO trials allow within-cluster contrasts where each cluster serves as its own control, whereas PA-LCRTs rely only on between-cluster comparisons.

Following the approach in Section~\ref{sec:crxo_LOD}, we use the standardized ceiling ratio $\lambda r$ to derive the LOD. The decimal-valued LOD is given as
\begin{align*}
    I_{\text{LOD}} = \frac{B}{c_1 + \sqrt{\vartheta c_1c_2J}}, \quad \text{and} \quad K_{\text{LOD}} = \sqrt{\frac{c_1\vartheta}{c_2J}},
\end{align*}
where
\begin{align}
    \vartheta = \frac{(1 + J\rho^E_1 - \rho^E_1) + 2(\rho_1^{EC} -J\rho_1^{EC} - \rho_2^{EC})(\lambda r)^{-1} + (1 + J\rho^C_1 - \rho^C_1)(\lambda r)^{-2}}{(J\rho^E_1 + \rho^E_0 - \rho^E_1) + 2(-J\rho_1^{EC} - \rho_0^{EC} + \rho_1^{EC})(\lambda r)^{-1} + (J\rho^C_1 + \rho^C_0 - \rho^C_1)(\lambda r)^{-2}} - 1. \label{eqn:vartheta_pa}
\end{align}
The proof is provided in Appendices~\ref{supp_sec:vartheta_pa} and \ref{supp_sec:LOD_pa}. These expressions share the same structural properties as the CRXO setting: $K_{\text{LOD}}$ depends only on unit costs and the correlation structure, while $I_{\text{LOD}}$ additionally depends on the total budget. However, the correlation structure enters differently through $\vartheta$ in Equation~\eqref{eqn:vartheta_pa}. As $\lambda r \to \infty$, this expression reduces to $\vartheta = (1 + J\rho^E_1 - \rho^E_1)/(J\rho^E_1 + \rho^E_0 - \rho^E_1) - 1$, corresponding to a design optimized only for clinical outcomes; conversely, as $\lambda r \to 0$, it reduces to $\vartheta = (1 + J\rho^C_1 - \rho^C_1)/(J\rho^C_1 + \rho^C_0 - \rho^C_1) - 1$, corresponding to a design optimized only for costs.

Table~\ref{tab:crxo_LOD_0.5} also presents LODs for PA-LCRTs under the same settings described in Section~\ref{sec:crxo_LOD}. The impact of the number of periods and CAC on design parameters follows similar patterns to those observed for CRXO trials, with two notable differences. First, the LOD for PA-LCRTs achieve lower power than the LOD for CRXO trials across all scenarios; for example, when $\rho_0^E = 0.05$ with CAC $= 0.5$ and $J = 2$, PA-LCRTs achieve power of 0.610 compared to 0.774 for CRXO trials. Second, PA-LCRTs require more clusters but smaller cluster-period sizes to achieve the optimal configuration; with $\rho_0^E = 0.10$, CAC $= 0.8$, and $J = 2$, PA-LCRTs yield $I_{\text{LOD}} = 50$ and $K_{\text{LOD}} = 6$, compared to $I_{\text{LOD}} = 26$ and $K_{\text{LOD}} = 17$ for CRXO trials.

In Figures~\ref{fig:crxo_lod_J4}(c) and (d), we illustrate the behavior of PA-LCRT LODs with varying CAC for $J = 4$ periods. The integer estimates of both $I_{\text{LOD}}$ and $K_{\text{LOD}}$ stabilize beyond CAC $\approx 0.21$, remaining constant at $I_{\text{LOD}} = 50$ and $K_{\text{LOD}} = 3$, with similar patterns observed when $\lambda r \in \{0.1, 1, 2\}$ in Figure~\ref{supp_fig:crxo_lod_J4_varying_lambda_r} of Appendix~\ref{supp_sec::figures_tables}. Figure~\ref{fig:crxo_lod_J4}(d) reveals that power of the LOD under PA-LCRT decreases as CAC increases, opposite to the pattern observed for CRXO trials. This occurs because higher between-period correlations reduce the effective information gained from additional periods in PA-LCRTs, where treatment comparisons rely only on between-cluster contrasts. The decimal estimates demonstrate similar declining power with increasing CAC, confirming that this efficiency loss is inherent to the PA-LCRT design rather than an artifact of integer constraints.

\subsection{MaxiMin Optimal Designs} \label{sec:pa_MaxiMin}
For PA-LCRTs, the relative efficiency between the decimal-valued LOD and any feasible integer design $(I, K)$ is given as
\begin{align}
    \text{RE}(\bm{\rho}, K) = \frac{g(\vartheta)}{\vartheta + K} \times \frac{K}{c_1 + c_2 J K}, \label{eqn:RE_PA}
\end{align}
where $g(\vartheta) = (\sqrt{c_1} + \sqrt{\vartheta c_2 J})^2$, with the proof provided in Appendix~\ref{supp_sec:MMD_pa}.

Table~\ref{tab:MMD} also presents MMDs for PA-LCRTs under the same settings described in Section~\ref{sec:crxo_MaxiMin}. The results show similar patterns to CRXO trials, with two key differences. First, the MMD for PA-LCRTs require more clusters but smaller cluster-period sizes; for example, when $\rho_0^E \in [0.05, 0.10]$ and $\rho_1^E \in [0.025, 0.040]$ with $J = 6$, PA-LCRTs yield $I_{\text{MMD}} = 40$ and $K_{\text{MMD}} = 3$ compared to $I_{\text{MMD}} = 22$ and $K_{\text{MMD}} = 7$ for CRXO trials. This difference becomes more evident with wider parameter spaces for $\rho_0^E$; when $\rho_0^E \in [0.05, 0.20]$ and $\rho_1^E \in [0.025, 0.040]$ with $J = 6$, PA-LCRTs require $I_{\text{MMD}} = 50$ and $K_{\text{MMD}} = 2$ versus $I_{\text{MMD}} = 32$ and $K_{\text{MMD}} = 4$ for CRXO trials. Second, PA-LCRTs maintain higher RE values than CRXO trials, typically exceeding 94\% across all scenarios, suggesting that the PA-LCRT design is less sensitive to ICC parameter uncertainty in the optimal design configuration.

In Figures~\ref{fig:crxo_mmd_J4}(c) and (d), we illustrate PA-LCRT MMDs with $J = 4$ periods under varying upper bounds of the between-period effect ICC, with parameter specifications identical to those described for CRXO trials in Section~\ref{sec:crxo_MaxiMin}. Unlike the CRXO results, where the integer estimates have two distinct regimes, the PA-LCRT integer estimates remain constant at $(I_{\text{MMD}}, K_{\text{MMD}}) = (42, 4)$ across the entire range of $\rho_{1,\max}^E \in [0.025, 0.045]$, while the decimal estimates show modest variation with $I_{\text{MMD}}$ increasing from 42.7 to 45.3 and $K_{\text{MMD}}$ decreasing from 4.0 to 3.6 as $\rho_{1,\max}^E$ increases. Figure~\ref{fig:crxo_mmd_J4}(d) shows that the integer RE decreases steadily from 0.975 to 0.959 as the parameter space widens, whereas the decimal RE remains near unity throughout, ranging from 0.995 to 0.9997. The stability of the integer sample size combination, combined with RE values exceeding 0.95, confirms that MMDs for PA-LCRTs are less sensitive to between-period ICC uncertainty than those for CRXO trials.
\section{Optimal Stepped Wedge Cluster Randomized Designs} \label{sec:swcrt}

\subsection{Local Optimal Designs} \label{sec:swcrt_LOD}

Finally, we consider a complete, cross-sectional SW-CRT with $I$ clusters, $J$ periods, and $K$ individuals per cluster-period. In this design, all clusters begin in the control condition and sequentially cross over to the intervention at pre-specified time points until all clusters receive the intervention \citep{Hussey2007}. We let $Q$ denote the number of treatment sequences, where each sequence corresponds to a distinct crossover time, and let $I_q$ denote the number of clusters assigned to sequence $q$, with $\sum_{q=1}^{Q} I_q = I$. The design requires $J \geq Q + 1$: when $J = Q + 1$, the trial consists of one baseline period (i.e., all clusters under control) followed by $Q$ periods during which clusters progressively cross over to the intervention; when $J > Q + 1$, the design additionally includes $J - Q - 1$ follow-up periods during which all clusters are under the intervention. For simplicity, we focus on balanced designs where an equal number of clusters cross over at each step (i.e., $I_q = I/Q$ for all $q \in \{1, \ldots, Q\}$). For example, the SW-CRT illustrated in Figure~\ref{fig:designs}(c) with $I = 12$ clusters, $J = 6$ periods, and $Q = 4$ treatment sequences yields $I_q = 3$ clusters crossing over at each step.

\begin{table}[t]
{\centering
\caption{LODs for SW-CRTs (integer estimates): optimal number of clusters $I_{\text{LOD}}$, optimal cluster-period size $K_{\text{LOD}}$, optimal number of periods $J_{\text{LOD}}$, and power of the LOD for detecting the average INMB across varying numbers of treatment sequences $Q$ and ICC parameters.}
\label{tab:sw_crt_4}
\resizebox{\linewidth}{!}{\begin{tabular}{lccccccccccccc}
\hline
Association & \multicolumn{4}{c}{$Q = 3$} & \multicolumn{4}{c}{$Q = 5$} & \multicolumn{4}{c}{$Q = 7$} \\
\cmidrule(lr){3-6}\cmidrule(lr){7-10}\cmidrule(lr){11-14}
parameter $(\rho_0^E, \rho_1^E)$ & $J$ & $J_{\text{LOD}}$ & $I_{\text{LOD}}$ & $K_{\text{LOD}}$ & Power & $J_{\text{LOD}}$ & $I_{\text{LOD}}$ & $K_{\text{LOD}}$ & Power & $J_{\text{LOD}}$ & $I_{\text{LOD}}$ & $K_{\text{LOD}}$ & Power \\
\toprule
    (0.05, 0.025) & Unconstrained &   4 & 30 & 7 & 0.436 &   6 & 25 & 6 & 0.520 &   8 & 14 & 9 & 0.526 \\ 
    & 9 &  & 21 & 5 & 0.270 &  & 25 & 4 & 0.414 &  & 21 & 5 & 0.524 \\ 
    (0.05, 0.040) & Unconstrained &   4 & 15 & 17 & 0.452 &   6 & 20 & 8 & 0.529 &   8 & 14 & 9 & 0.544 \\ 
    & 9 &  & 18 & 6 & 0.284 &  & 10 & 12 & 0.431 &  & 21 & 5 & 0.527 \\ 
    (0.10, 0.050) & Unconstrained &   4 & 30 & 7 & 0.378 &   6 & 25 & 6 & 0.455 &   8 & 42 & 2 & 0.468 \\ 
    & 9 &  & 21 & 5 & 0.249 &  & 25 & 4 & 0.381 &  & 21 & 5 & 0.467 \\ 
    (0.10, 0.080) & Unconstrained &   4 & 30 & 7 & 0.415 &   6 & 25 & 6 & 0.488 &   8 & 14 & 9 & 0.500 \\ 
    & 9 &  & 18 & 6 & 0.276 &  & 25 & 4 & 0.406 &  & 21 & 5 & 0.498 \\ 
    (0.20, 0.100) & Unconstrained &   4 & 42 & 4 & 0.319 &   6 & 40 & 3 & 0.404 &   8 & 42 & 2 & 0.433 \\ 
    & 9 &  & 30 & 3 & 0.227 &  & 40 & 2 & 0.341 &  & 21 & 5 & 0.402 \\ 
    (0.20, 0.160) & Unconstrained &   4 & 30 & 7 & 0.382 &   6 & 25 & 6 & 0.455 &   9 & 21 & 5 & 0.477 \\ 
    & 9 &  & 21 & 5 & 0.271 &  & 25 & 4 & 0.400 &  & 21 & 5 & 0.477 \\ 
\bottomrule
\end{tabular}}
\par}

\vspace{0.5em}
\footnotesize
SW-CRT: stepped wedge cluster randomized trial; LOD: local optimal design; $Q$: number of treatment sequences; $J$: number of periods; $K$: cluster-period size; $\rho_0^E$: within-period effect ICC; $\rho_1^E$: between-period effect ICC; $\rho_0^C = \rho_0^E$: within-period cost ICC; $\rho_1^C = \rho_1^E$: between-period cost ICC; $\rho_0^{EC} = 0.4\rho_0^E$: within-period effect-cost ICC; $\rho_1^{EC} = 0.4\rho_1^E$: between-period effect-cost ICC; $\rho_2^{EC} = 0.5$: within-individual effect-cost ICC; $\beta_1 = 4{,}000$: incremental net monetary benefit (INMB); $\lambda = 20{,}000$: ceiling ratio; $r = \sigma_E/\sigma_C = 1/3000$: standard deviation ratio; $B = \$300{,}000$: total budget; $c_1 = \$3{,}000$: cost per cluster; $c_2 = \$250$: cost per individual. We consider $J \geq Q + 1$, $I$ divisible by $Q$, $I_{\max} = 100$, $K_{\max} = 200$. "Unconstrained": optimal $J$ searched from 4 to 9 to maximize power; "9": fixed at $J = 9$.
\end{table}

Table~\ref{tab:sw_crt_4} presents LODs for SW-CRTs under varying numbers of treatment sequences $Q \in \{3, 5, 7\}$ and ICC parameters. For the number of periods, we consider two scenarios: an ``Unconstrained'' scenario where $J$ is optimized over the range $[Q + 1, 9]$ to maximize power, and a ``Fixed'' scenario where $J = 9$. We set the within-period effect ICC to $\rho_0^E \in \{0.05, 0.10, 0.20\}$, with the between-period effect ICC $\rho_1^E$ determined by CAC values of 0.5 and 0.8. The cost ICCs are set proportionally such that $\rho_0^C/\rho_0^E = \rho_1^C/\rho_1^E = 1$, while the between-outcome ICCs are specified as $\rho_0^{EC}/\rho_0^E = \rho_1^{EC}/\rho_1^E = 0.4$ and $\rho_2^{EC} = 0.5$. All other parameters follow the specifications described in Table~\ref{tab:crxo_LOD_0.5}.

Table~\ref{tab:sw_crt_4} reveals several important findings regarding the optimal design properties of SW-CRTs for cost-effectiveness analyses. First, increasing the number of treatment sequences from $Q = 3$ to $Q = 7$ improves power of the LOD across all ICC configurations; for example, under the ``Unconstrained'' scenario with $(\rho_0^E, \rho_1^E) = (0.05, 0.04)$, power increases from 0.452 to 0.544 as $Q$ increases from 3 to 7. Second, the optimal number of periods under the ``Unconstrained'' scenario is typically $J_{\text{LOD}} = Q + 1$, indicating that SW-CRTs for cost-effectiveness analyses achieve the highest LOD power with the minimum number of periods required, thereby allowing resources to be concentrated on additional clusters rather than extended follow-up. To see this, when $Q = 3$ and $(\rho_0^E, \rho_1^E) = (0.10, 0.05)$, the optimal design with $J_{\text{LOD}} = 4$ yields power of 0.378 compared to 0.249 when $J$ is fixed at 9. Third, higher CAC values improve the LOD power for the same within-period ICCs; when $\rho_0^E = 0.10$ and $Q = 5$ under the ``Unconstrained'' scenario, power increases from 0.455 to 0.488 as CAC increases from 0.5 to 0.8. Finally, higher within-period ICCs lead to designs favoring more clusters with smaller cluster-period sizes, but at the cost of reduced LOD power; when $\rho_0^E$ increases from 0.05 to 0.20 (with CAC $= 0.5$) under $Q = 3$, $I_{\text{LOD}}$ increases from 30 to 42, $K_{\text{LOD}}$ decreases from 7 to 4, and power decreases from 0.436 to 0.319.

In Figures~\ref{fig:crxo_lod_J4}(e) and (f), we illustrate LODs for SW-CRT with $J = 4$ periods and $Q = 3$ treatment sequences under varying CAC. The integer estimates of $I_{\text{LOD}}$ decrease from 33 to 30 around CAC $\approx 0.42$, while $K_{\text{LOD}}$ remains at 6 or 7 across the entire range. Figure~\ref{fig:crxo_lod_J4}(f) shows that power of the LOD increases gradually with CAC, rising from 0.367 to 0.415 as CAC increases from 0.1 to 0.8, a positive relationship consistent with CRXO trials but contrasting with PA-LCRTs. Among the three L-CRT designs, we find that the LOD for SW-CRTs achieves the lowest LOD power under comparable settings; when CAC $= 0.5$, optimal SW-CRTs yield power of 0.378 compared to 0.790 for optimal CRXO trials and 0.527 for optimal PA-LCRTs. This ranking is consistent with \citet{Liu2025}, who compared LODs across multiple L-CRT designs and similarly found that, under the cross-sectional setting, SW-CRTs require the largest sample size to achieve 80\% power, followed by PA-LCRTs and then CRXO trials. Figure~\ref{supp_fig:crxo_lod_J4_varying_lambda_r} in Appendix~\ref{supp_sec::figures_tables} further shows that SW-CRT LODs are sensitive to the standardized ceiling ratio $\lambda r$, with larger differences between $I_{\text{LOD}}$ and $K_{\text{LOD}}$ observed as $\lambda r \to 1$.

\begin{table}[t]
{\centering
\caption{MMDs for SW-CRTs (integer estimates): optimal number of clusters $I_{\text{MMD}}$, cluster-period size $K_{\text{MMD}}$, and E for a given parameter space $\bm{\Theta} = \{\bm{\rho}: \bm{\rho}_{\min} \leq \bm{\rho} \leq \bm{\rho}_{\max}\}$ across varying numbers of treatment sequences $Q$ and ICC parameters.}\label{tab:sw_mmd}
\resizebox{\linewidth}{!}{\begin{tabular}{cccccccccccccc}
\toprule
Association & Association & & \multicolumn{3}{c}{$Q = 3$} & & \multicolumn{3}{c}{$Q = 5$} & & \multicolumn{3}{c}{$Q = 7$} \\
\cmidrule(lr){4-6} \cmidrule(lr){8-10} \cmidrule(lr){12-14}
($\rho_{0_{\min}}^E$, $\rho_{0_{\max}}^E$) & ($\rho_{1_{\min}}^E$, $\rho_{1_{\max}}^E$) & $J$ & $I_{\text{MMD}}$ & $K_{\text{MMD}}$ & RE & & $I_{\text{MMD}}$ & $K_{\text{MMD}}$ & RE & & $I_{\text{MMD}}$ & $K_{\text{MMD}}$ & RE \\
\midrule
    (0.05, 0.10) & (0.025, 0.040) & Q + 1 & 33 & 6 & 0.973 &  & 25 & 6 & 0.941 &  & 21 & 5 & 0.857 \\ 
    & (0.025, 0.040) & 9 & 21 & 5 & 0.958 &  & 25 & 4 & 0.983 &  & 21 & 5 & 0.932 \\ 
    (0.05, 0.10) & (0.025, 0.045) & Q + 1 & 33 & 6 & 0.973 &  & 25 & 6 & 0.941 &  & 21 & 5 & 0.857 \\ 
    & (0.025, 0.045) & 9 & 21 & 5 & 0.958 &  & 25 & 4 & 0.983 &  & 21 & 5 & 0.932 \\ 
    (0.05, 0.20) & (0.025, 0.040) & Q + 1 & 42 & 4 & 0.956 &  & 40 & 3 & 0.983 &  & 28 & 3 & 0.802 \\ 
    & (0.025, 0.040) & 9 & 30 & 3 & 0.936 &  & 30 & 3 & 0.930 &  & 28 & 3 & 0.861 \\ 
    (0.05, 0.20) & (0.025, 0.045) & Q + 1 & 42 & 4 & 0.956 &  & 40 & 3 & 0.983 &  & 28 & 3 & 0.802 \\ 
    & (0.025, 0.045) & 9 & 30 & 3 & 0.936 &  & 30 & 3 & 0.930 &  & 28 & 3 & 0.861 \\ 
    (0.10, 0.20) & (0.050, 0.080) & Q + 1 & 42 & 4 & 0.960 &  & 40 & 3 & 0.984 &  & 42 & 2 & 0.854 \\ 
    & (0.050, 0.080) & 9 & 30 & 3 & 0.946 &  & 30 & 3 & 0.937 &  & 28 & 3 & 0.865 \\ 
    (0.10, 0.20) & (0.050, 0.090) & Q + 1 & 42 & 4 & 0.960 &  & 40 & 3 & 0.845 &  & 28 & 3 & 0.804 \\ 
    & (0.050, 0.090) & 9 & 30 & 3 & 0.946 &  & 30 & 3 & 0.937 &  & 28 & 3 & 0.865 \\
\bottomrule
\end{tabular}}
\par}
\vspace{0.5em}
\footnotesize SW-CRT: stepped wedge cluster randomized trial; MMD: MaxiMin design; RE: relative efficiency; $\beta_1 = 4{,}000$: incremental net monetary benefit (INMB); $\lambda = 20{,}000$: ceiling ratio; $r = \sigma_E/\sigma_C = 1/3000$: ratio of effect to cost standard deviations; $\rho_0^E$: within-period effect ICC; $\rho_1^E$: between-period effect ICC; $\rho_0^C$: within-period cost ICC; $\rho_1^C$: between-period cost ICC; $\rho_0^{EC}$: within-period effect-cost ICC; $\rho_1^{EC}$: between-period effect-cost ICC; $\rho_2^{EC}$: within-individual effect-cost ICC; $\rho_0^C \in [0.8\rho_{0\min}^E, 0.8\rho_{0\max}^E]$; $\rho_1^C \in [0.8\rho_{1\min}^E, 0.8\rho_{1\max}^E]$; $\rho_0^{EC} \in [0.01, 0.02]$; $\rho_1^{EC} \in [0.005, 0.01]$; $\rho_2^{EC} \in [0.5, 0.8]$; $Q$: number of treatment sequences; $J$: number of periods; $I_{\max} = 100$: maximum number of clusters; $K_{\max} = 200$: maximum cluster-period size; $B = \$300{,}000$: total budget; $c_1 = \$3{,}000$: cost per cluster; $c_2 = \$250$: cost per individual. We consider $J \geq Q + 1$, $I$ divisible by $Q$. "$Q + 1$": $J$ set to minimum required ($J = Q + 1$); "9": fixed at $J = 9$.
\end{table}

\subsection{MaxiMin Optimal Designs} \label{sec:swcrt_MaxiMin}

Table~\ref{tab:sw_mmd} presents MMDs for SW-CRTs under varying numbers of treatment sequences $Q \in \{3, 5, 7\}$ and parameter space specifications. We consider two scenarios for the number of periods: (a) $J = Q + 1$, the minimum required, and (b) $J = 9$, fixed across all treatment sequences. For the effect ICCs, we specify $\rho_0^E \in [\rho_{0\min}^E, \rho_{0\max}^E]$ with $(\rho_{0\min}^E, \rho_{0\max}^E) \in \{(0.05, 0.10), (0.05, 0.20), (0.10, 0.20)\}$, and $\rho_1^E \in [\rho_{1\min}^E, \rho_{1\max}^E]$ with upper bounds varying across configurations. The cost ICCs are scaled proportionally such that $\rho_0^C \in [0.8\rho_{0\min}^E, 0.8\rho_{0\max}^E]$ and $\rho_1^C \in [0.8\rho_{1\min}^E, 0.8\rho_{1\max}^E]$, while the between-outcome ICCs are specified as $\rho_0^{EC} \in [0.01, 0.02]$, $\rho_1^{EC} \in [0.005, 0.01]$, and $\rho_2^{EC} \in [0.5, 0.8]$. All other parameters follow the specifications described in Table~\ref{tab:sw_crt_4}.

Table~\ref{tab:sw_mmd} reveals several important findings regarding the robustness and design properties of SW-CRT MMDs. First, RE values exceed 0.80 across all scenarios, indicating that MMDs for SW-CRTs maintain robust efficiency under parameter uncertainty. Second, the $J = Q + 1$ scenario yields more clusters and larger cluster-period sizes than the $J = 9$ scenario, particularly for smaller $Q$; when $Q = 3$ and $(\rho_{0\min}^E, \rho_{0\max}^E) = (0.05, 0.10)$, the $J = Q + 1$ scenario yields $(I_{\text{MMD}}, K_{\text{MMD}}) = (33, 6)$ compared to $(21, 5)$ for $J = 9$. Third, as $Q$ increases, the optimal number of clusters decreases while cluster-period sizes remain relatively stable; when $J = Q + 1$ with $(\rho_{0\min}^E, \rho_{0\max}^E) = (0.05, 0.10)$, $I_{\text{MMD}}$ decreases from 33 to 21 as $Q$ increases from 3 to 7, whereas $K_{\text{MMD}}$ remains at 5 or 6. Finally, higher within-period ICC ranges lead to designs favoring more clusters with smaller cluster-period sizes; when $Q = 3$ under $J = Q + 1$, increasing $(\rho_{0\min}^E, \rho_{0\max}^E)$ from $(0.05, 0.10)$ to $(0.10, 0.20)$ shifts the MMD from $(I_{\text{MMD}}, K_{\text{MMD}}) = (33, 6)$ to $(42, 4)$. In Figures~\ref{fig:crxo_mmd_J4}(e) and (f), we further show that MMDs for SW-CRTs are insensitive to the upper bound $\rho_{1,\max}^E$, with both $(I_{\text{MMD}}, K_{\text{MMD}}) = (33, 6)$ and RE $= 0.973$ remaining constant across $\rho_{1,\max}^E \in [0.025, 0.045]$.
\section{Illustrative Applications to the Australia Reinvestment Trial} \label{sec:da}

We illustrate the proposed optimal study design methods using data from the Australian disinvestment and reinvestment stepped wedge trials, which evaluated whether weekend allied health services (e.g., physical therapy, occupational therapy, and social work) provided across acute medical and surgical hospital wards deliver health benefits commensurate with their costs \citep{Haines2017}. Two stepped wedge cluster randomized controlled trials were conducted across 12 wards from two metropolitan teaching hospitals in Australia from February 2014 to April 2015, with each period corresponding to one calendar month. Trial 1 employed a disinvestment design where the current weekend allied health service was removed from participating wards each month, while Trial 2 used a conventional stepped wedge design where a newly developed weekend allied health service was introduced. The primary clinical outcome is length of stay measured in days, which serves as a continuous measure of patient flow through the hospital; the secondary outcome is the cost (in Australian dollars) per patient to the healthcare system per admission.

We provide an illustrative data example in the context of Trial 2, and provide numerical results if a similar study were to proceed as a CRXO, PA-LCRT, or SW-CRT. We set $I_{\max} = 100$ clusters and $K_{\max} = 200$ individuals per cluster-period as upper bounds for the design space, allowing the optimization algorithms to identify the optimal sample size combination that achieves statistical power of at least 80\% under budget constraints. For CRXO and PA-LCRT designs, we fix $J = 8$ periods, while for SW-CRT designs, we specify $Q = 7$ treatment sequences with an equal number of clusters per sequence and examine $J \in \{8, 9, 10\}$ to assess how extended follow-up affects power for detecting cost-effectiveness.

To approximate ICC parameters for sample size calculation, we fit a bivariate linear mixed model based on~\eqref{eqn:model} using the \texttt{brms} package in R, which enables estimation of all variance components and between-outcome correlations from the trial data. For the length of stay outcome, we estimate the within-period ICC as $\rho_0^E \approx 0.048$ and the between-period ICC as $\rho_1^E \approx 0.042$, yielding CAC $= \rho_1^E/\rho_0^E \approx 0.88$ with residual standard deviation $\sigma_E = 6.48$ days. For the cost outcome, we estimate the within-period ICC as $\rho_0^C \approx 0.020$ and the between-period ICC as $\rho_1^C \approx 0.018$, yielding CAC $= \rho_1^C/\rho_0^C \approx 0.88$ with residual standard deviation $\sigma_C = \$11{,}635$. The similar CAC values for both outcomes suggest similar degrees of decay in within-cluster correlation across periods, which may reflect relatively stable patient flow patterns and resource utilisation within hospital wards over time. The ratio of residual standard deviations is $r = \sigma_E/\sigma_C \approx 5.57 \times 10^{-4}$. For the between-outcome ICCs, the bivariate model yields $\rho_0^{EC} \approx 0.007$, $\rho_1^{EC} \approx 0.004$, and $\rho_2^{EC} \approx 0.75$. The high within-individual correlation reflects the strong association between length of stay and cost at the patient level, as clinical costing data are largely driven by length of stay \citep{Haines2017}. For the LODs, we use these point estimates directly, while for the MMDs, we specify plausible ranges for the between-outcome ICCs such that $0 \leq \rho_0^{EC} \leq 0.01$, $0 \leq \rho_1^{EC} \leq 0.005$, and $0.5 < \rho_2^{EC} < 0.8$.

We next specify the ceiling ratio and effect size parameters in the context of Trial 2. Because the primary clinical outcome is length of stay measured in days, the ceiling ratio $\lambda$ represents the maximum willingness-to-pay per bed-day saved. Drawing on a contingent valuation study of Australian hospital Chief Executive Officers, \citet{Page2017} estimated a willingness-to-pay of \$216 per ward bed-day, and we adopt this value for illustration. For the anticipated treatment effect on length of stay, an earlier quasi-experimental study described in the trial protocol observed a reduction of approximately 1.44 days when weekend physiotherapy was provided to surgical inpatients \citep{Haines2015}, and we set $\hat{\alpha}_1 = 1.44$ days accordingly. For the cost outcome, \citet{Haines2017} indicated that newly introduced health services may initially increase healthcare costs before efficiency gains are realized, because service delivery models require time to be refined and integrated into routine care. For illustration, we set $\hat{\gamma}_1 = \$2{,}400$ per patient, which is consistent with the findings subsequently reported in Trial 2 \citep{Haines2017}. The resulting INMB is given as $\beta_1 = \lambda\hat{\alpha}_1 - \hat{\gamma}_1 = 216 \times 1.44 - 2{,}400 = -\$2{,}089$, where the negative value reflects a setting in which the clinical benefit of reduced length of stay, valued at \$216 per bed-day, does not offset the higher delivery costs of the newly introduced service. At the design stage, the sample size calculation targets $|\beta_1|$ to ensure adequate power for detecting whether the intervention achieves cost-effectiveness. Following the settings in Sections~\ref{sec:crxo}--\ref{sec:swcrt}, we set unit costs at $c_1 = \$3{,}000$ and $c_2 = \$250$, with a budget constraint of $B = \$600{,}000$ to ensure sufficient resources for all designs to achieve 80\% power.

For CRXO trials with $J = 8$ periods and a balanced design ($\pi = 0.5$), the LOD allocates $I_{\text{LOD}} = 8$ clusters and $K_{\text{LOD}} = 36$ individuals per cluster-period, achieving power of 0.996. The corresponding MMD yields the same optimal sample size combination with RE $= 0.991$, indicating minimal efficiency loss under between-outcome ICC uncertainty. For PA-LCRTs with $J = 8$ periods and a balanced design ($\pi = 0.5$), the LOD requires substantially more clusters ($I_{\text{LOD}} = 66$) but smaller cluster-period sizes ($K_{\text{LOD}} = 3$), achieving power of 0.893; the MMD for PA-LCRTs maintains the same optimal sample size configuration with RE $= 0.990$. This shift toward more clusters with fewer individuals per cluster-period is consistent with the patterns observed in Section~\ref{sec:pa}. For SW-CRTs with $Q = 7$ treatment sequences, we examine $J \in \{8, 9, 10\}$ periods. The LOD with $J = 8$ (i.e., the minimum number of periods required) yields $I_{\text{LOD}} = 35$ clusters and $K_{\text{LOD}} = 7$ individuals per cluster-period with power of 0.833. Increasing the number of periods reduces LOD power: $J = 9$ yields power of 0.799 with $(I_{\text{LOD}}, K_{\text{LOD}}) = (28, 8)$, and $J = 10$ yields power of 0.770 with $(I_{\text{LOD}}, K_{\text{LOD}}) = (21, 10)$, confirming the finding from Section~\ref{sec:swcrt} that SW-CRTs for cost-effectiveness analyses achieve the highest power with the minimum number of periods required. Because $J = Q + 1$ provides optimal power, we focus the MMD for SW-CRT on this configuration; the resulting design coincides with the LOD, with RE $= 0.979$. Across all three designs, CRXO trials achieve the highest LOD power, followed by PA-LCRTs and SW-CRTs.

\subsection{Further Illustration Based on the Incomplete Design} \label{sec:da_incomplete}

In the Australian reinvestment trial, the 6 wards from Dandenong Hospital started two months earlier than the 5 wards from Footscray Hospital and ended one month earlier, resulting in an incomplete design where certain cluster-periods were not observed. The variance formulas derived in Section~\ref{sec:swcrt} assume a complete design where all cluster-periods are observed and therefore cannot directly accommodate these missing data structures. To address this incomplete design, we extend our framework to incomplete SW-CRTs by introducing a selection matrix approach that allows each cluster to contribute to the analysis according to its observed periods.

To represent such designs, we follow \citet{Zhang2023} and adopt a Design Pattern matrix $\mathbf{F}$ that encodes the design structure, with entries $F_{ij} = 0$ for control, $F_{ij} = 1$ for intervention, and $F_{ij} = \cdot$ for missing cluster-periods. Table~\ref{tab:incomplete} displays the Design Pattern matrix for the Australian reinvestment trial, where clusters 1--6 have observed periods $\mathcal{J}_i = \{1, \ldots, 7\}$ and clusters 7--11 have $\mathcal{J}_i = \{3, \ldots, 8\}$. Here, $\mathcal{J}_i = \{j : F_{ij} \neq \cdot\}$ denotes the set of observed periods for cluster $i$ with $J_i = |\mathcal{J}_i|$ representing the total number of observed periods.

To derive the variance of the INMB estimator under incomplete designs, we first introduce the notation for the full design where all cluster-periods are observed. We work with cluster-period means for the bivariate outcomes such that for cluster $i$ in period $j$, $\bar E_{ij} = K^{-1}\sum_{k=1}^{K} E_{ijk}$ and $\bar C_{ij} = K^{-1}\sum_{k=1}^{K} C_{ijk}$ denote the cluster-period means for effect and cost, respectively. Stacking these means across all $J$ periods yields the $2J \times 1$ outcome vector
\begin{align*}
    \bar{\mathbf{Y}}_i = (\bar{E}_{i1}, \bar{C}_{i1}, \bar{E}_{i2}, \bar{C}_{i2}, \ldots, \bar{E}_{iJ}, \bar{C}_{iJ})'.
\end{align*}
We define the fixed effects vector as $\bm{\theta} = (\alpha_{01}, \gamma_{01}, \alpha_{02}, \gamma_{02}, \ldots, \alpha_{0J}, \gamma_{0J}, \alpha_1, \gamma_1)'$, which has dimension $(2J + 2) \times 1$, and the $2J \times (2J + 2)$ design matrix for the full design as $\mathbf{D}_i = \left(\mathbf{I}_J,\; \mathbf{Z}_i\right) \otimes \mathbf{I}_2$, where $\mathbf{Z}_i = (Z_{i1}, \ldots, Z_{iJ})'$ is the $J \times 1$ treatment indicator vector for cluster $i$. The $2J \times 2J$ covariance matrix for $\bar{\mathbf{Y}}_i$ is given as
\begin{align*}
    \mathbf{G}_i = \mathbf{I}_J \otimes \left(\bm{\Sigma}_s + \frac{1}{K}\bm{\Sigma}_\epsilon\right) + \mathbf{1}_J\mathbf{1}_J' \otimes \bm{\Sigma}_b,
\end{align*}
where $\bm{\Sigma}_b$, $\bm{\Sigma}_s$, and $\bm{\Sigma}_\epsilon$ are the $2 \times 2$ covariance matrices for the cluster-level, cluster-period-level, and individual-level random effects, respectively.

\begin{table}[hbtp]
    \centering 
    \caption{Design Pattern matrix $\mathbf{F}$ for the Australian reinvestment trial with an incomplete design structure, where ``0'' indicates control, ``1'' indicates intervention, and ``$\cdot$'' indicates missing cluster-periods. Clusters 1--6 from the Dandenong Hospital were observed during periods 1--7; clusters 7--11 from the Footscray Hospital were observed during periods 3-9.} \label{tab:incomplete} \begin{tabular}{lcccccccc}
        \toprule
        & \multicolumn{8}{c}{\textbf{Period}} \\
        \cmidrule{2-9}
        & 1 & 2 & 3 & 4 & 5 & 6 & 7 & 8 \\
        \midrule
        Cluster 1 & 0 & 1 & 1 & 1 & 1 & 1 & 1 & $\cdot$ \\
        Cluster 2 & 0 & 0 & 1 & 1 & 1 & 1 & 1 & $\cdot$ \\
        Cluster 3 & 0 & 0 & 0 & 1 & 1 & 1 & 1 & $\cdot$ \\
        Cluster 4 & 0 & 0 & 0 & 0 & 1 & 1 & 1 & $\cdot$ \\
        Cluster 5 & 0 & 0 & 0 & 0 & 0 & 1 & 1 & $\cdot$ \\
        Cluster 6 & 0 & 0 & 0 & 0 & 0 & 0 & 1 & $\cdot$ \\
        \midrule
        Cluster 7 & $\cdot$ & $\cdot$ & 0 & 1 & 1 & 1 & 1 & 1 \\
        Cluster 8 & $\cdot$ & $\cdot$ & 0 & 0 & 1 & 1 & 1 & 1 \\
        Cluster 9 & $\cdot$ & $\cdot$ & 0 & 0 & 0 & 1 & 1 & 1 \\
        Cluster 10 & $\cdot$ & $\cdot$ & 0 & 0 & 0 & 0 & 1 & 1 \\
        Cluster 11 & $\cdot$ & $\cdot$ & 0 & 0 & 0 & 0 & 0 & 1 \\
        \bottomrule
    \end{tabular} 
\end{table}

For incomplete designs, not all cluster-periods are observed, and we therefore construct a selection matrix to extract the relevant information from the full-design quantities defined above. We note that $\mathbf{D}_i$ and $\mathbf{G}_i$ are constructed as if all cluster-periods were observed; the corresponding entries in $\mathbf{Z}_i$ for missing cluster-periods (i.e., $Z_{ij}$ for $j \notin \mathcal{J}_i$) can be assigned arbitrary values (e.g., 0 or 1), because these entries do not affect the subsequent variance calculation since they will be removed by the selection matrix. Specifically, we let $\mathbf{S}_i$ denote the $J_i \times J$ matrix where, for each observed period $j \in \mathcal{J}_i$, the corresponding row contains a 1 in column $j$ and 0 elsewhere, with rows ordered by increasing $j$. Taking cluster 7 in the Australian reinvestment trial (Table~\ref{tab:incomplete}) as an example, $J_7 = 6$ and the selection matrix is
\begin{align*}
    \mathbf{S}_7 = \begin{pmatrix}
        0 & 0 & 1 & 0 & 0 & 0 & 0 & 0 \\
        0 & 0 & 0 & 1 & 0 & 0 & 0 & 0 \\
        0 & 0 & 0 & 0 & 1 & 0 & 0 & 0 \\
        0 & 0 & 0 & 0 & 0 & 1 & 0 & 0 \\
        0 & 0 & 0 & 0 & 0 & 0 & 1 & 0 \\
        0 & 0 & 0 & 0 & 0 & 0 & 0 & 1
    \end{pmatrix},
\end{align*}
which selects only the rows corresponding to periods 3-8 from the full-design matrices. We then define $\tilde{\mathbf{S}}_i = \mathbf{S}_i \otimes \mathbf{I}_2$ as the $2J_i \times 2J$ selection matrix for bivariate outcomes. The matrix $\tilde{\mathbf{S}}_i$ operates on the full-design quantities to yield only the observed components: $\tilde{\mathbf{S}}_i \bar{\mathbf{Y}}_i$ gives the observed outcome vector, $\tilde{\mathbf{S}}_i \mathbf{D}_i$ gives the design matrix, and $\tilde{\mathbf{S}}_i \mathbf{G}_i \tilde{\mathbf{S}}_i'$ gives the covariance matrix, each reduced to the observed cluster-periods for cluster $i$.

The variance of the INMB estimator under incomplete designs then follows from standard mixed model theory. Specifically, the covariance matrix $\bfSigma_{\hat\alpha_1, \hat\gamma_1}$ is obtained as the lower $2 \times 2$ submatrix of
\begin{align}
    \left(\sum_{i=1}^{I} \mathbf{D}_i' \tilde{\mathbf{S}}_i' \left(\tilde{\mathbf{S}}_i \mathbf{G}_i \tilde{\mathbf{S}}_i'\right)^{-1} \tilde{\mathbf{S}}_i \mathbf{D}_i \right)^{-1}, \label{eqn:incomplete_var}
\end{align}
and the variance $\mathbb{V}(\hat{\beta}_1)$ is computed as $\lambda^2 \mathbb{V}(\hat{\alpha}_1) + \mathbb{V}(\hat{\gamma}_1) - 2\lambda \text{Cov}(\hat{\alpha}_1, \hat{\gamma}_1)$. This formulation accommodates arbitrary incomplete cluster-periods, including implementation periods and other incomplete designs common in practice, by allowing each cluster to contribute to the analysis according to its observed periods.

Given the incomplete enrollment structure described above, Table~\ref{tab:incomplete} displays the corresponding Design Pattern matrix, where clusters 1--6 have observed periods $\mathcal{J}_i = \{1, \ldots, 7\}$ and clusters 7--11 have $\mathcal{J}_i = \{3, \ldots, 8\}$. For the following sample size calculation, we preserve this design pattern by letting the first half of clusters end one period earlier and the second half start two periods later. All other design parameters are identical to those described for the complete design. For $J = 8$ periods with $Q = 7$ treatment sequences, the LOD allocates $(I_{\text{LOD}}, K_{\text{LOD}}) = (28, 11)$, achieving power of 0.866. Compared to the complete design, the incomplete design achieves higher LOD power with fewer clusters but larger cluster-period sizes (i.e., power increases from 0.833 to 0.866). This improvement occurs because the incomplete design concentrates observations in periods with greater information content for estimating the treatment effect \citep{Li2023information}. Nevertheless, the LOD power of 0.866 for the incomplete SW-CRT remains lower than the LOD power achieved by complete CRXO trials (0.996) and complete PA-LCRTs (0.893). Increasing the number of periods reduces LOD power: $J = 9$ yields power of 0.845 with $(I_{\text{LOD}}, K_{\text{LOD}}) = (42, 6)$, and $J = 10$ yields power of 0.792 with $(I_{\text{LOD}}, K_{\text{LOD}}) = (28, 8)$, confirming that the minimum number of periods ($J = Q + 1 = 8$) achieves the highest power for incomplete SW-CRTs as well. Because $J = 8$ provides optimal power, we focus the MMD analysis on this configuration. The resulting design allocates $(I_{\text{MMD}}, K_{\text{MMD}}) = (7, 41)$ with RE $= 0.992$, which exceeds the RE $= 0.979$ obtained for the complete design MMD.
\section{Discussion} \label{sec:discussion}

This work makes several contributions to the design of cost-effectiveness L-CRTs. First, we provide a unified framework that jointly models clinical outcomes and costs through a bivariate linear mixed model, accommodating the complex correlation structures inherent in longitudinal cluster designs, including both within-period and between-period ICCs for each outcome as well as three types of between-outcome ICCs. The closed-form variance expression derived in Theorem~\ref{thm:variance} avoids the computational burden of simulation-based variance calculations. To account for the heterogeneous scales of clinical and cost outcomes, we introduce the \textit{standardized ceiling ratio} $\lambda r$, which adjusts willingness-to-pay by the relative variability of the two outcomes, ensuring that both outcomes contribute appropriately to the optimal design. Second, we develop two design optimization strategies for L-CRTs: LODs that maximize the power of the optimal design when ICC parameters are known, with closed-form solutions available for CRXO and PA-LCRT designs (Sections~\ref{sec:crxo_LOD} and~\ref{sec:pa_LOD}), and MMDs that ensure robust performance across a pre-specified parameter space when ICC parameters are uncertain. Across all three design variants, MMDs maintain relative efficiencies exceeding 80\% (Tables~\ref{tab:MMD} and~\ref{tab:sw_mmd}), suggesting that investigators can adopt reasonably robust designs without severe efficiency loss under the range of parameter uncertainty considered in our numerical studies. Third, we find that the three L-CRT design variants have distinct performance characteristics under cost-effectiveness objectives. LODs for CRXO trials consistently achieve the highest statistical power because within-cluster contrasts across periods allow each cluster to serve as its own control, substantially reducing variance compared to PA-LCRTs that rely only on between-cluster comparisons; see Equations~\eqref{eqn:CRXO_var} and~\eqref{eqn:pa_var} for the CRXO and PA-LCRT variance expressions, respectively. LODs for SW-CRTs yield the lowest power among the three designs under the same setting (Figure~\ref{fig:crxo_lod_J4}), although SW-CRTs offer practical advantages when complete intervention rollout is required prior to study completion or when the intervention is perceived to be beneficial by stakeholders \citep{Hemming2020}. Additionally, the optimal number of periods for SW-CRTs under cost-effectiveness objectives is typically $J = Q + 1$, indicating that resources should be concentrated on additional clusters rather than extended follow-up, consistent with results from \citet{Liu2024} and \citet{Liu2025} for univariate clinical outcomes.

To illustrate the proposed methods, we apply them to the context of the Australian Reinvestment Trial, where the main findings are consistent with the theoretical results and numerical studies presented above. Under the same budget constraint, CRXO trials achieved the highest power of 0.996 with $(I_{\text{LOD}}, K_{\text{LOD}}) = (8, 36)$, followed by PA-LCRTs with power of 0.893 and $(I_{\text{LOD}}, K_{\text{LOD}}) = (66, 3)$, and SW-CRTs with power of 0.833 and $(I_{\text{LOD}}, K_{\text{LOD}}) = (35, 7)$. The MMDs for all three designs retained the same optimal sample size combination as their respective LODs, with RE $= 0.991$, $0.990$, and $0.979$ for CRXO, PA-LCRT, and SW-CRT designs, respectively. These results indicate that investigators can adopt robust designs with minimal efficiency loss even under substantial between-outcome ICC uncertainty. We further extended the sample size methodology to accommodate incomplete SW-CRT designs, which achieved higher power than complete designs (i.e., 0.866 versus 0.833 for $J = 8$) by concentrating observations in periods with greater information content for estimating the treatment effect \citep{Li2023information}. Consistent with the complete design results, the minimum number of periods $J = Q + 1$ also achieves the highest power for incomplete SW-CRTs. The incomplete design MMD allocated $(I_{\text{MMD}}, K_{\text{MMD}}) = (7, 41)$ with RE $= 0.992$, which exceeds the RE $= 0.979$ obtained for the complete design MMD, suggesting that strategic use of incomplete designs can improve both efficiency and robustness to parameter uncertainty in cost-effectiveness L-CRTs.

We offer several practical recommendations for investigators planning cost-effectiveness L-CRTs. When design flexibility permits and the primary objective is maximizing the power of the design, CRXO trials are preferred due to their superior efficiency. However, CRXO trials require that the intervention effect is reversible and that carryover effects are negligible; when these assumptions do not hold, PA-LCRTs or SW-CRTs may be more appropriate. For SW-CRTs, investigators should consider designs with the minimum number of periods required ($J = Q + 1$) rather than extended follow-up, as additional periods do not yield proportional gains in the power of the optimal design. When ICC parameters are uncertain at the design stage, MMDs provide a robust alternative to LODs; the high relative efficiencies observed across all design variants suggest that MMDs are a pragmatic choice for most applications. Careful consideration should also be given to the standardized ceiling ratio $\lambda r$, as it determines the relative contribution of clinical and cost outcomes to the optimization. When prior estimates of outcome variability are available, investigators can compute $\lambda r$ to ensure that the chosen design appropriately balances both outcomes. To facilitate the implementation of these methods, we provide a user-friendly \texttt{R Shiny} application available at \href{https://f07k8s-hao-wang.shinyapps.io/Cost-effectiveness_LCRT/}{https://f07k8s-hao-wang.shinyapps.io/Cost-effectiveness\_LCRT/} that automates LOD and MMD calculations for all three design variants given user-specified parameters. We provide a tutorial in Appendix~\ref{supp_sec:tutorial_shiny} that walks investigators through the application interface, from specifying the required input parameters to obtaining the optimization results, with illustrative examples of both the LOD and MMD for CRXOs based on the Australian reinvestment trial data. We note that our variance expressions and optimization algorithms require specification of seven ICC parameters. While ICCs for clinical outcomes are increasingly reported in CRT publications, ICCs for cost outcomes remain rarely documented, and between-outcome ICCs ($\rho_0^{EC}$, $\rho_1^{EC}$, $\rho_2^{EC}$) are largely absent from the literature. This highlights the need for greater transparency in reporting ICCs for cost outcomes and between-outcome correlations in future L-CRTs, which would facilitate more accurate sample size calculations for cost-effectiveness objectives.

Although we have attempted to provide a unified framework for designing cost-effectiveness L-CRTs, the current study has several limitations that warrant discussion. First, our budget constraint assumes identical cluster-level costs $c_1$ and individual-level costs $c_2$ across treatment and control arms. In contrast, \citet{Manju2014} allowed separate cost parameters for each arm, recognizing that intervention delivery may incur different expenses than usual care. This distinction can be relevant when the intervention requires specialized equipment, additional personnel training, or more intensive monitoring, and future extensions could incorporate arm-specific cost structures to enable more realistic budget allocation. Second, we assume equal cluster-period sizes in this article, which may not always hold in practice when cluster sizes are heterogeneous or when recruitment resources vary across periods. Extending our framework to accommodate heterogeneous cluster-period sizes can be a direction for future research. Third, we assume a constant treatment effect throughout the study duration. However, treatment effects in L-CRTs may vary over time depending on exposure time \citep{Kenny2022}, calendar time \citep{Lee2024}, or both \citep{Wang2024}. In cost-effectiveness L-CRTs, this complexity is further compounded by the potential for clinical and cost outcomes to follow different temporal trajectories; for example, clinical benefits may accumulate gradually with prolonged exposure while cost savings materialize immediately, or vice versa. Developing optimal design strategies that account for time-varying treatment effect structures across both outcomes can be studied in the future. Fourth, we have considered only continuous outcomes, and our current results can at most provide an approximation for binary outcomes when the effect measure is the risk difference. When $E_{ijk}$ is binary and the link function in the linear mixed model is the canonical logit link, the variance of the treatment effect estimator cannot be obtained in simple analytical forms. Extending our framework to accommodate binary outcomes with odds ratio or risk ratio effect measures would require simulation-based variance estimation or numerical approximations, which we leave to future work.

\section*{Acknowledgments}
Research in this article was supported by a Transformational Pilot Funding grant awarded by the Yale School of Public Health. F. Li and D. Meeker are also supported by USC-Yale Roybal Center for Behavioral Interventions in Aging, which is funded by the National Institute on Aging under the National Institutes of Health (award number P30AG024968). The content is solely the responsibility of the authors and does not necessarily represent the official views of the National Institutes of Health. The authors also thank Joseph S. Ross for his helpful comments and suggestions.

\section*{Data availability statement}
The data from the Australia disinvestment and reinvestment trial are publicly available at \href{https://journals.plos.org/plosmedicine/article?id=10.1371/journal.pmed.1002412}{https://doi.org/10.1371/journal.pmed.1002412}.

\section*{Supporting information}

Additional supporting information, including Web Appendices A--E, Web Tables, and Figures, may be found online in the supporting information tab for this article. All R code for the simulation study and trial planning software is publicly available at \href{https://github.com/haowangbiostat/Cost-effectiveness_LCRT}{https://github.com/haowangbiostat/Cost-effectiveness\_LCRT}.

{
\bibliographystyle{informs2014}
\bibliography{reference}
}

\newpage
\setcounter{page}{1}
\begin{APPENDICES}
    \section{Average INMB Estimator} \label{supp_sec:INMB}

\subsection{Eigenvalues of $\mathbf{R}_i$} \label{supp_sec:eigen}

We model the effect and cost data using the following bivariate linear mixed model with a nested exchangeable correlation structure \citep{Davis-Plourde2023}:
\begin{align}\label{supp_eqn:model}
     E_{ijk} = \alpha_{0j} + \alpha_1Z_{ij} + b_i^E + s_{ij}^E + \epsilon_{ijk}^E, \quad C_{ijk} = \gamma_{0j} + \gamma_1Z_{ij} + b_i^C + s_{ij}^C + \epsilon_{ijk}^C,
\end{align}
where $\bm{b}_i = (b_i^E, b_i^C)^\prime \sim \calN(\bm{0}_{2 \times 1}, \bfSigma_b)$, $\bm{s}_{ij} = (s_{ij}^E, s_{ij}^C)^\prime \sim \calN(\bm{0}_{2 \times 1}, \bfSigma_s)$, and $\bm{\epsilon}_{ijk} = (\epsilon_{ijk}^E, \epsilon_{ijk}^C)^\prime \sim \calN(\bm{0}_{2 \times 1}, \bfSigma_e)$ with
\begin{align*}
    \bfSigma_b = \left(\begin{array}{cc}
        \sigma_{bE}^2 & \sigma_{bEC}\\
        \sigma_{bEC} & \sigma_{bC}^2
    \end{array}\right), 
    \quad \bfSigma_s = \left(\begin{array}{cc}
        \sigma_{sE}^2 & \sigma_{sEC}\\
        \sigma_{sEC} & \sigma_{sC}^2
    \end{array}\right), 
    \quad \bfSigma_e = \left(\begin{array}{cc}
        \sigma_{\epsilon E}^2 & \sigma_{\epsilon EC}\\
        \sigma_{\epsilon EC} & \sigma_{\epsilon C}^2
    \end{array}\right).
\end{align*}
Model~\eqref{supp_eqn:model} induces a linear mixed model for net monetary benefit (NMB) given by 
\begin{align*}
    \text{NMB}_{ijk} = \lambda E_{ijk} - C_{ijk} = (\lambda\alpha_{0j}-\gamma_{0j}) + \beta_1 Z_{ij} + (\lambda b_i^E - b_i^C) + (\lambda s_{ij}^E - s_{ij}^C) + (\lambda\epsilon_{ijk}^E-\epsilon_{ijk}^C),
\end{align*}
where $\beta_1 = \lambda\alpha_1 - \gamma_1$ is the average incremental net monetary benefit (INMB).

In this section, we derive the eigenvalues of the individual-level correlation matrix $\mathbf{R}_i$
based on the nested exchangeable structure in~\eqref{supp_eqn:model}.  For cluster $i$, we stack the bivariate outcomes as
\begin{align*}
    \mathbf{Y}_i = (E_{i11}, C_{i11}, \ldots, E_{i1K}, C_{i1K}, E_{i21}, C_{i21}, \ldots, E_{i2K}, C_{i2K}, \ldots, E_{iJ1}, C_{iJ1}, \ldots, E_{iJK}, C_{iJK})^\prime,
\end{align*}
where $\mathbf{Y}_i$ has dimension $2JK\times 1$.  Let $\mathbf{R}_i$ denote the corresponding
$2JK\times 2JK$ correlation matrix. Under the cross-sectional design, correlations between observations in the same cluster are determined by three components: (i) the within-individual correlation $\bfGamma_2$ for the same individual; (ii) the within-period correlation $\bfGamma_0$ for different individuals in the same period; and (iii) the between-period correlation $\bfGamma_1$ for individuals in different periods, where
\begin{align*}
    \bfGamma_0 = \left(\begin{array}{cc}
        \rho_0^E & \rho_0^{EC} \\
        \rho_0^{EC} & \rho_0^C
    \end{array}\right), \quad
    \bfGamma_1 = \left(\begin{array}{cc}
        \rho_1^E & \rho_1^{EC} \\
        \rho_1^{EC} & \rho_1^C
    \end{array}\right), \quad
    \bfGamma_2 = \left(\begin{array}{cc}
        1 & \rho_2^{EC} \\
        \rho_2^{EC} & 1
    \end{array}\right).
\end{align*}
Using these ICC matrices, the correlation matrix $\mathbf{R}_i$ can be expressed as
\begin{align*} 
    \mathbf{R}_i = \mathbf{I}_J \otimes \{\mathbf{I}_K\otimes \bfGamma_2 + (\mathbf{1}_K\mathbf{1}_K^\prime - \mathbf{I}_K)\otimes \bfGamma_0\} + (\mathbf{1}_J\mathbf{1}_J^\prime - \mathbf{I}_J)\otimes (\mathbf{1}_K\mathbf{1}_K^\prime \otimes \bfGamma_1),
\end{align*} 
where $\mathbf{I}_J$ and $\mathbf{I}_K$ are identity matrices of dimension $J$ and $K$, respectively, $\mathbf{1}_J$ and $\mathbf{1}_K$ are vectors of ones of length $J$ and $K$, respectively, and $\otimes$ denotes the Kronecker product.

\begin{lemma} \label{supp_lemma:eigenvalues}
    The eigenvalues of $\mathbf{R}_i$ are determined by the eigenvalues of three $2\times 2$ matrices $\mathbf{M}_1$, $\mathbf{M}_2$, and $\mathbf{M}_3$, where
    \begin{align*}
        \mathbf{M}_1 &= \bfGamma_2 + (K-1)\bfGamma_0 + (J-1)K\bfGamma_1, \displaybreak[0]\\
        \mathbf{M}_2 &= \bfGamma_2 + (K-1)\bfGamma_0 - K\bfGamma_1, \displaybreak[0]\\
        \mathbf{M}_3 &= \bfGamma_2 - \bfGamma_0,
    \end{align*}
    where the eigenvalues of $\mathbf{M}_1$ have multiplicity $1$, the eigenvalues of $\mathbf{M}_2$ have multiplicity $J-1$, and the eigenvalues of $\mathbf{M}_3$ have multiplicity $J(K-1)$.
\end{lemma}

\begin{proof}
    We introduce the projection matrices $\mathbf{P}_J = J^{-1}\mathbf{1}_J\mathbf{1}_J^\prime$ and $\mathbf{Q}_J = \mathbf{I}_J - \mathbf{P}_J$, and similarly define $\mathbf{P}_K = K^{-1}\mathbf{1}_K\mathbf{1}_K^\prime$ and $\mathbf{Q}_K = \mathbf{I}_K - \mathbf{P}_K$. These projection matrices satisfy the following properties: (i) $\mathbf{P}_J\mathbf{Q}_J = \mathbf{Q}_J\mathbf{P}_J = \mathbf{0}_{J \times J}$; (ii) $\mathbf{P}_J\mathbf{P}_J = \mathbf{P}_J$ and $\mathbf{Q}_J\mathbf{Q}_J = \mathbf{Q}_J$; (iii) $\mathbf{P}_J + \mathbf{Q}_J = \mathbf{I}_J$. The same properties hold for $\mathbf{P}_K$ and $\mathbf{Q}_K$. Note that $\mathbf{1}_J\mathbf{1}_J^\prime = J\mathbf{P}_J$ and $\mathbf{1}_K\mathbf{1}_K^\prime = K\mathbf{P}_K$. Then, we can re-express $\mathbf{R}_i$ as
    \begin{align*}
        \mathbf{R}_i &= \mathbf{I}_J \otimes \{\mathbf{I}_K\otimes \bfGamma_2 + (\mathbf{1}_K\mathbf{1}_K^\prime - \mathbf{I}_K)\otimes \bfGamma_0\} + (\mathbf{1}_J\mathbf{1}_J^\prime - \mathbf{I}_J)\otimes (\mathbf{1}_K\mathbf{1}_K^\prime \otimes \bfGamma_1) \displaybreak[0]\\
        &= \mathbf{I}_J \otimes \{\mathbf{P}_K\otimes (\bfGamma_2 + (K-1)\bfGamma_0) + \mathbf{Q}_K\otimes (\bfGamma_2 - \bfGamma_0)\} \displaybreak[0]\\
        &\quad + (J-1)K\mathbf{P}_J\otimes \mathbf{P}_K \otimes \bfGamma_1 - K\mathbf{Q}_J\otimes \mathbf{P}_K \otimes \bfGamma_1 \displaybreak[0]\\
        &= \mathbf{P}_J \otimes \mathbf{P}_K \otimes \mathbf{M}_1 + \mathbf{Q}_J \otimes \mathbf{P}_K \otimes \mathbf{M}_2 + \mathbf{I}_J \otimes \mathbf{Q}_K \otimes \mathbf{M}_3,
    \end{align*}
    where
    \begin{align*}
        \mathbf{M}_1 &= \bfGamma_2 + (K-1)\bfGamma_0 + (J-1)K\bfGamma_1 \displaybreak[0]\\
        &= \left(\begin{array}{cc}
            1 + (K-1)\rho_0^E + (J-1)K\rho_1^E & \rho_2^{EC} + (K-1)\rho_0^{EC} + (J-1)K\rho_1^{EC} \\
            \rho_2^{EC} + (K-1)\rho_0^{EC} + (J-1)K\rho_1^{EC} & 1 + (K-1)\rho_0^C + (J-1)K\rho_1^C
        \end{array}\right), \displaybreak[0]\\
        \mathbf{M}_2 &= \bfGamma_2 + (K-1)\bfGamma_0 - K\bfGamma_1 = \left(\begin{array}{cc}
            1 + (K - 1)\rho_0^E - K\rho_1^E & \rho_2^{EC} + (K-1)\rho_0^{EC} - K\rho_1^{EC} \\
            \rho_2^{EC} + (K-1)\rho_0^{EC} - K\rho_1^{EC} & 1 + (K - 1)\rho_0^C - K\rho_1^C
        \end{array}\right), \displaybreak[0]\\
        \mathbf{M}_3 &= \bfGamma_2 - \bfGamma_0 = \left(\begin{array}{cc}
            1 - \rho_0^E & \rho_2^{EC} - \rho_0^{EC} \\
            \rho_2^{EC} - \rho_0^{EC} & 1 - \rho_0^C
        \end{array}\right).
    \end{align*}
    By Theorem 4.2.12 in \citet{Horn1991}, if $\mathbf{A}$ and $\mathbf{B}$ are square matrices of dimensions $n$ and $m$ with eigenvalues $\{\lambda_1, \ldots, \lambda_n\}$ and $\{\mu_1, \ldots, \mu_m\}$ (listed according to multiplicity), respectively, then the eigenvalues of $\mathbf{A} \otimes \mathbf{B}$ are $\{\lambda_i\mu_j: i = 1, \ldots, n,\, j = 1, \ldots, m\}$. The projection matrix $\mathbf{P}_J$ has eigenvalue $1$ with multiplicity $1$ and eigenvalue $0$ with multiplicity $J-1$. Furthermore, $\mathbf{Q}_J$ has eigenvalue $1$ with multiplicity $J-1$ and eigenvalue $0$ with multiplicity $1$. The same structure applies to $\mathbf{P}_K$ and $\mathbf{Q}_K$.

    We now determine the eigenvalues of $\mathbf{R}_i$ by analyzing each term in its decomposition. First, we note that the Kronecker products of projection matrices in our decomposition act on orthogonal subspaces. To see this, observe that
    \begin{align*}
        (\mathbf{P}_J \otimes \mathbf{P}_K)(\mathbf{Q}_J \otimes \mathbf{P}_K) = (\mathbf{P}_J\mathbf{Q}_J) \otimes (\mathbf{P}_K\mathbf{P}_K) = \mathbf{0}_{J \times J} \otimes \mathbf{P}_K = \mathbf{0}_{JK \times JK},
    \end{align*}
    and similarly for other pairs. This orthogonality implies that when $\mathbf{R}_i$ acts on a vector in one subspace, only the corresponding term contributes while the other terms vanish. For the first term $\mathbf{P}_J \otimes \mathbf{P}_K \otimes \mathbf{M}_1$, because $\mathbf{P}_J$ and $\mathbf{P}_K$ each have eigenvalue $1$ with multiplicity $1$, the Kronecker product $\mathbf{P}_J \otimes \mathbf{P}_K$ has eigenvalue 1 with multiplicity 1, and eigenvalue $0$ with multiplicity $JK - 1$. Therefore, on the one-dimensional subspace where $\mathbf{P}_J \otimes \mathbf{P}_K$ has eigenvalue $1$, the matrix $\mathbf{R}_i$ acts as $\mathbf{M}_1$, contributing the two eigenvalues of $\mathbf{M}_1$, each with multiplicity $1$. For the second term $\mathbf{Q}_J \otimes \mathbf{P}_K \otimes \mathbf{M}_2$, the matrix $\mathbf{Q}_J$ has eigenvalue $1$ with multiplicity $J-1$, and $\mathbf{P}_K$ has eigenvalue $1$ with multiplicity $1$. Thus, $\mathbf{Q}_J \otimes \mathbf{P}_K$ has eigenvalue $1$ with multiplicity $J-1$. On this $(J-1)$-dimensional subspace, $\mathbf{R}_i$ acts as $\mathbf{M}_2$, contributing the two eigenvalues of $\mathbf{M}_2$, each with multiplicity $J-1$. For the third term $\mathbf{I}_J \otimes \mathbf{Q}_K \otimes \mathbf{M}_3$, the matrix $\mathbf{I}_J$ has eigenvalue $1$ with multiplicity $J$, and $\mathbf{Q}_K$ has eigenvalue $1$ with multiplicity $K-1$. Thus, $\mathbf{I}_J \otimes \mathbf{Q}_K$ has eigenvalue $1$ with multiplicity $J(K-1)$. On this $J(K-1)$-dimensional subspace, $\mathbf{R}_i$ acts as $\mathbf{M}_3$, contributing the two eigenvalues of $\mathbf{M}_3$, each with multiplicity $J(K-1)$. Of note, the total number of eigenvalues is $2 \times 1 + 2 \times (J-1) + 2 \times J(K-1) = 2JK$, which equals the dimension of $\mathbf{R}_i$.
\end{proof}

\begin{theorem} \label{supp_thm:eigenvalues}
    Under Model \eqref{supp_eqn:model}, the unique eigenvalues of $\mathbf{R}_i$ are given by $\{\lambda_1^+, \lambda_1^-, \lambda_2^+, \lambda_2^-, \lambda_3^+, \lambda_3^-\}$ with the following explicit forms:
    \begin{align*}
        \lambda_1^\pm &= \frac{1}{2}\{2 + (K-1)(\rho_0^E + \rho_0^C) + (J-1)K(\rho_1^E + \rho_1^C)\} \pm \frac{1}{2}\sqrt{\xi_1}, \displaybreak[0]\\
        \lambda_2^\pm &= \frac{1}{2}(\kappa^E + \kappa^C) \pm \frac{1}{2}\sqrt{(\kappa^E - \kappa^C)^2 + 4(\kappa^{EC})^2}, \displaybreak[0]\\
        \lambda_3^\pm &= \frac{1}{2}(2 - \rho_0^E - \rho_0^C) \pm \frac{1}{2}\sqrt{(\rho_0^E - \rho_0^C)^2 + 4(\rho_2^{EC} - \rho_0^{EC})^2},
    \end{align*}
    where $\xi_1 = \{(K-1)(\rho_0^E - \rho_0^C) + (J-1)K(\rho_1^E - \rho_1^C)\}^2 + 4\{\rho_2^{EC} + (K-1)\rho_0^{EC} + (J-1)K\rho_1^{EC}\}^2$, $\kappa^E = 1 + (K - 1)\rho_0^E - K\rho_1^E$, $\kappa^C = 1 + (K - 1)\rho_0^C - K\rho_1^C$, and $\kappa^{EC} = \rho_2^{EC} + (K - 1)\rho_0^{EC} - K\rho_1^{EC}$. The eigenvalues $\lambda_1^\pm$ each have multiplicity $1$, $\lambda_2^\pm$ each have multiplicity $J-1$, and $\lambda_3^\pm$ each have multiplicity $J(K-1)$.
\end{theorem}

\begin{proof}
    By Lemma~\ref{supp_lemma:eigenvalues}, the eigenvalues of $\mathbf{R}_i$ are determined by the three $2 \times 2$ matrices $\mathbf{M}_1$, $\mathbf{M}_2$, and $\mathbf{M}_3$. For a symmetric $2 \times 2$ matrix of the form
    \begin{align*}
        \mathbf{M} = \left(\begin{array}{cc}
            a & b \\
            b & c
        \end{array}\right),
    \end{align*}
    the eigenvalues are $(a + c)/2 \pm \sqrt{b^2 + (a-c)^2/4}$. Applying this formula to each matrix leads to the stated results.
\end{proof}

\subsection{Closed-Form Variance Expression of The Average INMB Estimator} \label{supp_sec:var}

\begin{theorem} \label{supp_thm:var}
    Define $U = \sumi\sumj Z_{ij}$, $V = \sumi(\sumj Z_{ij})^2$, $W=\sumj(\sumi Z_{ij})^2$ as study design constants that depend on the treatment status, the variance of average INMB estimator $\hat\beta_1$ is
    \begin{align*}
        \bbV(\hat\beta_1) = \frac{IJ\sigma_E\sigma_C}{\eta}\left\{\frac{\varphi}{K}\left(\frac{\lambda^2\kappa^E\sigma_E}{\sigma_C}-2\lambda\kappa^{EC} + \frac{\kappa^C\sigma_C}{\sigma_E}\right) - \frac{J}{\Delta^\ast}(U^2-IV)\left(\frac{\lambda^2\rho_1^E\sigma_E}{\sigma_C} - 2\lambda\rho_1^{EC} + \frac{\rho_1^C\sigma_C}{\sigma_E}\right)\right\},
    \end{align*} 
    where $\eta = \varphi J(IU - W) + J^2{\Delta^\ast}^{-1}(U^2 - IV)\sigma_E^2\sigma_C^2(\rho_1^E\rho_1^C - {\rho_1^{EC}}^2)\{\varphi + {\Delta^\ast}^{-1}(U^2 - IV)\}$, $\varphi = J(IU - W)\Delta^{-1} + (U^2 - IV)(\Delta^{-1} - {\Delta^\ast}^{-1})$, $\Delta^\ast = \Delta + J\sigma_E^2\sigma_C^2\{K^{-1}(\kappa^E\rho_1^C + \kappa^C\rho_1^E - 2\kappa^{EC}\rho_1^{EC}) + J(\rho_1^E\rho_1^C - {\rho_1^{EC}}^2)\}$, $\Delta = K^{-2}\sigma_E^2\sigma_C^2(\kappa^E\kappa^C - {\kappa^{EC}}^2)$, and $\kappa^E = 1 + (K - 1)\rho_0^E - K\rho_1^E$, $\kappa^C = 1 + (K - 1)\rho_0^C - K\rho_1^C$ are eigenvalues of a nested exchangeable correlation structure for effect and cost, respectively, and $\kappa^{EC} = \rho_2^{EC} + (K - 1)\rho_0^{EC} - K\rho_1^{EC}$ carries the impact of three between-outcome ICCs.
\end{theorem}

\begin{proof}

To facilitate the derivation, we rewrite~\eqref{supp_eqn:model} in terms of cluster-period means, so that they become
\begin{align*}
    \bar E_{ij} = \alpha_{0j} + \alpha_1Z_{ij} + b_i^E + s_{ij}^E + \bar \epsilon_{ij}^E, \quad \bar C_{ij} = \gamma_{0j} + \gamma_1Z_{ij} + b_i^C + s_{ij}^C + \bar \epsilon_{ij}^C.
\end{align*}
where $\bar E_{ij} = K^{-1}\sumk E_{ijk}$, $\bar C_{ij} = K^{-1}\sumk C_{ijk}$, $\bar \epsilon_{ij}^E = K^{-1}\sumk \epsilon_{ijk}^E$, and $\bar \epsilon_{ij}^C = K^{-1}\sumk \epsilon_{ijk}^C$. Based on standard results in mixed model theory \citep{Hussey2007}, the covariance matrix $\bfSigma_{\hat\alpha_1, \hat\gamma_1}$ is the lower $2 \times 2$ matrix of $(\sumi \mathbf{D}_i^\prime\mathbf{G}_i^{-1}\mathbf{D}_i)^{-1}$, where $\mathbf{D}_i$ is the $2J \times (2J+2)$ design matrix for cluster $i$, and $\mathbf{G}_i$ is the $2J \times 2J$ covariance matrix for the two cluster-period mean outcomes. We further note that $\bbV(\bar E_{ij}) = \sigma_{bE}^2 + \sigma_{sE}^2 + \sigma_{\epsilon E}^2/K$, $\bbV(\bar C_{ij}) = \sigma_{bC}^2 + \sigma_{sC}^2 + \sigma_{\epsilon C}^2/K$, and $\Cov(\bar E_{ij}, \bar C_{ij}) = \sigma_{bEC}^2 + \sigma_{sEC}^2 + \sigma_{\epsilon EC}^2/K$ for within-period components; $\Cov(\bar E_{ij}, \bar E_{ij^\prime}) = \sigma_{bE}^2$, $\Cov(\bar C_{ij}, \bar C_{ij^\prime}) = \sigma_{bC}^2$, and $\Cov(\bar E_{ij}, \bar C_{ij}) = \sigma_{bEC}^2$ for between-period components. Therefore, $\mathbf{G}_i$ can be written as
\begin{align*}
    \mathbf{G}_i = \mathbf{I}_J \otimes \left(\bm{\Sigma}_s + \frac{1}{K}\bm{\Sigma}_\epsilon\right) + \bm{1}_J^\prime\bm{1}_J \otimes \bm{\Sigma}_b.
\end{align*}
Then, we follow \citet{Davis-Plourde2023} and write the inverse of $\mathbf{V}$ as
\begin{align*}
    \mathbf{G}_i^{-1} = \mathbf{I}_J \otimes \left(\bm{\Sigma}_s + \frac{1}{K}\bm{\Sigma}_\epsilon\right)^{-1} + \bm{1}_J^\prime\bm{1}_J \otimes \frac{1}{J} \left\{\left(J\bm{\Sigma}_b + \bm{\Sigma}_s + \frac{1}{K}\bm{\Sigma}_\epsilon\right)^{-1} - \left(\bm{\Sigma}_s + \frac{1}{K}\bm{\Sigma}_\epsilon\right)^{-1}\right\},
\end{align*}
which leads to
\begin{align}
    \bfSigma_{\hat\alpha_1, \hat\gamma_1} = IJ \left\{(IJU - JW + U^2 - IV)\left(\bm{\Sigma}_s + \frac{1}{K}\bm{\Sigma}_\epsilon\right)^{-1} - (U^2 - IV)\left(J\bm{\Sigma}_b + \bm{\Sigma}_s + \frac{1}{K}\bm{\Sigma}_\epsilon\right)^{-1}\right\}^{-1}. \label{supp_eqn:covariance_v0}
\end{align}
Here, $U = \sumi\sumj Z_{ij}$, $V = \sumi(\sumj Z_{ij})^2$, and $W = \sumj(\sumi Z_{ij})^2$. We further define the following intra-cluster correlation coefficients (ICC) matrices:
\begin{align*}
    \bfGamma_0 = \left(\begin{array}{cc}
        \rho_0^E & \rho_0^{EC} \\
        \rho_0^{EC} & \rho_0^C
    \end{array}\right), 
    \quad \bfGamma_1 = \left(\begin{array}{cc}
        \rho_1^E & \rho_1^{EC} \\
        \rho_1^{EC} & \rho_1^C
    \end{array}\right), 
    \quad \bfGamma_2 = \left(\begin{array}{cc}
        1 & \rho_2^{EC}\\
        \rho_2^{EC} & 1
    \end{array}\right),
\end{align*}
where $\bfGamma_0, \bfGamma_1, \bfGamma_2$ are the within-period ICC matrix, between-period ICC matrix, and within-individual ICC matrix, respectively. Given that $\sigma_{bE}^2 = \rho_1^E\sigma_E^2$, $\sigma_{bC}^2 = \rho_1^C\sigma_C^2$, $\sigma_{bEC} = \rho_1^{EC}\sigma_E\sigma_C$, $\sigma_{sE}^2 = (\rho_0^E - \rho_1^E)\sigma_E^2$, $\sigma_{sC}^2 = (\rho_0^C - \rho_1^C)\sigma_C^2$, $\sigma_{sEC}^2 = (\rho_0^{EC} - \rho_1^{EC})\sigma_E\sigma_C$, $\sigma_{\epsilon E}^2 = \sigma_E^2(1 - \rho_0^E)$, $\sigma_{\epsilon C}^2 = \sigma_C^2(1 - \rho_0^C)$, $\sigma_{\epsilon EC} = (\rho_2^{EC} - \rho_0^{EC})\sigma_E\sigma_C$, we can also rewrite the covariance matrix of the random effects in terms of the ICCs such that $\bfSigma_b = \bfLambda_y^{1/2}\bfGamma_1\bfLambda_y^{1/2}$, $\bfSigma_s = \bfLambda_y^{1/2}(\bfGamma_0 - \bfGamma_1)\bfLambda_y^{1/2}$, and $\bfSigma_\epsilon = \bfLambda_y^{1/2}(\bfGamma_2 - \bfGamma_0)\bfLambda_y^{1/2}$, where $\bfLambda_y = \diag(\sigma_E^2, \sigma_C^2)$. Then, we proceed to re-express~\eqref{supp_eqn:covariance_v0} in terms of ICC parameters:
\begin{align*}
    &\quad \bfSigma_{\hat\alpha_1, \hat\gamma_1} \displaybreak[0]\\
    &\resizebox{\linewidth}{!}{$= \frac{IJ}{K} \bfLambda_y^{1/2} \left[(IJU - JW + U^2 - IV)\left\{\bfGamma_2 - K\bfGamma_1 + (K - 1)\bfGamma_0\right\}^{-1} - (U^2 - IV)\{\bfGamma_2 + (J - 1)K\bfGamma_1 + (K - 1)\bfGamma_0\}^{-1} \right]^{-1} \bfLambda_y^{1/2}$} \displaybreak[0]\\
    &= \frac{IJ}{K} \left(\begin{array}{cc}
       \sigma_E & 0 \\
        0 & \sigma_C
    \end{array}\right) \mathbf{M}^{-1} \left(\begin{array}{cc}
       \sigma_E & 0 \\
        0 & \sigma_C
    \end{array}\right),
\end{align*}
where
\begin{align*}
    \mathbf{M} = \left\{(IJU - JW + U^2 - IV)\left(\begin{array}{cc}
        \kappa^E & \kappa^{EC} \\
        \kappa^{EC} & \kappa^C
    \end{array}\right)^{-1} - (U^2 - IV)\left(\begin{array}{cc}
        \kappa^E + JK\rho_1^E & \kappa^{EC} + JK\rho_1^{EC} \\
        \kappa^{EC} + JK\rho_1^{EC} & \kappa^C + JK\rho_1^C
    \end{array}\right)^{-1}\right\}.
\end{align*}
Here, $\kappa^E = 1 + (K - 1)\rho_0^E - K\rho_1^E$ and $\kappa^C = 1 + (K - 1)\rho_0^C - K\rho_1^C$ are eigenvalues of a nested exchangeable correlation structure for effect and cost outcomes, respectively \citep{Li2021overview}. Moreover, $\kappa^{EC} = \rho_2^{EC} + (K-1)\rho_0^{EC} - K\rho_1^{EC}$ is consist of three effect-cost ICCs. We further note that
\begin{align*}
    \frac{K^2}{\sigma_E^2\sigma_C^2}\mathbf{M} &= \frac{IJU - JW + U^2 - IV}{\Delta}\left(\begin{array}{cc}
        \kappa^C & -\kappa^{EC} \\
        -\kappa^{EC} & \kappa^E
    \end{array}\right) - \frac{U^2 - IV}{\Delta^\ast}\left(\begin{array}{cc}
        \kappa^C + JK\rho_1^C & -\kappa^{EC} - JK\rho_1^{EC} \\
        -\kappa^{EC} - JK\rho_1^{EC} & \kappa^E + JK\rho_1^E
    \end{array}\right) \displaybreak[0]\\
    &= \left(\frac{IJU - JW + U^2 - IV}{\Delta} - \frac{U^2 - IV}{\Delta^\ast}\right)\left(\begin{array}{cc}
        \kappa^C & -\kappa^{EC} \\
        -\kappa^{EC} & \kappa^E
    \end{array}\right) - \frac{JK(U^2 - IV)}{\Delta^\ast}\left(\begin{array}{cc}
        \rho_1^C & -\rho_1^{EC} \\
        -\rho_1^{EC} & \rho_1^E
    \end{array}\right) \displaybreak[0]\\
    &= \left(\begin{array}{cc}
        \varphi\kappa^C - \frac{JK(U^2 - IV)\rho_1^C}{\Delta^\ast} & -\varphi\kappa^{EC} + \frac{JK(U^2 - IV)\rho_1^{EC}}{\Delta^\ast} \\
        -\varphi\kappa^{EC} + \frac{JK(U^2 - IV)\rho_1^{EC}}{\Delta^\ast} & \varphi\kappa^E - \frac{JK(U^2 - IV)\rho_1^E}{\Delta^\ast}
    \end{array}\right),
\end{align*}
where
\begin{align*}
    \Delta &= K^{-2}\sigma_E^2\sigma_C^2\{\kappa^E\kappa^C - (\kappa^{EC})^2\}, \displaybreak[0]\\
    \Delta^\ast &= \Delta + J\sigma_E^2\sigma_C^2[K^{-1}(\kappa^E\rho_1^C + \kappa^C\rho_1^E - 2\kappa^{EC}\rho_1^{EC}) + J\{\rho_1^E\rho_1^C - (\rho_1^{EC})^2\}], \displaybreak[0]\\
    \varphi &= J(IU - W)\frac{1}{\Delta} + (U^2 - IV)\left(\frac{1}{\Delta} - \frac{1}{\Delta^\ast}\right).
\end{align*}
We let $\det(\cdot)$ be the determinant of a matrix. For clarity, we also let $\eta = \varphi J(IU - W) + J^2(\Delta^\ast)^{-1}(U^2 - IV)\sigma_E^2\sigma_C^2\{\rho_1^E\rho_1^C - (\rho_1^{EC})^2\}\{\varphi + (\Delta^\ast)^{-1}(U^2 - IV)\}$. Then,
\begin{align*}
    &\quad\det\left(\frac{K^2}{\sigma_E^2\sigma_C^2}\mathbf{M}\right) \displaybreak[0] \\
    &= \left(\varphi\kappa^C - \frac{JK(U^2 - IV)\rho_1^C}{\Delta^\ast}\right)\left(\varphi\kappa^E - \frac{JK(U^2 - IV)\rho_1^E}{\Delta^\ast}\right) - \left(-\varphi\kappa^{EC} + \frac{JK(U^2 - IV)\rho_1^{EC}}{\Delta^\ast}\right)^2 \displaybreak[0] \\
    &= \varphi^2\kappa^C\kappa^E - \frac{JK(U^2 - IV)\rho_1^E\varphi\kappa^C}{\Delta^\ast} - \frac{JK(U^2 - IV)\rho_1^C\varphi\kappa^E}{\Delta^\ast} + \frac{J^2K^2(U^2 - IV)^2\rho_1^C\rho_1^E}{(\Delta^\ast)^2} \displaybreak[0] \\
    &\qquad - \varphi^2 (\kappa^{EC})^2 - \frac{J^2K^2(U^2 - IV)^2(\rho_1^{EC})^2}{(\Delta^\ast)^2} + \frac{2JK(U^2 - IV)\rho_1^{EC}\varphi\kappa^{EC}}{\Delta^\ast} \displaybreak[0] \\
    &= (\kappa^E\kappa^C - (\kappa^{EC})^2)\left[\varphi \cdot \frac{J(IU - W)}{\Delta} + \frac{J^2(U^2 - IV)\sigma_E^2\sigma_C^2\{\rho_1^E\rho_1^C - (\rho_1^{EC})^2\}}{\Delta\Delta^\ast}\left(\varphi + \frac{U^2 - IV}{\Delta^\ast}\right)\right] \displaybreak[0] \\
    & = \frac{K^2\eta}{\sigma_E^2\sigma_C^2},
\end{align*}
which leads to
\begin{align*}
    \left(\frac{K^2}{\sigma_E^2\sigma_C^2}\mathbf{M}\right)^{-1} &= \frac{\sigma_E^2\sigma_C^2}{K^2\eta}\left(\begin{array}{cc} 
    \varphi\kappa^E - \frac{JK(U^2 - IV)\rho_1^E}{\Delta^\ast} & \varphi\kappa^{EC} - \frac{JK(U^2 - IV)\rho_1^{EC}}{\Delta^\ast} \\
    \varphi\kappa^{EC} - \frac{JK(U^2 - IV)\rho_1^{EC}}{\Delta^\ast} & \varphi\kappa^C - \frac{JK(U^2 - IV)\rho_1^C}{\Delta^\ast}
    \end{array}\right).
\end{align*}
As a result,
\begin{align*}
    \bfSigma_{\hat\alpha_1, \hat\gamma_1} &= \frac{IJ}{K} \left(\begin{array}{cc}
       \sigma_E & 0 \\
        0 & \sigma_C
    \end{array}\right) \mathbf{M}^{-1} \left(\begin{array}{cc}
       \sigma_E & 0 \\
        0 & \sigma_C
    \end{array}\right) \displaybreak[0]\\
    &= \frac{IJ}{K\eta} \left(\begin{array}{cc}
           \sigma_E & 0 \\
            0 & \sigma_C
        \end{array}\right) \left(\begin{array}{cc} 
    \varphi\kappa^E - \frac{JK(U^2 - IV)\rho_1^E}{\Delta^\ast} & \varphi\kappa^{EC} - \frac{JK(U^2 - IV)\rho_1^{EC}}{\Delta^\ast} \\
    \varphi\kappa^{EC} - \frac{JK(U^2 - IV)\rho_1^{EC}}{\Delta^\ast} & \varphi\kappa^C - \frac{JK(U^2 - IV)\rho_1^C}{\Delta^\ast}
    \end{array}\right) \left(\begin{array}{cc}
           \sigma_E & 0 \\
            0 & \sigma_C
        \end{array}\right) \displaybreak[0]\\
    &= \frac{IJ}{K\eta} \left(\begin{array}{cc} 
        \sigma_E^2\left(\varphi\kappa^E - \frac{JK(U^2 - IV)\rho_1^E}{\Delta^\ast}\right) & \sigma_E\sigma_C\left(\varphi\kappa^{EC} - \frac{JK(U^2 - IV)\rho_1^{EC}}{\Delta^\ast}\right) \\
        \sigma_E\sigma_C\left(\varphi\kappa^{EC} - \frac{JK(U^2 - IV)\rho_1^{EC}}{\Delta^\ast}\right) & \sigma_C^2\left(\varphi\kappa^C - \frac{JK(U^2 - IV)\rho_1^C}{\Delta^\ast}\right)
        \end{array}\right).
\end{align*}
Clearly, we have
\begin{align*}
    \bbV(\hat{\alpha}_1) &= \frac{IJ\sigma_E^2}{K\eta}\left\{\varphi\kappa^E - \frac{JK(U^2 - IV)\rho_1^E}{\Delta^\ast}\right\}, \displaybreak[0]\\
    \bbV(\hat{\gamma}_1) &= \frac{IJ\sigma_C^2}{K\eta}\left\{\varphi\kappa^C - \frac{JK(U^2 - IV)\rho_1^C}{\Delta^\ast}\right\}, \\
    \Cov(\hat{\alpha}_1, \hat{\gamma}_1) &= \frac{IJ\sigma_E\sigma_C}{K\eta}\left\{\varphi\kappa^{EC} - \frac{JK(U^2 - IV)\rho_1^{EC}}{\Delta^\ast}\right\}.
\end{align*}
Given that $\bbV(\hat\beta_1) = \lambda^2\bbV(\hat{\alpha}_1) + \bbV(\hat{\gamma}_1) - 2\lambda\Cov(\hat{\alpha}_1, \hat{\gamma}_1)$, we have
\begin{align*}
    \bbV(\hat\beta_1) &= \frac{IJ\sigma_E\sigma_C}{\eta}\left\{\frac{\varphi}{K}\left(\frac{\lambda^2\kappa^E\sigma_E}{\sigma_C} - 2\lambda\kappa^{EC} + \frac{\kappa^C\sigma_C}{\sigma_E}\right) - \frac{J(U^2 - IV)}{\Delta^\ast}\left(\frac{\lambda^2\rho_1^E\sigma_E}{\sigma_C} - 2\lambda\rho_1^{EC} + \frac{\rho_1^C\sigma_C}{\sigma_E}\right)\right\}.
\end{align*}

\end{proof}

\section{Multiple-Period CRXO Trials} \label{supp_sec:crxo}

\subsection{Closed-Form Variance Expression of the Average INMB Estimator} \label{supp_sec:var_crxo}

Consider a two-sequence CRXO design with $I$ clusters and $J$ periods. For a two-sequence design, let $\pi\in(0,1)$ denote the proportion of clusters assigned to sequence 1 (so $1-\pi$ are assigned to sequence 2). In a CRXO with two sequences, each cluster receives treatment in exactly half of the periods, hence
\begin{align*}
    \sumj Z_{ij}=\frac{J}{2}\quad \forall i,
\end{align*}
which leads to
\begin{align*}
    U = \sumi\sumj Z_{ij} = \frac{IJ}{2},\quad \text{and} \quad V = \sumi\left(\sumj Z_{ij}\right)^2 = \frac{IJ^2}{4}.
\end{align*}
Moreover, for any fixed period $j$, exactly a fraction $\pi$ of clusters are treated in that period under the two sequences. Therefore, we have
\begin{align*}
    W = \sumj\left(\sumi Z_{ij}\right)^2 = \frac{I^2J}{2}\{\pi^2 + (1-\pi)^2\}.
\end{align*}
\begin{theorem} \label{supp_thm:var_crxo}
Under Theorem~\ref{supp_thm:var}, for a two-sequence CRXO design,
    \begin{align}
        \bbV(\hat\beta_1) = \frac{\kappa^{C}\sigma_C^2 - 2\lambda\kappa^{EC}\sigma_C\sigma_E + \lambda^2\kappa^E\sigma_E^2}{IJK\pi(1-\pi)}. \label{supp_eqn:var_crxo}
    \end{align}
\end{theorem}
\begin{proof}

Starting from Theorem~\ref{supp_thm:var},
\begin{align}
    \bbV(\hat\beta_1) = \frac{IJ\sigma_E\sigma_C}{\eta}\left\{\frac{\varphi}{K}\left(\frac{\lambda^2\kappa^E\sigma_E}{\sigma_C} - 2\lambda\kappa^{EC} + \frac{\kappa^C\sigma_C}{\sigma_E}\right) - \frac{J(U^2 - IV)}{\Delta^\ast}\left(\frac{\lambda^2\rho_1^E\sigma_E}{\sigma_C} - 2\lambda\rho_1^{EC} + \frac{\rho_1^C\sigma_C}{\sigma_E}\right)\right\}. \label{supp_eqn:var}
\end{align}
Because $U^2 - IV = 0$, the entire second term in~\eqref{supp_eqn:var} can be dropped and the definitions in Theorem~\ref{supp_thm:var} reduce to
\begin{align*}
    \varphi = J(IU - W)\Delta^{-1}, \quad \text{and} \quad \eta = \varphi J(IU - W).
\end{align*}
Therefore,
\begin{align*}
    \frac{IJ\sigma_E\sigma_C}{\eta} \cdot \varphi = \frac{I\sigma_E\sigma_C}{IU - W}.
\end{align*}
We further note that $IU - W = I^2J\pi(1 - \pi)$, which leads to
\begin{align*}
    \bbV(\hat\beta_1) &= \frac{I\sigma_E\sigma_C}{I^2J\pi(1-\pi)}\cdot \frac{1}{K}\left(\frac{\lambda^2\kappa^E\sigma_E}{\sigma_C}-2\lambda\kappa^{EC} + \frac{\kappa^C\sigma_C}{\sigma_E}\right) \displaybreak[0]\\
    &= \frac{\kappa^{C}\sigma_C^2 - 2\lambda\kappa^{EC}\sigma_C\sigma_E + \lambda^2\kappa^E\sigma_E^2}{IJK\pi(1 - \pi)}.
\end{align*} 

\end{proof}

\subsection{Derivation of $\vartheta$} \label{supp_sec:vartheta_crxo}

We define $r \coloneqq \sigma_E/\sigma_C$ as the ratio of effect to cost standard deviations, so that $(\lambda r)^{-1} = \sigma_C/(\lambda\sigma_E)$ and $(\lambda r)^{-2} = \sigma_C^2/(\lambda^2\sigma_E^2)$. Recall the eigenvalues in Theorem~\ref{supp_thm:var},
\begin{align}
    \kappa^E &= 1 + (K - 1)\rho_0^E - K\rho_1^E, \label{supp_eqn:kappaE} \displaybreak[0]\\
    \kappa^C &= 1 + (K - 1)\rho_0^C - K\rho_1^C, \label{supp_eqn:kappaC} \displaybreak[0]\\
    \kappa^{EC} &= \rho_2^{EC} + (K - 1)\rho_0^{EC} - K\rho_1^{EC}. \label{supp_eqn:kappaEC}
\end{align}
Substituting~\eqref{supp_eqn:kappaE}-\eqref{supp_eqn:kappaEC} into the numerator of~\eqref{supp_eqn:var_crxo} and dividing by $\lambda^2\sigma_E^2$ yields
\begin{align*}
    \frac{\kappa^C\sigma_C^2 - 2\lambda\kappa^{EC}\sigma_C\sigma_E + \lambda^2\kappa^E\sigma_E^2}{\lambda^2\sigma_E^2}
    &= \kappa^E - 2\kappa^{EC}\frac{\sigma_C}{\lambda\sigma_E} + \kappa^C\frac{\sigma_C^2}{\lambda^2\sigma_E^2} \displaybreak[0]\\
    &= \kappa^E -2\kappa^{EC}(\lambda r)^{-1} + \kappa^C(\lambda r)^{-2} \displaybreak[0]\\
    &= \underbrace{(1-\rho_1^E) + 2(\rho_1^{EC} - \rho_2^{EC})(\lambda r)^{-1} + (1 - \rho_1^C)(\lambda r)^{-2}}_{\coloneqq T_1} \displaybreak[0]\\
    &\quad + (K-1)\underbrace{\{(\rho_0^E-\rho_1^E)+2(\rho_1^{EC}-\rho_0^{EC})(\lambda r)^{-1}+(\rho_0^C-\rho_1^C)(\lambda r)^{-2}\}}_{\coloneqq T_2}.
\end{align*}
Therefore, the numerator of~\eqref{supp_eqn:var_crxo} can be re-expressed as
\begin{align}
    &\quad\kappa^C\sigma_C^2 - 2\lambda\kappa^{EC}\sigma_C\sigma_E + \lambda^2\kappa^E\sigma_E^2 \nonumber \displaybreak[0]\\
    &=\lambda^2\sigma_E^2\{T_1+(K-1)T_2\} \nonumber \displaybreak[0]\\
    &=\lambda^2\sigma_E^2T_2\{K+\vartheta\}, \label{supp_eqn:varthe_crxo}
\end{align}
where
\begin{align*}
    \vartheta \coloneqq \frac{T_1}{T_2}-1 = \frac{(1-\rho_1^E) + 2(\rho_1^{EC}-\rho_2^{EC})(\lambda r)^{-1} + (1-\rho_1^C)(\lambda r)^{-2}}{(\rho_0^E-\rho_1^E) + 2(\rho_1^{EC}-\rho_0^{EC})(\lambda r)^{-1} + (\rho_0^C-\rho_1^C)(\lambda r)^{-2}} - 1.
\end{align*}
Combining~\eqref{supp_eqn:var_crxo} and~\eqref{supp_eqn:varthe_crxo} gives the useful representation
\begin{align}
    \bbV(\hat\beta_1) = \frac{\lambda^2\sigma_E^2T_2}{IJ\pi(1-\pi)}\cdot \frac{K+\vartheta}{K}. \label{supp_eqn:var_crxo_v2}
\end{align}

\subsection{Decimal-Valued LOD Under a Linear Budget Constraint} \label{supp_sec:LOD_crxo}

We assume the total budget takes the standard linear form
\begin{align} 
    B = I(c_1 + c_2JK), \label{supp_eqn:budget}
\end{align} 
where $c_1$ is the cost per cluster and $c_2$ is the cost per participant. Solving~\eqref{supp_eqn:budget} for $I$ yields
\begin{align}
I = \frac{B}{c_1 + c_2JK}. \label{supp_eqn:I_of_K}
\end{align}
Substituting~\eqref{supp_eqn:I_of_K} into~\eqref{supp_eqn:var_crxo_v2} gives
\begin{align}
    \bbV(\hat\beta_1)
    &= \frac{\lambda^2\sigma_E^2T_2}{J\pi(1 - \pi)}\cdot \frac{c_1 + c_2JK}{B}\cdot\frac{K + \vartheta}{K}\nonumber\\
    &= \frac{\lambda^2\sigma_E^2T_2}{J\pi(1 - \pi)B}\cdot f(c_1, c_2, J, K, \vartheta),\label{supp_eq:var_times_f}
\end{align}
where
\begin{align*}
    f(c_1, c_2, J, K, \vartheta) \coloneqq (c_1 + c_2JK)\left(1+\frac{\vartheta}{K}\right).
\end{align*}

\begin{theorem}\label{supp_thm:LOD_crxo}
Suppose $\vartheta > 0$ and the budget constraint~\eqref{supp_eqn:budget} holds. The decimal-valued LOD for a two-sequence CRXO trial is
\begin{align*}
    K_{\text{LOD}} = \sqrt{\frac{c_1\vartheta}{c_2J}}, \quad \text{and} \quad I_{\text{LOD}} = \frac{B}{c_1 + \sqrt{\vartheta c_1c_2J}},
\end{align*}
where
\begin{align*}
    \vartheta = \frac{(1-\rho_1^E) + 2(\rho_1^{EC}-\rho_2^{EC})(\lambda r)^{-1} + (1-\rho_1^C)(\lambda r)^{-2}}{(\rho_0^E-\rho_1^E) + 2(\rho_1^{EC}-\rho_0^{EC})(\lambda r)^{-1} + (\rho_0^C-\rho_1^C)(\lambda r)^{-2}} - 1.
\end{align*}
\end{theorem}

\begin{proof}

Because the prefactor in~\eqref{supp_eq:var_times_f} does not depend on $K$, minimizing $\mathbb{V}(\hat\beta_1)$ over $K>0$ is equivalent to minimizing $f(c_1, c_2, J, K, \vartheta)$. We thus differentiate $f(c_1, c_2, J, K, \vartheta)$ with respect to $K$:
\begin{align*}
    f^\prime(c_1, c_2, J, K, \vartheta) = c_2J - \frac{c_1 \vartheta}{K^2},\quad \text{and} \quad f^{\prime\prime}(c_1, c_2, J, K, \vartheta) = \frac{2c_1 \vartheta}{K^3}.
\end{align*}
When $\vartheta>0$ and $c_1>0$, we have $f^{\prime\prime}(K)>0$ for all $K>0$, so any critical point is the unique global minimizer. Setting $f^\prime(c_1, c_2, J, K, \vartheta) = 0$ yields
\begin{align*}
    K_{\text{LOD}} = \sqrt{\frac{c_1\vartheta}{c_2J}}
\end{align*}
Substituting $K_{\text{LOD}}$ into~\eqref{supp_eqn:I_of_K} gives     
\begin{align*}
    I_{\text{LOD}} = \frac{B}{c_1 + \sqrt{\vartheta c_1c_2J}}.
\end{align*}

\end{proof}

Because $K_{\text{LOD}} \propto \sqrt{\vartheta}$ and $I_{\text{LOD}}$ is inversely related to $\sqrt{\vartheta}$, both maximizing the number of clusters and minimizing cluster-period size are achieved by minimizing $\vartheta$ with respect to $\lambda r$. To facilitate derivation, we let $x \coloneqq \lambda r$. Minimizing $\vartheta(x)$ is equivalent to minimizing $\phi(x) = \vartheta(x) + 1$. Multiplying the numerator and denominator of $\vartheta(x)$ by $x^2$ gives
\begin{align*}
    \phi(x) = \frac{T_3(x)}{T_4(x)},
\end{align*}
where
\begin{align*}
    T_3(x) &= (1-\rho_1^E)x^2 + 2(\rho_1^{EC}-\rho_2^{EC})x + (1-\rho_1^C),\\
    T_4(x) &= (\rho_0^E-\rho_1^E)x^2 + 2(\rho_1^{EC}-\rho_0^{EC})x + (\rho_0^C-\rho_1^C).
\end{align*}
Differentiating $\phi(x)$ yields
\begin{align*}
    \phi'(x) = \frac{T_3^\prime(x)T_4(x) - T_3(x)T_4^\prime(x)}{T_4^2(x)}.
\end{align*}
Setting $\phi'(x) = 0$ leads to the quadratic equation
\begin{align}
    t_2 x^2 + t_1 x + t_0 = 0, \label{supp_eqn:quad}
\end{align}
with coefficients
\begin{align*}
    t_2 &= 2(1-\rho_1^E)(\rho_1^{EC}-\rho_0^{EC}) - 2(\rho_1^{EC}-\rho_2^{EC})(\rho_0^E-\rho_1^E),\\
    t_1 &= 2(1-\rho_1^E)(\rho_0^C-\rho_1^C) - 2(1-\rho_1^C)(\rho_0^E-\rho_1^E),\\
    t_0 &= 2(\rho_1^{EC}-\rho_2^{EC})(\rho_0^C-\rho_1^C) - 2(1-\rho_1^C)(\rho_1^{EC}-\rho_0^{EC}).
\end{align*}
If $t_2 \neq 0$, solving~\eqref{supp_eqn:quad} gives
\begin{align*}
    x_{\pm} = \frac{-t_1 \pm \sqrt{t_1^2 - 4t_2t_0}}{2t_2}.
\end{align*}
The optimal ratio is $x^* = \lambda r$ equal to the positive root satisfying $T_4(x^*) > 0$ that yields the minimum. If $t_2 = 0$ and $t_1 \neq 0$, then $x^* = -t_0/t_1$.

\subsection{Closed-Form RE of the MMD} \label{supp_sec:MMD_crxo}

\begin{theorem}\label{supp_thm:MMD_crxo}
Suppose $\vartheta > 0$ and the budget constraint~\eqref{supp_eqn:budget} holds. The RE of the MMD for a two-sequence CRXO trial is
\begin{align*} 
    \text{RE}(\bm{\rho}, K) = \frac{g(\vartheta)}{\vartheta + K} \times \frac{K}{c_1 + JKc_2},
\end{align*} 
where $g(\vartheta) = (\sqrt{c_1} + \sqrt{\vartheta c_2 J})^2$ and 
\begin{align*}
    \vartheta = \frac{(1-\rho_1^E) + 2(\rho_1^{EC}-\rho_2^{EC})(\lambda r)^{-1} + (1-\rho_1^C)(\lambda r)^{-2}}{(\rho_0^E-\rho_1^E) + 2(\rho_1^{EC}-\rho_0^{EC})(\lambda r)^{-1} + (\rho_0^C-\rho_1^C)(\lambda r)^{-2}} - 1.
\end{align*}
\end{theorem}

\begin{proof}

In Section 3 of the main article, we define RE of the MMD at $(\bm{\rho},K)$ as
\begin{align*}
    \text{RE} \coloneqq \frac{\bbV_{\text{dec}}(\hat\beta_1)}{\bbV(\hat\beta_1)},
\end{align*}
where
\begin{align*}
    \bbV(\hat\beta_1) = \frac{\lambda^2\sigma_E^2T_2}{J\pi(1 - \pi)}\cdot \frac{c_1 + c_2JK}{B}\cdot\frac{K + \vartheta}{K}
\end{align*}
from Web Appendix B.4 and $\bbV_{\text{dec}}(\hat\beta_1)$ is obtained by evaluating $\bbV(\hat\beta_1)$ at $K = K_{\text{LOD}}$ in Theorem~\ref{supp_thm:LOD_crxo}. Because the multiplicative prefactor $\lambda^2\sigma_E^2T_2\{J\pi(1 - \pi)B\}^{-1}$ in $\bbV(\hat\beta_1)$ does not depend on $K$, it cancels in the RE. Hence,
\begin{align*}
    \text{RE} = \frac{c_1 + c_2JK_{\text{LOD}}}{c_1 + c_2JK}\cdot \frac{\vartheta + K_{\text{LOD}}}{K_{\text{LOD}}} \cdot \frac{K}{\vartheta + K}.
\end{align*}
From Theorem~\ref{supp_thm:LOD_crxo}, we have
\begin{align*}
    K_{\text{LOD}} = \sqrt{\frac{c_1\vartheta}{c_2J}}, \quad \text{and} \quad I_{\text{LOD}} = \frac{B}{c_1 + \sqrt{\vartheta c_1c_2J}},
\end{align*}
which leads to
\begin{align*}
    c_1+c_2JK_{\text{LOD}} = c_1 + \sqrt{\vartheta c_1c_2J}, \quad \text{and} \quad \frac{\vartheta + K_{\text{LOD}}}{K_{\text{LOD}}} = \frac{\sqrt{c_1}+\sqrt{\vartheta c_2J}}{\sqrt{c_1}}.
\end{align*}
Putting everything together, we have
\begin{align*}
    \text{RE} = \frac{g(\vartheta)}{\vartheta + K} \times \frac{K}{c_1 + JKc_2},
\end{align*}
where $g(\vartheta) = (\sqrt{c_1} + \sqrt{\vartheta c_2 J})^2$.

\end{proof}

\section{PA-LCRTs With Multiple Periods} \label{supp_sec:pa}

\subsection{Closed-Form Variance Expression of the Average INMB Estimator} \label{supp_sec:var_pa}

Consider a PA-LCRT with $I$ clusters and $J$ periods. Clusters are randomized at baseline to either treatment ($Z_{ij}=1$ for all $j$) or control ($Z_{ij}=0$ for all $j$) and remain in that arm for the entire trial. Let $\pi\in(0,1)$ denote the proportion of clusters assigned to treatment, so that $\pi I$ clusters receive treatment and $(1-\pi)I$ clusters receive control. The design constants are
\begin{align*}
    U &= \sumi\sumj Z_{ij} = \pi IJ, \\
    V &= \sumi\left(\sumj Z_{ij}\right)^2 = \pi IJ^2, \\
    W &= \sumj\left(\sumi Z_{ij}\right)^2 = \pi^2 I^2J.
\end{align*}

\begin{theorem} \label{supp_thm:var_pa}
Under Theorem~\ref{supp_thm:var}, for a PA-CRT with $J$ periods,
\begin{align*}
    \bbV(\hat\beta_1) = \frac{\kappa^C \sigma_C^2 - 2 \lambda \kappa^{EC} \sigma_C \sigma_E + \lambda^2 \kappa^E \sigma_E^2}{I J K \pi (1 - \pi)} + \frac{\rho_1^C \sigma_C^2 - 2 \lambda \rho_1^{EC} \sigma_C \sigma_E + \lambda^2 \rho_1^E \sigma_E^2}{I \pi (1 - \pi)}.
\end{align*}
\end{theorem}

\begin{proof}

Under a PA-CRT, the above design constants imply
\begin{align*}
    U^2 - IV = -I^2J^2\pi(1 - \pi), \quad \text{and} \quad IU - W = I^2J\pi(1 - \pi),
\end{align*}
which implies that $U^2 - IV = -J(IU - W)$. Then, using the definitions in Theorem~\ref{supp_thm:var},
\begin{align*}
    \varphi = J(IU - W)\Delta^{*-1}, \quad \text{and} \quad \eta = \varphi J(IU - W).
\end{align*}
Therefore,
\begin{align*}
    \frac{IJ\sigma_E\sigma_C}{\eta}\cdot \varphi
    = \frac{I\sigma_E\sigma_C}{IU - W}
    = \frac{\sigma_E\sigma_C}{IJ\pi(1-\pi)}, \quad \text{and} \quad \frac{IJ\sigma_E\sigma_C}{\eta}\cdot \frac{-J}{\Delta^\ast}(U^2-IV) = \frac{J\sigma_E\sigma_C}{IJ\pi(1 - \pi)}.
\end{align*}
Substituting these identities into~\eqref{supp_eqn:var} yields
\begin{align*}
    \bbV(\hat\beta_1)
    &= \frac{\sigma_E\sigma_C}{IJ\pi(1-\pi)}
    \left\{
    \frac{1}{K}\left(\frac{\lambda^2\kappa^E\sigma_E}{\sigma_C} - 2\lambda\kappa^{EC} + \frac{\kappa^C\sigma_C}{\sigma_E}\right)
    + J\left(\frac{\lambda^2\rho_1^E\sigma_E}{\sigma_C} - 2\lambda\rho_1^{EC} + \frac{\rho_1^C\sigma_C}{\sigma_E}\right)
    \right\} \displaybreak[0]\\
    &= \frac{\kappa^C \sigma_C^2 - 2 \lambda \kappa^{EC} \sigma_C \sigma_E + \lambda^2 \kappa^E \sigma_E^2}{I J K \pi (1 - \pi)} + \frac{\rho_1^C \sigma_C^2 - 2 \lambda \rho_1^{EC} \sigma_C \sigma_E + \lambda^2 \rho_1^E \sigma_E^2}{I \pi (1 - \pi)}.
\end{align*}

\end{proof}

\subsection{Derivation of $\vartheta$} \label{supp_sec:vartheta_pa}

Under Theorem~\ref{supp_thm:var_pa}, the variance for a PA-CRT can be written as
\begin{align}
    \bbV(\hat\beta_1)
    =\frac{(\kappa^C\sigma_C^2 - 2\lambda\kappa^{EC}\sigma_C\sigma_E + \lambda^2\kappa^E\sigma_E^2) + JK(\rho_1^C\sigma_C^2 - 2\lambda\rho_1^{EC}\sigma_C\sigma_E + \lambda^2\rho_1^E\sigma_E^2)}{IJK\pi(1-\pi)}. \label{supp_eqn:var_pa}
\end{align}
Dividing the numerator of~\eqref{supp_eqn:var_pa} by $\lambda^2\sigma_E^2$ yields
\begin{align*}
    &\frac{
    (\kappa^C\sigma_C^2 - 2\lambda\kappa^{EC}\sigma_C\sigma_E + \lambda^2\kappa^E\sigma_E^2)
    + JK(\rho_1^C\sigma_C^2 - 2\lambda\rho_1^{EC}\sigma_C\sigma_E + \lambda^2\rho_1^E\sigma_E^2)
    }{\lambda^2\sigma_E^2} \\
    &\qquad=
    \underbrace{\{\kappa^E - 2\kappa^{EC}(\lambda r)^{-1} + \kappa^C(\lambda r)^{-2}\}}_{\coloneqq T_5}
    + JK\underbrace{\{\rho_1^E - 2\rho_1^{EC}(\lambda r)^{-1} + \rho_1^C(\lambda r)^{-2}\}}_{\coloneqq T_6}. \displaybreak[0]
\end{align*}
Substituting~\eqref{supp_eqn:kappaE}-\eqref{supp_eqn:kappaEC} into $T_5$ gives
\begin{align*}
    \kappa^E - 2\kappa^{EC}(\lambda r)^{-1} + \kappa^C(\lambda r)^{-2}
    &= \underbrace{(1-\rho_1^E) + 2(\rho_1^{EC}-\rho_2^{EC})(\lambda r)^{-1} + (1-\rho_1^C)(\lambda r)^{-2}}_{\coloneqq T_7} \displaybreak[0]\\
    &\quad + (K-1)\underbrace{\{(\rho_0^E-\rho_1^E)+2(\rho_1^{EC}-\rho_0^{EC})(\lambda r)^{-1}+(\rho_0^C-\rho_1^C)(\lambda r)^{-2}\}}_{\coloneqq T_8}.
\end{align*}
Therefore, the numerator of~\eqref{supp_eqn:var_pa} can be re-expressed as
\begin{align}
    &\quad(\kappa^C\sigma_C^2 - 2\lambda\kappa^{EC}\sigma_C\sigma_E + \lambda^2\kappa^E\sigma_E^2) + JK(\rho_1^C\sigma_C^2 - 2\lambda\rho_1^{EC}\sigma_C\sigma_E + \lambda^2\rho_1^E\sigma_E^2) \nonumber\displaybreak[0]\\
    &= \lambda^2\sigma_E^2 \{T_7 + (K-1)T_8 + JK T_6\} \nonumber\displaybreak[0]\\
    &= \lambda^2\sigma_E^2 \{(T_7 + JT_6) + (K-1)(T_8 + JT_6)\} \nonumber\displaybreak[0]\\
    &= \lambda^2\sigma_E^2 (T_8 + JT_6) (K + \vartheta), \label{supp_eqn:vartheta_pa}
\end{align}
where
\begin{align*}
    \vartheta &\coloneqq \frac{T_7 + JT_6}{T_8 + JT_6} - 1. \displaybreak[0]\\
    &= \frac{(1 + J\rho^E_1 - \rho^E_1) + 2(\rho_1^{EC} -J\rho_1^{EC} - \rho_2^{EC})(\lambda r)^{-1} + (1 + J\rho^C_1 - \rho^C_1)(\lambda r)^{-2}}
    {(J\rho^E_1 + \rho^E_0 - \rho^E_1) + 2(-J\rho_1^{EC} - \rho_0^{EC} + \rho_1^{EC})(\lambda r)^{-1} + (J\rho^C_1 + \rho^C_0 - \rho^C_1)(\lambda r)^{-2}} - 1.
\end{align*}
Combining~\eqref{supp_eqn:var_pa} and~\eqref{supp_eqn:vartheta_pa} gives the useful representation
\begin{align}
    \bbV(\hat\beta_1)
    = \frac{\lambda^2\sigma_E^2(T_8 + JT_6)}{IJ\pi(1-\pi)}\cdot \frac{K+\vartheta}{K}. \label{supp_eqn:var_pa_v2}
\end{align}

\subsection{Decimal-Valued LOD Under a Linear Budget Constraint} \label{supp_sec:LOD_pa}

\begin{theorem}\label{supp_thm:LOD_pa}
Suppose $\vartheta > 0$ and the budget constraint~\eqref{supp_eqn:budget} holds. The decimal-valued LOD for a PA-CRT with $J$ periods is
\begin{align*}
    K_{\text{LOD}} = \sqrt{\frac{c_1\vartheta}{c_2J}}, \quad \text{and} \quad I_{\text{LOD}} = \frac{B}{c_1 + \sqrt{\vartheta c_1c_2J}},
\end{align*}
where
\begin{align*}
    \vartheta = \frac{(1 + J\rho^E_1 - \rho^E_1) + 2(\rho_1^{EC} - J\rho_1^{EC} - \rho_2^{EC})(\lambda r)^{-1} + (1 + J\rho^C_1 - \rho^C_1)(\lambda r)^{-2}}{(J\rho^E_1 + \rho^E_0 - \rho^E_1) + 2(-J\rho_1^{EC} - \rho_0^{EC} + \rho_1^{EC})(\lambda r)^{-1} + (J\rho^C_1 + \rho^C_0 - \rho^C_1)(\lambda r)^{-2}} - 1.
\end{align*}
\end{theorem}

\begin{proof}
    Substituting~\eqref{supp_eqn:I_of_K} into~\eqref{supp_eqn:var_pa_v2} yields a variance of the form~\eqref{supp_eq:var_times_f} with prefactor independent of $K$. The result then follows from the argument in Theorem~\ref{supp_thm:LOD_crxo}.
\end{proof}

\subsection{Closed-Form RE of the MMD} \label{supp_sec:MMD_pa}

\begin{theorem}\label{supp_thm:MMD_pa}
Suppose $\vartheta>0$ and the budget constraint~\eqref{supp_eqn:budget} holds. The RE of the MMD for a PA-CRT with $J$ periods is
\begin{align*}
    \text{RE}(\bm{\rho},K) = \frac{g(\vartheta)}{\vartheta+K}\times \frac{K}{c_1 + JKc_2},
\end{align*}
where $g(\vartheta) = (\sqrt{c_1} + \sqrt{\vartheta c_2 J})^2$ and
\begin{align*}
    \vartheta = \frac{(1 + J\rho^E_1 - \rho^E_1) + 2(\rho_1^{EC} - J\rho_1^{EC} - \rho_2^{EC})(\lambda r)^{-1} + (1 + J\rho^C_1 - \rho^C_1)(\lambda r)^{-2}}{(J\rho^E_1 + \rho^E_0 - \rho^E_1) + 2(-J\rho_1^{EC} - \rho_0^{EC} + \rho_1^{EC})(\lambda r)^{-1} + (J\rho^C_1 + \rho^C_0 - \rho^C_1)(\lambda r)^{-2}} - 1.
\end{align*}
\end{theorem}

\begin{proof}
    Since~\eqref{supp_eqn:var_pa_v2} and~\eqref{supp_eqn:var_crxo_v2} share the same functional form in $K$, the prefactor cancels in the RE ratio and the result follows from Theorem~\ref{supp_thm:MMD_crxo}.
\end{proof}

\section{R Shiny Application Tutorial} \label{supp_sec:tutorial_shiny}

To facilitate the implementation of the proposed methods, we provide an R Shiny web application available at \href{supp_https://f07k8s-hao-wang.shinyapps.io/Cost-effectiveness_LCRT/}{https://f07k8s-hao-wang.shinyapps.io/Cost-effectiveness\_LCRT/} that automates LOD and MMD calculations for all three L-CRT design variants. This tutorial provides step-by-step guidance for investigators planning cost-effectiveness L-CRTs who wish to determine the optimal number of clusters $I$ and cluster-period size $K$ under a fixed budget. We organize this tutorial into five parts: (i) an overview of the application layout (Section~\ref{supp_sec:tut_overview}), (ii) guidance on entering the required input parameters (Section~\ref{supp_sec:tut_inputs}), (iii) an illustrative example for computing an LOD when ICC parameters are known (Section~\ref{supp_sec:tut_lod}), and (iv) an illustrative example for computing an MMD when ICC parameters are uncertain (Section~\ref{supp_sec:tut_mmd}).

\subsection{Application Overview} \label{supp_sec:tut_overview}

The application interface is organized into two rows. The first row presents background information: the left panel displays the bivariate linear mixed model specification and the INMB definition, while the right panel displays a graphical representation of the seven ICCs along with their formal definitions and ordering constraints. Investigators may refer to this row throughout the design process to clarify the meaning and relationships among the ICC parameters.

The second row is divided into three columns. The first column contains all design inputs, including the trial design (parallel-arm LCRT, cluster randomized crossover, or stepped-wedge CRT), the type of design optimization (LOD or MMD), and the general, design, and budget parameters. The second column contains the ICC parameter inputs, which change dynamically depending on the selected optimization type: when LOD is selected, the application displays a single input field for each of the seven ICC parameters; when MMD is selected, the application displays minimum and maximum fields for each ICC, defining the parameter space $\bm{\Theta}$ over which the worst-case efficiency is maximized. The third column displays the optimization results after the investigator clicks ``Run Optimization.''

\begin{figure}[ht!]
    \centering
    \includegraphics[width=1\linewidth]{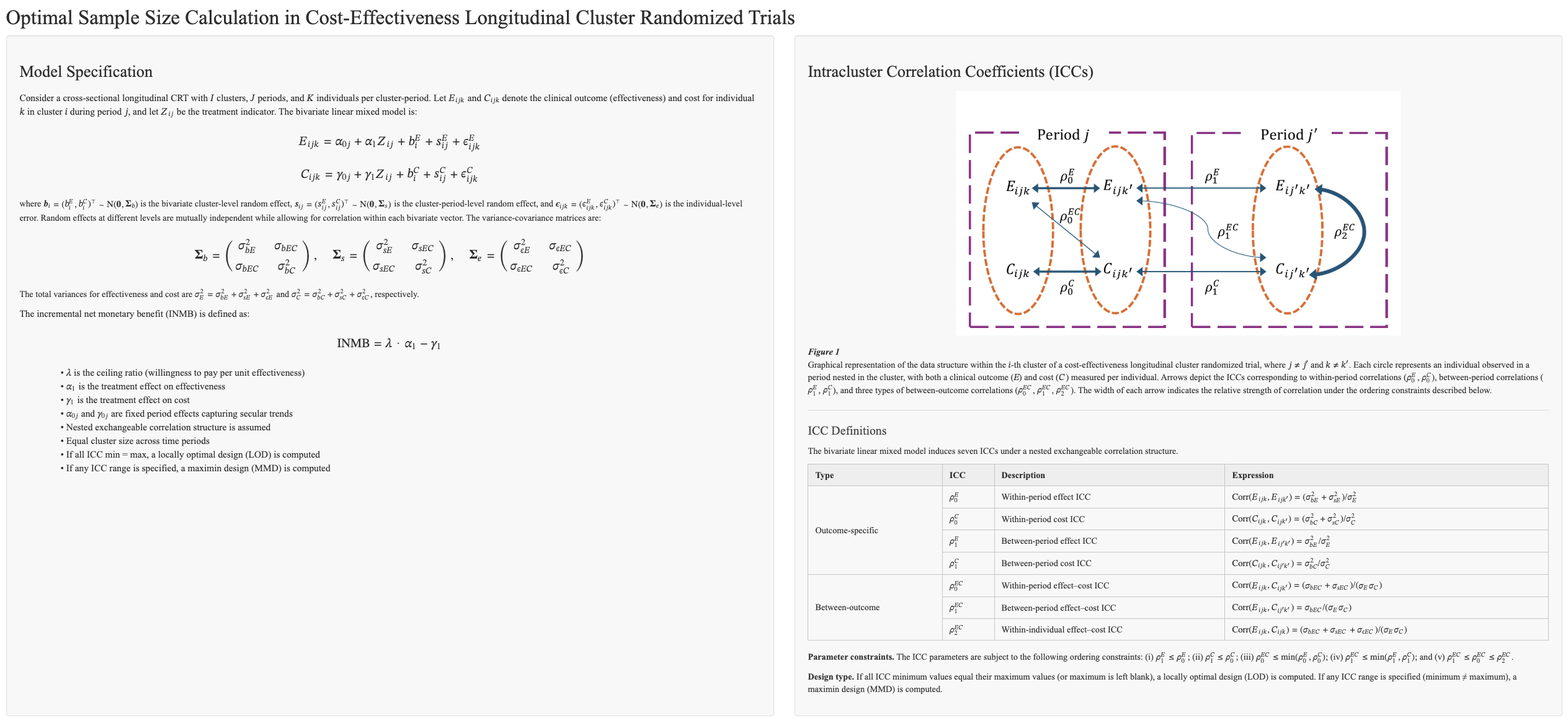}
    \caption{First row of the R Shiny application interface. The left panel displays the bivariate linear mixed model specification for jointly modeling clinical outcomes ($E_{ijk}$) and costs ($C_{ijk}$), including the variance-covariance matrices $\boldsymbol{\Sigma}_b$, $\boldsymbol{\Sigma}_s$, and $\boldsymbol{\Sigma}_e$ for the cluster-level, cluster-period-level, and individual-level random effects, respectively, and the INMB definition. The right panel displays a graphical representation of the data structure within the $i$-th cluster of a cost-effectiveness L-CRT, illustrating the seven intracluster correlation coefficients: within-period ICCs ($\rho_0^E$, $\rho_0^C$), between-period ICCs ($\rho_1^E$, $\rho_1^C$), and three between-outcome ICCs ($\rho_0^{EC}$, $\rho_1^{EC}$, $\rho_2^{EC}$), along with their formal definitions and ordering constraints.}
    \label{supp_fig:shiny_row_one}
\end{figure}

\subsection{Input Parameters} \label{supp_sec:tut_inputs}

Before using the application, investigators need to determine the following information. We describe each group of parameters in turn, noting which quantities are typically available from prior literature, pilot data, or content knowledge.

\paragraph{Trial design and type of optimization} The investigator first selects one of three L-CRT designs: (a) parallel-arm LCRT, where clusters are randomized to intervention or control for all periods and the same individuals are followed longitudinally; (b) CRXO, where clusters alternate between intervention and control across periods with different individuals sampled in each period; or (c) SW-CRT, where all clusters begin under control and sequentially cross over to intervention at pre-specified periods with different individuals sampled in each period. The investigator then selects the optimization type: LOD if reasonable point estimates of the ICC parameters are available from pilot data or the literature, or MMD if the ICC parameters are uncertain and only plausible ranges can be specified.

\paragraph{General parameters} The investigator specifies five quantities: (i) the type I error rate $\alpha$ (typically set to 0.05 for a two-sided test); (ii) the treatment effect on the INMB scale, $\beta_1 = \lambda\alpha_1 - \gamma_1$, where $\alpha_1$ is the anticipated treatment effect on the clinical outcome and $\gamma_1$ is the anticipated treatment effect on cost; (iii) the ceiling ratio $\lambda$, representing the maximum willingness-to-pay per unit of clinical benefit; (iv) $\sigma_E$, the total standard deviation of the clinical outcome; and (v) $\sigma_C$, the total standard deviation of the cost outcome. We note that $\sigma_E$ and $\sigma_C$ can be approximated from prior studies or pilot data, and should reflect the total (marginal) variability including both within-cluster and between-cluster components.

\paragraph{Design parameters} For all three designs, the investigator specifies the number of time periods $J$. For PA-LCRTs and CRXO trials, the investigator additionally specifies the proportion of clusters allocated to the treatment arm $\pi$ as a fraction (e.g., numerator $= 1$ and denominator $= 2$ for a balanced design with $\pi = 0.5$). Of note, we only support an even number of periods for CRXO trials. For SW-CRTs, the investigator specifies the number of steps $L$ (i.e., the number of time points at which clusters cross over from control to intervention); the total number of clusters $I$ must be divisible by $L$.

\paragraph{Budget parameters} The investigator specifies the cost per cluster $c_1$, the cost per individual per period $c_2$, the total available budget $B$, and upper bounds on the number of clusters ($I_{\max}$) and cluster-period size ($K_{\max}$). The budget constraint takes the form $B = I(c_1 + c_2JK)$, where $c_1$ encompasses all cluster-level expenses (e.g., site setup, administrative overhead, institutional review) and $c_2$ encompasses all individual-level expenses per period (e.g., data collection, outcome measurement, follow-up costs).

\paragraph{ICC parameters} The bivariate linear mixed model requires specification of seven ICCs: $\rho_0^E$ and $\rho_1^E$ (within-period and between-period effect ICCs), $\rho_0^C$ and $\rho_1^C$ (within-period and between-period cost ICCs), $\rho_0^{EC}$ and $\rho_1^{EC}$ (within-period and between-period effect--cost ICCs), and $\rho_2^{EC}$ (within-individual effect--cost ICC). These parameters must satisfy the ordering constraints: (i) $\rho_1^E \leq \rho_0^E$; (ii) $\rho_1^C \leq \rho_0^C$; (iii) $\rho_0^{EC} \leq \min(\rho_0^E, \rho_0^C)$; (iv) $\rho_1^{EC} \leq \min(\rho_1^E, \rho_1^C)$; and (v) $\rho_1^{EC} \leq \rho_0^{EC} \leq \rho_2^{EC}$. In practice, ICCs for clinical outcomes are increasingly reported in CRT publications \citep{Korevaar2021}, but ICCs for cost outcomes are rarely documented, and between-outcome ICCs are largely absent from the literature. When direct estimates are unavailable, we encourage investigators to conduct a sensitivity analysis by trying a range of possible values for the cost ICCs and between-outcome ICCs. If the optimal sample size combination remains stable across these scenarios, the investigator can be confident in the chosen design; if it varies substantially, the MMD option should be used with conservative ranges that encompass the full set of plausible values.

\subsection{Illustrative Example: LOD} \label{supp_sec:tut_lod}

We illustrate the LOD workflow using a CRXO trial with parameters drawn from the Australian reinvestment trial described in Section~7 of the main article. The investigator proceeds as follows.

\paragraph{Step 1: Select the trial design and optimization type} In the first column, the investigator selects ``Cluster randomized crossover'' as the trial design and ``Locally optimal design (LOD)'' as the optimization type.

\paragraph{Step 2: Enter general parameters} The investigator enters $\alpha = 0.05$, $|\beta_1| = 2{,}089$ (the anticipated INMB, where the absolute value is used because the sample size calculation targets $|\beta_1|$ to ensure adequate power for detecting whether the intervention achieves cost-effectiveness), $\lambda = 216$ (ceiling ratio, representing the maximum willingness-to-pay per bed-day saved), $\sigma_E = 6.48$ (standard deviation of length of stay in days), and $\sigma_C = 11{,}635$ (standard deviation of cost in Australian dollars), yielding a standardized ceiling ratio of $\lambda r = \lambda \sigma_E / \sigma_C = 216 \times 6.48 / 11{,}635 \approx 0.12$.

\paragraph{Step 3: Enter design parameters} The investigator sets $J = 8$ periods (corresponding to eight calendar months) and specifies a balanced allocation with proportion numerator $= 1$ and denominator $= 2$ (i.e., $\pi = 0.5$).

\paragraph{Step 4: Enter budget parameters} The investigator sets $c_1 = \$3{,}000$ (cost per cluster), $c_2 = \$250$ (cost per individual per period), $B = \$600{,}000$ (total budget), $I_{\max} = 100$ (maximum clusters), and $K_{\max} = 200$ (maximum cluster-period size).

\paragraph{Step 5: Enter ICC parameters} Because the LOD assumes known ICC values, the application displays a single input field for each ICC. The investigator enters the point estimates obtained from the bivariate linear mixed model fitted to the Trial 2 data: $\rho_0^E = 0.048$, $\rho_1^E = 0.042$, $\rho_0^C = 0.020$, $\rho_1^C = 0.018$, $\rho_0^{EC} = 0.007$, $\rho_1^{EC} = 0.004$, and $\rho_2^{EC} = 0.75$.

\paragraph{Step 6: Run and review results} The investigator clicks ``Run Optimization.'' The results panel displays both the decimal (continuous) and integer (rounded) optimal designs. Figure~\ref{supp_fig:shiny_lod} presents a screenshot of the application with these inputs and results. The integer estimates yield $(I_{\text{LOD}}, K_{\text{LOD}}) = (8, 36)$, achieving statistical power of 0.996 for detecting the specified INMB at the $\alpha = 0.05$ significance level.

\begin{figure}[t]
    \centering
    \includegraphics[width=\linewidth]{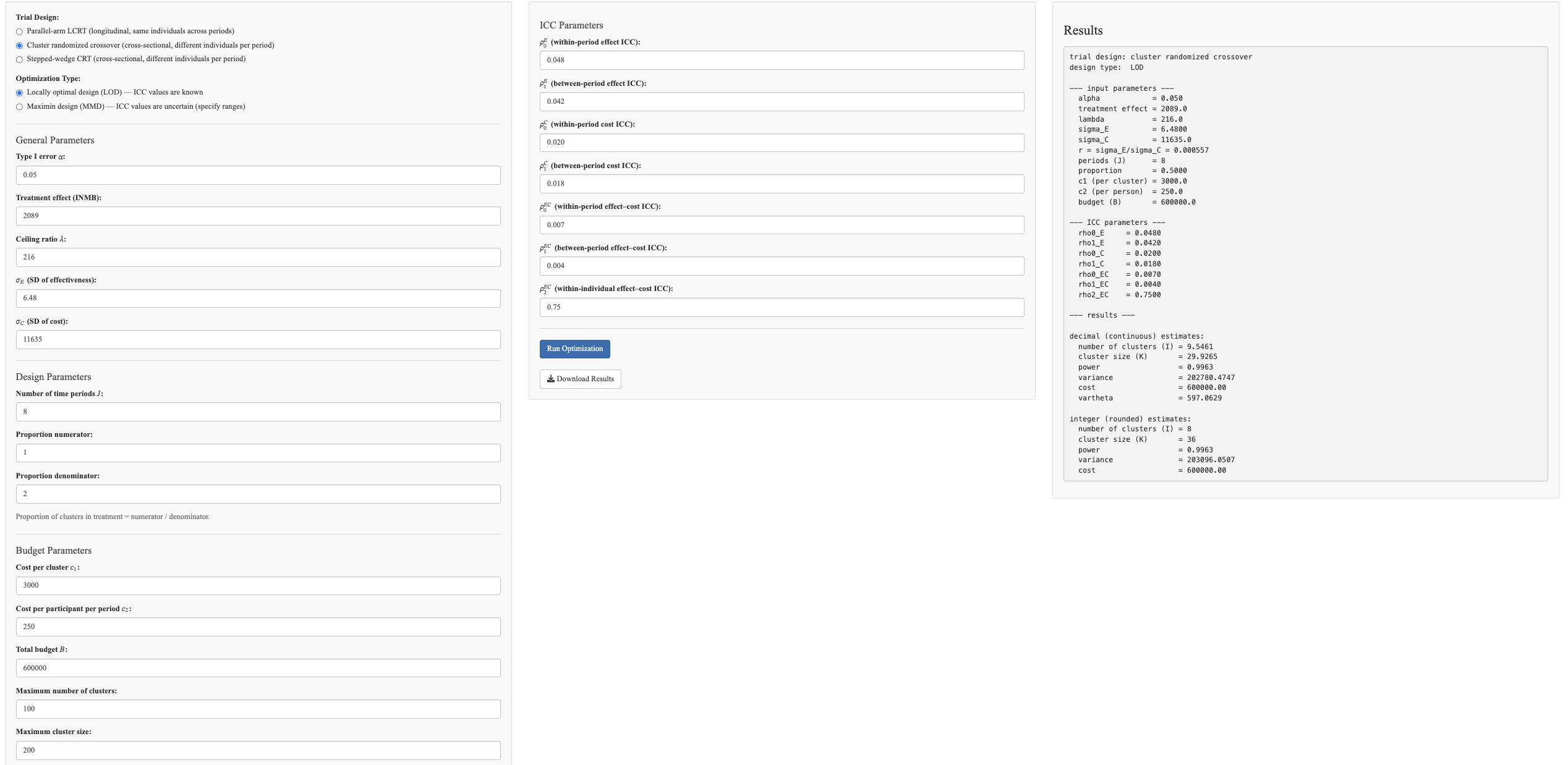}
    \caption{Screenshot of the R Shiny application for computing the LOD for a CRXO trial with $J = 8$ periods, $\pi = 0.5$, $B = \$600{,}000$, $c_1 = \$3{,}000$, $c_2 = \$250$, and ICC parameters $\rho_0^E = 0.048$, $\rho_1^E = 0.042$, $\rho_0^C = 0.020$, $\rho_1^C = 0.018$, $\rho_0^{EC} = 0.007$, $\rho_1^{EC} = 0.004$, and $\rho_2^{EC} = 0.75$, estimated from the Australian reinvestment trial. The results panel displays the decimal and integer optimal designs, with the integer LOD allocating $I_{\text{LOD}} = 8$ clusters and $K_{\text{LOD}} = 36$ individuals per cluster-period, achieving power of 0.996.}
    \label{supp_fig:shiny_lod}
\end{figure}

\begin{figure}[ht!]
    \centering
    \includegraphics[width=\linewidth]{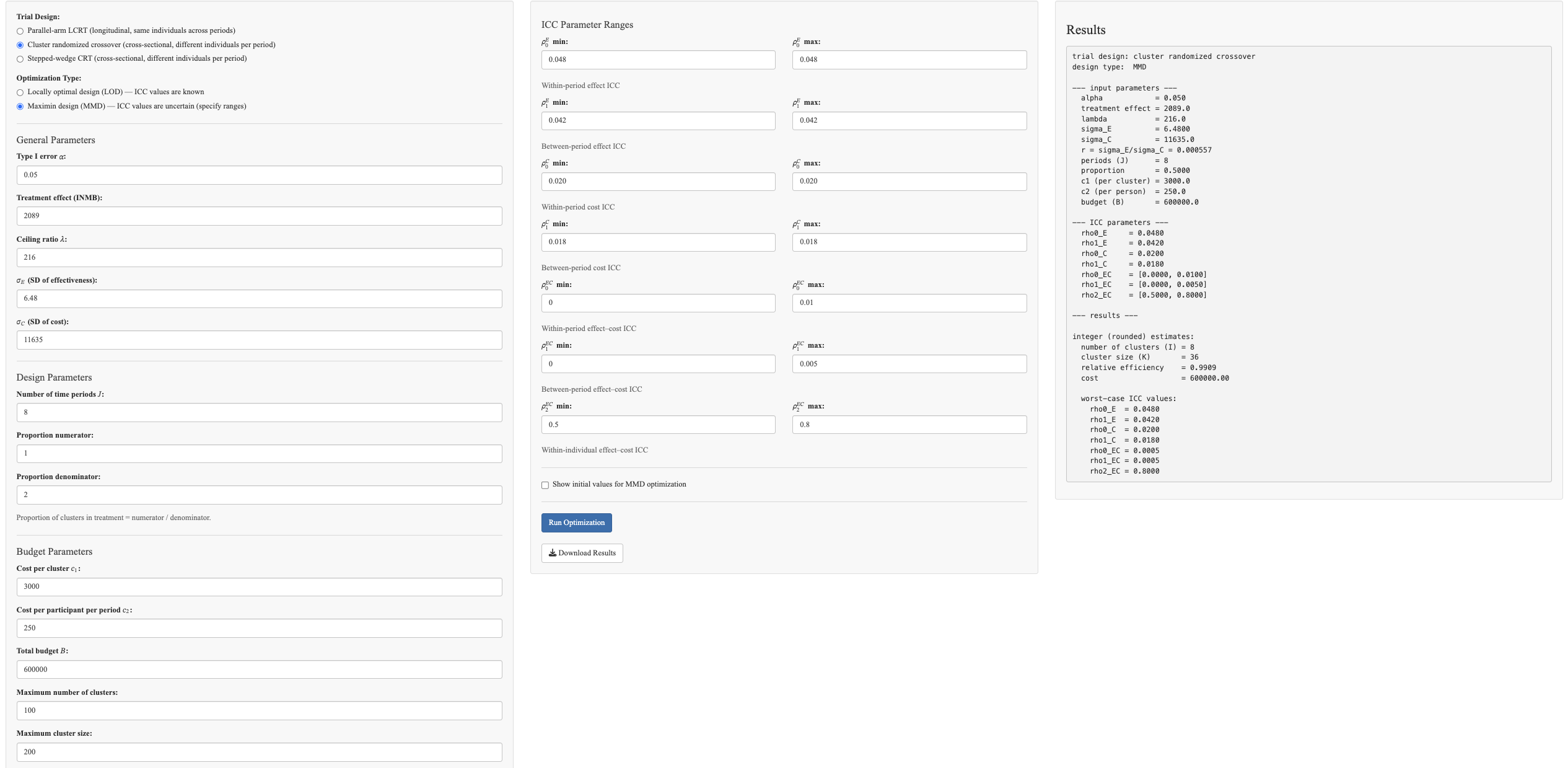}
    \caption{Screenshot of the R Shiny application for computing the MMD for a CRXO trial with $J = 8$ periods, $\pi = 0.5$, $B = \$600{,}000$, $c_1 = \$3{,}000$, $c_2 = \$250$, outcome-specific ICCs fixed at $\rho_0^E = 0.048$, $\rho_1^E = 0.042$, $\rho_0^C = 0.020$, $\rho_1^C = 0.018$, and between-outcome ICC ranges $\rho_0^{EC} \in [0, 0.01]$, $\rho_1^{EC} \in [0, 0.005]$, $\rho_2^{EC} \in [0.5, 0.8]$. The results panel displays the integer MMD allocating $I_{\text{MMD}} = 8$ clusters and $K_{\text{MMD}} = 36$ individuals per cluster-period, with RE $= 0.991$.}
    \label{supp_fig:shiny_mmd}
\end{figure}

\subsection{Illustrative Example: MMD} \label{supp_sec:tut_mmd}

\paragraph{Step 1: Select the trial design and optimization type} The investigator selects ``Cluster randomized crossover'' as the trial design and ``Maximin design (MMD)'' as the optimization type. Upon selecting MMD, the ICC input panel changes to display minimum and maximum fields for each ICC parameter.

\paragraph{Step 2: Enter general, design, and budget parameters} These parameters are identical to the LOD example: $\alpha = 0.05$, $|\beta_1| = 2{,}089$, $\lambda = 216$, $\sigma_E = 6.48$, $\sigma_C = 11{,}635$, $J = 8$, $\pi = 0.5$, $c_1 = \$3{,}000$, $c_2 = \$250$, $B = \$600{,}000$, $I_{\max} = 100$, and $K_{\max} = 200$.

\paragraph{Step 3: Enter ICC parameter ranges} The investigator specifies plausible ranges for each ICC. Because reliable point estimates of the outcome-specific ICCs are available from the bivariate linear mixed model fitted to the Trial 2 data, we fix these at their estimated values by setting equal minimum and maximum: $\rho_0^E = 0.048$, $\rho_1^E = 0.042$, $\rho_0^C = $, and $\rho_1^C = 0.018$. For the between-outcome ICCs, the investigator specifies wider ranges: $\rho_0^{EC} \in [0, 0.01]$, $\rho_1^{EC} \in [0, 0.005]$, and $\rho_2^{EC} \in [0.5, 0.8]$. These ranges define the parameter space $\bm{\Theta} = \{\bm{\rho}: \bm{\rho}_{\min} \preceq \bm{\rho} \preceq \bm{\rho}_{\max}\}$ over which the optimization seeks the design that maximizes the worst-case relative efficiency. Investigators should ensure that all specified ranges satisfy the ordering constraints described in Section~F.2; the application will report an error if any constraint is violated.

\paragraph{Step 4: (Optional) Set initial values} The investigator may check ``Show initial values for MMD optimization'' to provide starting values for the inner optimization that searches for the worst-case ICC configuration. In most cases, the default initial values are adequate; however, if the application reports convergence issues, the investigator may experiment with alternative starting points within the specified ranges.

\paragraph{Step 5: Run and review results} The investigator clicks ``Run Optimization.'' The results panel displays both decimal and integer estimates, along with the RE and the worst-case ICC configuration. Figure~\ref{supp_fig:shiny_mmd} presents a screenshot of the application with these inputs and results. The integer MMD yields $(I_{\text{MMD}}, K_{\text{MMD}}) = (8, 36)$ with RE $= 0.991$. Of note, the MMD coincides with the LOD in this setting, confirming that the optimal sample size combination is robust to between-outcome ICC uncertainty. The results also report the worst-case ICC values at which the minimum efficiency is attained, providing investigators with insight into which parameter configurations are most challenging for the chosen design.

\clearpage
\section{Additional Figures and Tables} \label{supp_sec::figures_tables}

\begin{figure}[ht!]
    \centering
    \includegraphics[width=0.95\linewidth]{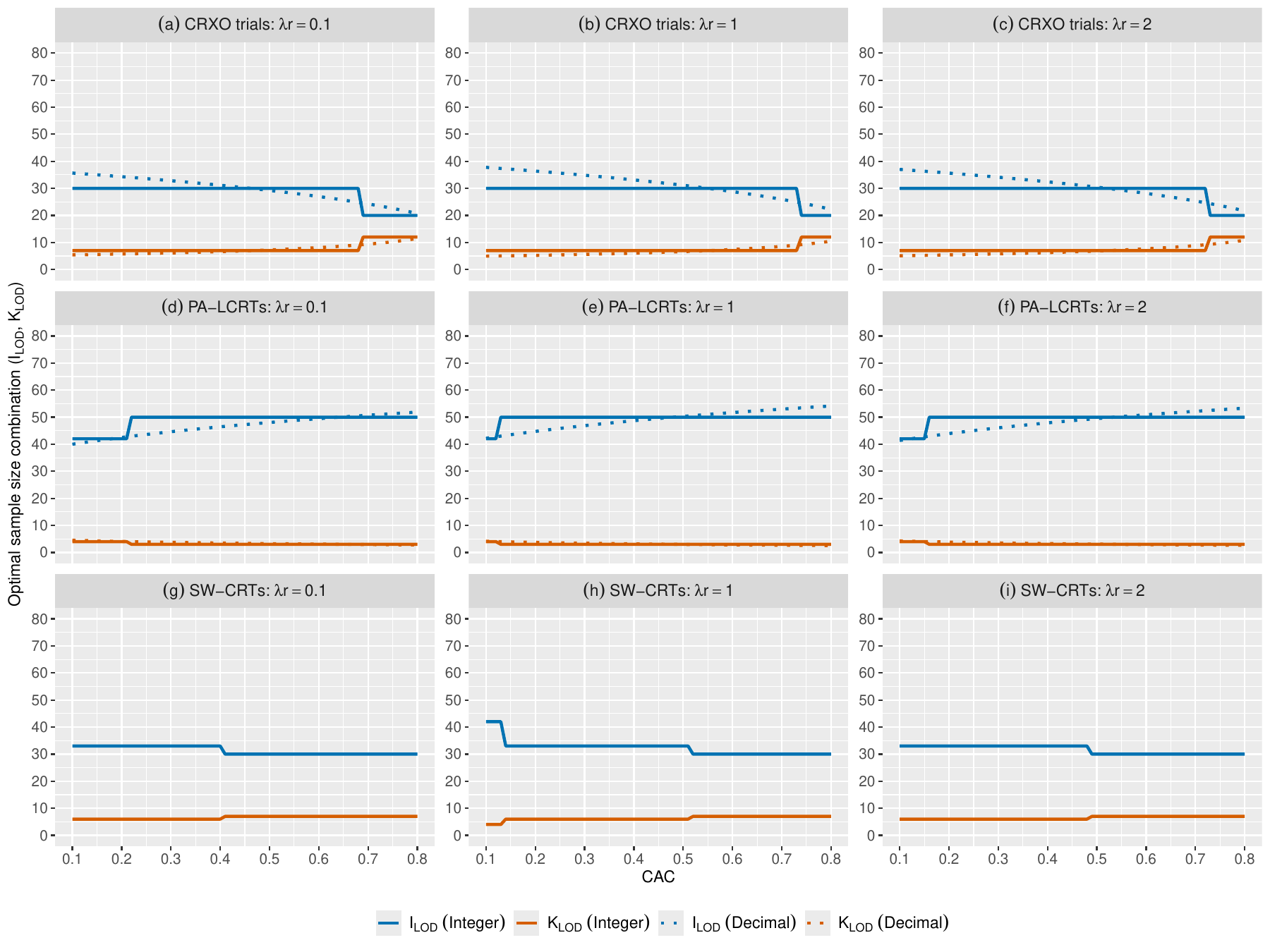}
    \caption{LODs for CRXO trials, PA-LCRTs, and SW-CRTs with $J = 4$ periods under three standardized ceiling ratio scenarios. Panels (a), (b), and (c) present $I_{\text{LOD}}$ and $K_{\text{LOD}}$ as functions of CAC when $\lambda r \in \{0.1, 1, 2\}$, respectively. Both integer (solid lines) and decimal (dotted lines) estimates are shown. We consider $J = 4$ periods, $B = \$300,000$, $c_1 = \$3,000$, and $c_2 = \$250$. For CRXO trials, clusters are equally allocated to two treatment sequences (ICIC/CICI); for PA-LCRTs, clusters are equally allocated to treatment or control for the entire study duration. For ICC parameters, we fix $\rho_0^C = \rho_0^E = 0.1$, $\rho_0^{EC} = 0.04$, $\rho_2^{EC} = 0.5$, with CAC $= \rho_1^E/\rho_0^E = \rho_1^C/\rho_0^C = \rho_1^{EC}/\rho_0^{EC} \in [0.1, 0.8]$. Additional parameters include $\lambda = 20,000$, $\sigma_E = 1$, and $\sigma_C \in \{10,000, 20,000, 200,000\}$ for panels (a), (b), and (c), respectively} \label{supp_fig:crxo_lod_J4_varying_lambda_r}
\end{figure}

\begin{table}
\caption{LODs for cross-sectional, multiple-period CRXO trials: optimal number of clusters $I_{\text{LOD}}$, optimal cluster-period size $K_{\text{LOD}}$, and power of the LOD for detecting the average INMB, assuming known ICC parameters $\bm{\rho} = (\rho^E_0, \rho^E_1, \rho^C_0, \rho^C_1, \rho^{EC}_0, \rho^{EC}_1, \rho^{EC}_2)^\prime$, total budget $B$, cluster-level cost $c_1$, and individual-level cost $c_2$.}\label{supp_tab:crxo_LOD}
{\centering
\resizebox{\linewidth}{!}{\begin{tabular}{lllccccccccc}
    \toprule
    Effect Association & Cost Association & Cost-Effect Association & \multicolumn{3}{c}{$J = 2$} & \multicolumn{3}{c}{$J = 4$} & \multicolumn{3}{c}{$J = 6$} \\
    \cmidrule(lr){4-6} \cmidrule(lr){7-9} \cmidrule(lr){10-12}
    $(\rho_0^E, \rho_1^E)$ & $(\rho_0^C, \rho_1^C)$ & $(\rho_0^{EC}, \rho_1^{EC}, \rho_2^{EC})$ & $K_{\text{LOD}}$ & $I_{\text{LOD}}$ & Power & $K_{\text{LOD}}$ & $I_{\text{LOD}}$ & Power & $K_{\text{LOD}}$ & $I_{\text{LOD}}$ & Power \\
    \midrule
    (0.05, 0.025) & (0.04, 0.020) & (0.02, 0.010, 0.5) & 30 & 14 & 0.774 & 20 & 12 & 0.841 & 20 & 8 & 0.870 \\ 
    (0.05, 0.040) & (0.04, 0.032) & (0.02, 0.016, 0.5) & 20 & 24 & 0.859 & 12 & 22 & 0.894 & 10 & 18 & 0.911 \\ 
    (0.10, 0.050) & (0.08, 0.040) & (0.04, 0.020, 0.5) & 40 & 9 & 0.693 & 30 & 7 & 0.790 & 22 & 7 & 0.827 \\ 
    (0.10, 0.080) & (0.08, 0.064) & (0.04, 0.032, 0.5) & 26 & 17 & 0.813 & 20 & 12 & 0.873 & 20 & 8 & 0.896 \\ 
    (0.20, 0.100) & (0.16, 0.080) & (0.08, 0.040, 0.5) & 46 & 7 & 0.598 & 42 & 4 & 0.715 & 40 & 3 & 0.781 \\ 
    (0.20, 0.160) & (0.16, 0.128) & (0.08, 0.064, 0.5) & 40 & 9 & 0.759 & 30 & 7 & 0.847 & 20 & 8 & 0.879 \\  
    \midrule
    (0.05, 0.025) & (0.06, 0.030) & (0.02, 0.010, 0.5) & 30 & 14 & 0.773 & 20 & 12 & 0.840 & 20 & 8 & 0.870 \\ 
    (0.05, 0.040) & (0.06, 0.048) & (0.02, 0.016, 0.5) & 20 & 24 & 0.858 & 12 & 22 & 0.894 & 10 & 18 & 0.910 \\ 
    (0.10, 0.050) & (0.12, 0.060) & (0.04, 0.020, 0.5) & 40 & 9 & 0.691 & 30 & 7 & 0.789 & 22 & 7 & 0.826 \\ 
    (0.10, 0.080) & (0.12, 0.096) & (0.04, 0.032, 0.5) & 30 & 14 & 0.813 & 20 & 12 & 0.873 & 20 & 8 & 0.896 \\ 
    (0.20, 0.100) & (0.24, 0.120) & (0.08, 0.040, 0.5) & 50 & 6 & 0.596 & 42 & 4 & 0.714 & 40 & 3 & 0.781 \\ 
    (0.20, 0.160) & (0.24, 0.192) & (0.08, 0.064, 0.5) & 40 & 9 & 0.758 & 30 & 7 & 0.846 & 20 & 8 & 0.879 \\   
    \midrule
    (0.05, 0.025) & (0.04, 0.020) & (0.02, 0.010, 0.8) & 30 & 14 & 0.807 & 20 & 12 & 0.870 & 20 & 8 & 0.899 \\ 
    (0.05, 0.040) & (0.04, 0.032) & (0.02, 0.016, 0.8) & 20 & 24 & 0.888 & 20 & 12 & 0.921 & 10 & 18 & 0.934 \\ 
    (0.10, 0.050) & (0.08, 0.040) & (0.04, 0.020, 0.8) & 40 & 9 & 0.726 & 30 & 7 & 0.823 & 22 & 7 & 0.858 \\ 
    (0.10, 0.080) & (0.08, 0.064) & (0.04, 0.032, 0.8) & 30 & 14 & 0.847 & 20 & 12 & 0.902 & 20 & 8 & 0.923 \\ 
    (0.20, 0.100) & (0.16, 0.080) & (0.08, 0.040, 0.8) & 50 & 6 & 0.630 & 42 & 4 & 0.753 & 40 & 3 & 0.819 \\ 
    (0.20, 0.160) & (0.16, 0.128) & (0.08, 0.064, 0.8) & 40 & 9 & 0.796 & 30 & 7 & 0.880 & 22 & 7 & 0.908 \\ 
    \midrule
    (0.05, 0.025) & (0.05, 0.025) & (0.02, 0.010, 0.8) & 30 & 14 & 0.806 & 20 & 12 & 0.870 & 20 & 8 & 0.899 \\ 
    (0.05, 0.040) & (0.05, 0.040) & (0.02, 0.016, 0.8) & 20 & 24 & 0.887 & 20 & 12 & 0.921 & 10 & 18 & 0.934 \\ 
    (0.10, 0.050) & (0.10, 0.050) & (0.04, 0.020, 0.8) & 40 & 9 & 0.726 & 30 & 7 & 0.823 & 22 & 7 & 0.857 \\ 
    (0.10, 0.080) & (0.10, 0.080) & (0.04, 0.032, 0.8) & 30 & 14 & 0.847 & 20 & 12 & 0.901 & 20 & 8 & 0.923 \\ 
    (0.20, 0.100) & (0.20, 0.100) & (0.08, 0.040, 0.8) & 50 & 6 & 0.630 & 42 & 4 & 0.752 & 40 & 3 & 0.819 \\ 
    (0.20, 0.160) & (0.20, 0.160) & (0.08, 0.064, 0.8) & 40 & 9 & 0.796 & 30 & 7 & 0.880 & 22 & 7 & 0.908 \\   
    \midrule
    (0.05, 0.025) & (0.06, 0.030) & (0.02, 0.010, 0.8) & 30 & 14 & 0.806 & 20 & 12 & 0.870 & 20 & 8 & 0.898 \\ 
    (0.05, 0.040) & (0.06, 0.048) & (0.02, 0.016, 0.8) & 20 & 24 & 0.887 & 20 & 12 & 0.921 & 10 & 18 & 0.934 \\ 
    (0.10, 0.050) & (0.12, 0.060) & (0.04, 0.020, 0.8) & 40 & 9 & 0.725 & 30 & 7 & 0.823 & 22 & 7 & 0.857 \\ 
    (0.10, 0.080) & (0.12, 0.096) & (0.04, 0.032, 0.8) & 30 & 14 & 0.846 & 20 & 12 & 0.901 & 20 & 8 & 0.923 \\ 
    (0.20, 0.100) & (0.24, 0.120) & (0.08, 0.040, 0.8) & 50 & 6 & 0.629 & 42 & 4 & 0.752 & 40 & 3 & 0.819 \\ 
    (0.20, 0.160) & (0.24, 0.192) & (0.08, 0.064, 0.8) & 40 & 9 & 0.796 & 30 & 7 & 0.880 & 22 & 7 & 0.908 \\ 
    \bottomrule
\end{tabular}}
\par}

\vspace{0.5em}
\footnotesize
CRXO: cluster randomized crossover; LOD: local optimal design; $\beta_1 = 4000$: incremental net monetary benefit (INMB); $\lambda = 20,000$: ceiling ratio; $r = \sigma_E/\sigma_C = 1/3000$: ratio of effect to cost standard deviations; $\rho_0^E$: within-period effect ICC; $\rho_1^E$: between-period effect ICC; $\rho_0^C$: within-period cost ICC; $\rho_1^C$: between-period cost ICC; $\rho_0^{EC}$: within-period effect-cost ICC; $\rho_1^{EC}$: between-period effect-cost ICC; $\rho_2^{EC}$: within-individual effect-cost ICC; $J$: number of periods; $I_{\max} = 100$: maximum number of clusters; $K_{\max} = 200$: maximum cluster-period size; $B = 300,000$: total budget; $c_1 = 3000$: cost per cluster; $c_2 = 250$: cost per individual. For CRXO trials, clusters are equally allocated to two treatment sequences: IC/CI for $J = 2$, ICIC/CICI for $J = 4$, and ICICIC/CICICI for $J = 6$.
\end{table}
\end{APPENDICES}

\end{document}